\documentclass{article}
\usepackage{mathpazo}
\usepackage[english]{babel}
\usepackage{float}
\usepackage{tikz}
\usetikzlibrary{calc} %
\usepackage{array}
\usepackage{multirow}

\usepackage[letterpaper,top=2cm,bottom=2cm,left=3cm,right=3cm,marginparwidth=1.75cm]{geometry}

\usepackage[noend]{algpseudocode}
\usepackage{amsmath, amssymb}
\usepackage{amsthm}
\usepackage{comment}
\usepackage{mdframed}
\usepackage{xspace}
\usepackage{multicol}
\usepackage{tabularx}
\usepackage{makecell}

\usepackage{enumitem}

\usepackage{thmtools} 
\usepackage{thm-restate}

\newtheorem{theorem}{Theorem}[section]
\newtheorem{lemma}[theorem]{Lemma}

\newtheorem{claim}{Claim}
\newtheorem*{claim*}{Claim}
\newtheorem{definition}{Definition}

\numberwithin{theorem}{section}

\usepackage{thmtools}

\declaretheoremstyle[
    bodyfont=\normalfont  %
]{upright}

\declaretheorem[
    style=upright,  %
    refname={Algorithm,Algorithms},
    Refname={Algorithm,Algorithms},
    name={Algorithm}
]{algorithm}
\usepackage{graphicx}
\usepackage[colorlinks=true, allcolors=blue]{hyperref}

\newcommand{\dist}{d}
\newcommand{\eps}{\varepsilon}
\newcommand{\JMS}{Greedy\xspace}

\newcommand{\Sf}{S_{\mathsf{free}}}
\newcommand{\Sr}{S_{\mathsf{reg}}}
\newcommand{\Ff}{F_{\mathsf{free}}}

\newcommand{\calI}{\mathcal{I}}
\newcommand{\poly}{\mathrm{poly}}

\newcommand{\calB}{\mathcal{B}}
\newcommand{\calC}{\mathcal{C}}

\newcommand{\calN}{\mathcal{N}}
\newcommand{\calQ}{\mathcal{Q}}

\newcommand{\calH}{\mathcal{H}}
\newcommand{\calM}{\mathcal{M}}
\newcommand{\calP}{\mathcal{P}}
\newcommand{\calX}{\mathcal{X}}

\newcommand{\calU}{\mathcal{U}}
\newcommand{\bcalU}{\widetilde{\mathcal{U}}}
\newcommand{\tmu}{\widetilde{\mu}}
\newcommand{\calR}{\mathcal{R}}

\newcommand{\bcalR}{\widetilde{\mathcal{R}}}
\newcommand{\calL}{\mathcal{L}}

\newcommand{\clcost}{\text{cost}}
\newcommand{\closedclcost}{\text{closedcost}}
\newcommand{\increase}{\text{cost-inc}}

\newcommand{\clients}{D}
\newcommand{\clientspure}{D^*_{\texttt{pure}}}
\newcommand{\clientscheap}{D^*_{\texttt{cheap}}}
\newcommand{\clientsexpensive}{D^*_{\texttt{exp}}}
\newcommand{\facilities}{F}
\newcommand{\opt}{\text{OPT}}
\newcommand{\sopt}{\text{opt}}
\newcommand{\optlpfl}{\text{opt}_{\text{LP}}}

\newcommand{\Maxdist}{M}
\newcommand{\dummyset}{\Lambda}
\newcommand{\core}{C^{\text{core}}}

\newcommand{\algstatex}{\ensuremath{(\alpha, S, A, \theta)}\xspace}

\newcommand{\logadaptalg}{\textsc{LogAdaptive}\xspace}

\newcommand{\JMSalg}{\textsc{Greedy}\xspace}
\newcommand{\mergealg}{\textsc{MergeSolutions}\xspace}
\newcommand{\completesol}{\textsc{CompleteSolution}\xspace}
\newcommand{\completesolx}{\textsc{CompleteSolution}}
\newcommand{\completesequence}{\textsc{CompleteSequence}\xspace}
\newcommand{\calLexp}{\calL_{\text{exp}}}

\title{A $(2{+\eps})$-Approximation Algorithm for Metric $k$-Median}
\author{Vincent Cohen-Addad\footnote{Google Research, France.} \and
Fabrizio Grandoni\footnote{IDSIA, USI-SUPSI, Switzerland. Supported in part by the SNF Grants 200021-200731 and 200021-236706.} \and Euiwoong Lee\footnote{University of Michigan, Ann Arbor, MI, USA. Supported in part by NSF 2236669.} \and Chris Schwiegelshohn \footnote{Aarhus University, Denmark. Supported in part by the Independent Research Fund Denmark (DFF) under grant No 1051-00106B and a Google Research Award.} \and Ola Svensson\footnote{EPFL, Switzerland. Supported in part by the Swiss State Secretariat for Education, Research and Innovation (SERI) under contract number MB22.00054.}}

\usepackage[colorinlistoftodos,textsize=tiny,textwidth=2cm,color=green!50!gray]{todonotes}

\newlength\hwhatcwidth%
\newlength\hwhatuwidth%
\newcommand*\halfwidehat[1]{%
  \settowidth\hwhatcwidth{\ensuremath{#1}}%
  \settowidth\hwhatuwidth{\ensuremath{\mathrm{#1}}}%
  \addtolength\hwhatuwidth{-\hwhatcwidth}%
  \makebox[0pt][l]{%
    \kern\dimexpr0.25\hwhatcwidth-0.667\hwhatuwidth\relax%
    \ensuremath{\widehat{\vphantom{#1}\rule{0.5\hwhatcwidth}{0pt}}}%
  }#1%
}%

\date{}
\begin{document}
\maketitle

\begin{abstract}
\noindent In the classical NP-hard (metric) $k$-median problem, we are given a set of $n$ clients and centers with metric distances between them, along with an integer parameter $k \geq 1$. The objective is to select a subset of $k$ \emph{open} centers that minimizes the total distance from each client to its closest open center.

In their seminal work, Jain, Mahdian, Markakis, Saberi, and Vazirani presented the \JMS algorithm for facility location,  which implies a $2$-approximation algorithm for $k$-median that opens $k$ centers in \emph{expectation}. 
Since then, substantial research has aimed at narrowing the gap between their algorithm and the best achievable approximation by an algorithm guaranteed to open \emph{exactly} $k$ centers,
as required in the $k$-median problem.
 During the last decade, all improvements have been achieved by leveraging their algorithm (or a small improvement thereof), followed by a second step called bi-point rounding, which inherently adds an additional factor to the approximation guarantee.

Our main result closes this gap: for any $\eps > 0$, we present a $(2+\eps)$-approximation algorithm for the $k$-median problem, improving the previous best-known approximation factor of $2.613$. Our approach builds on a combination of two key algorithms. First, we present a non-trivial modification of the \JMS algorithm that operates with only $O(\log n{/\eps^2})$ adaptive phases. Through a novel walk-between-solutions approach, this enables us to construct a $(2+\eps)$-approximation algorithm for $k$-median that consistently opens at most $k + O(\log n{/\eps^2})$ centers: via known results, this already implies a $(2+\eps)$-approximation algorithm that runs in quasi-polynomial time. Second, we develop a novel $(2+\eps)$-approximation algorithm tailored for stable instances,  
where removing any center from an optimal solution increases the cost by at least an $\Omega(\eps^3/\log n)$ fraction. Achieving this involves several ideas, including a sampling approach inspired by the $k$-means++ algorithm and a reduction to submodular optimization subject to a partition matroid.
This allows us to convert the previous result into a polynomial time algorithm that opens exactly $k$ centers while maintaining the $(2+\eps)$-approximation guarantee. 
 
\end{abstract}

\thispagestyle{empty}

\newpage

\tableofcontents
\thispagestyle{empty}

\newpage
\setcounter{page}{1}
\section{Introduction}

Given a collection $\clients$ of {$n$} \emph{points}, a clustering is a partition of the points into (typically) disjoint subsets (the \emph{clusters}) such that points in the same cluster are \emph{similar} and points in different clusters are \emph{dissimilar}.
One of the classical and best-studied clustering problems is \emph{$k$-median} (see also Table \ref{tab:results}). Here, in addition to the set $D$ (points are also called \emph{clients} in this framework), we are given an integer parameter $k>0$ and a set $\facilities$\footnote{Sometimes in the literature one assumes $\facilities=\clients$ or $\clients\subseteq \facilities$. We consider the more general case and make no such assumption.} (the \emph{facilities} or \emph{centers}), with metric distances $d(a,b)$ for $a,b\in \clients\cup \facilities$. Our goal is to select $k$ centers $S\subseteq F$ (\emph{open} centers) so as to minimize the sum of the distances from each client to the closest open center. In other words, the goal is to minimize the following objective function
$$
\clcost(S):=\sum_{j\in D} \min_{i\in S}d(i,j).
$$
The $k$-clustering induced by $S$ is obtained by considering, for each $c\in S$, the subset of points that have $c$ as a closest center in $S$ (breaking ties arbitrarily). 
We use $\sopt_k$ to denote the optimal cost, and simply use $\sopt$ when $k$ is clear from the context.

For most clustering variants, it is NP-hard to compute an optimal clustering and metric $k$-median is no exception. For this reason a lot of effort was devoted to the design of (polynomial-time) approximation algorithms for this problem, i.e., algorithms that compute a feasible solution $S$ whose cost is provably within some factor $\alpha>1$ of $\sopt$ (the \emph{approximation factor} or \emph{approximation ratio/guarantee}). 

Most of the best-known approximation algorithms for $k$-median, including the current-best one \cite{Cohen-AddadLS23,GowdaPST23}, are based on the following high-level approach (see Section \ref{sec:relatedWork} for alternative approaches in the literature). One considers the associated facility location problem, where instead of $k$ we are given a (uniform) facility opening cost $f$. In this problem a feasible solution $S$ can contain an arbitrary number of centers/facilities (rather than $k$), however in the objective function, besides the \emph{connection cost} $cost(S)$, one has to pay also the \emph{opening cost} $f|S|$. Then one uses a Lagrangian Multiplier Preserving (LMP) $\alpha_{LMP}$-approximation for the problem (formal definition in Section \ref{sec:prelim}) to construct a \emph{bi-point} solution for $k$-median, i.e., a convex combination of a solution $S_1$ (opening fewer than $k$ centers) and a solution $S_2$ (opening more than $k$ centers), that opens exactly $k$ centers in a fractional sense. In particular, sampling one of the two solutions according to the respective coefficients would give a solution of expected cost $\alpha_{LMP}\cdot \sopt$ that opens $k$ centers \emph{in expectation}. Finally, a bi-point rounding algorithm is used to derive a feasible $k$-median solution, while losing another factor $\alpha_{BPR}$ in the approximation. The product $\alpha_{LMP}\cdot \alpha_{BPR}$ gives the overall approximation factor. We remark that the state of the art leaves a significant gap between the approximation factor for a solution with $k$ centers \emph{in expectation} and a true $k$-median solution.

\begin{table}[h]
\vspace{1em}
    \centering
    \begin{tabularx}{430pt}{|c|c|>{\centering\arraybackslash}X|}
    \hline
     \textbf{Reference} & \textbf{Approximation factor} & \textbf{Technique}  \\
         \hline
         \hline
           \cite{Bartal96} & $O(\log n \log\log n)$ & tree embeddings 
           \\\hline
           \cite{CharikarCGG98} & $O(\log k \log\log k)$ & tree embeddings 
           \\
           \hline
           \cite{CharikarGTS99} & 6.667 & dependent LP rounding 
           \\\hline
           \cite{JaiV01} & 6 & 
           LMP + bi-point rounding 
           \\\hline
           \cite{JMS02,JainMMSV03} & 4 & LMP + bi-point rounding
           \\\hline
           \cite{AryaGKMMP04} & $3+\varepsilon$ & 
           local search 
           \\\hline
           \cite{LiS16} & 2.733 & 
           LMP + bi-point rounding
           \\\hline
           \cite{ByrkaPRST17} & 2.675 & 
           LMP + bi-point rounding
           \\\hline
           \cite{Cohen-AddadLS23,GowdaPST23} & 2.613 & LMP + bi-point rounding
           \\\hline
           \textbf{Theorem \ref{thm:main}} & $2+\varepsilon$ & LMP interpolation + stable algorithm
           \\
           \hline
    \end{tabularx}
    \caption{Progression of approximation factors for $k$-median. Dependencies on $\varepsilon$ are omitted due to numerical overestimation, except for \cite{AryaGKMMP04} and the final result.
    }\label{tab:results}
\end{table}

In more detail, for the LMP stage, \cite{JaiV01} gave a $3$-approximation algorithm which was subsequently improved to $2$ by \cite{JMS02,JainMMSV03}, and then to $2-\eta$ for some $\eta > 2.25\cdot 10^{-7}$ by \cite{Cohen-AddadLS23}.
\cite{JaiV01} also presented a factor $2$ algorithm for the bi-point rounding stage. $\cite{LiS16}$ subsequently improved this to a factor $\frac{1+\sqrt{3}}{2} \approx 1.366$ when allowing to use $k+O(1)$ many centers\footnote{Solutions opening slightly more than $k$ centers are sometimes called \emph{pseudo-solutions} in the literature.}, and also gave a polynomial-time algorithm for obtaining a proper $k$-median clustering from a $k+O(1)$ clustering essentially without any loss in the approximation ratio. 
This approach was further refined by \cite{ByrkaPRST17} who gave a $1.337$-approximate bi-point rounding algorithm and also showed that no LP-based rounding algorithm can do better than $\frac{1+\sqrt{2}}{2}$.
The state of the art for rounding methods is due to \cite{GowdaPST23} who gave a $1.306$-approximation algorithm, while showing that no algorithm can do better than $\sqrt{\phi} \approx 1.272$, where $\phi = \frac{1+\sqrt{5}}{2}$ is the golden ratio. 

Unfortunately, there is an inherent limitation to this approach. No approximation algorithm for the LMP stage can do better than $1+2/e$ \cite{JMS02,JainMMSV03}, and using the lower bound for bi-point rounding algorithms by \cite{GowdaPST23}, the best possible analysis cannot yield an approximation factor better than $(1+2/e)\sqrt{\phi}\approx 2.207$.
This leaves a significant gap to the inapproximability lower bound of $(1+2/e)\approx 1.735$ \cite{JMS02,JainMMSV03}, and even to the best-known approximation algorithm that opens $k$-centers in expectation~\cite{Cohen-AddadLS23}, which achieves an approximation guarantee of $2-\eta$ for some small $\eta>0$.

\subsection{Our Results}
Our main result is stated as follows. (As commonly done, we use "with high probability" to mean that an event occurs with probability at least $1 - 1/\text{poly}(n)$.)
\begin{theorem}
\label{thm:main}
   For every $\eps>0$, there is a {randomized} polynomial-time algorithm {for $k$-median} that returns a solution with cost at most $(2+\eps)\sopt$ {with high probability}.
\end{theorem}

The algorithm in our main theorem combines two novel algorithms. First, we address the limitations of the bi-point rounding method by introducing a non-trivial variant of the classic \JMS algorithm~\cite{JMS02, JainMMSV03}. The \JMS algorithm can be highly {\em adaptive} in that tiny changes to the opening cost $f$ may lead to vastly different solutions $S$ and $S'$, where $|S| < k$ and $|S'| > k$, forcing previous bi-point rounding approaches to incur a large additional factor in the approximation guarantee. We resolve this high-adaptivity by presenting a low-adaptivity version, which makes adaptive decisions in only $O(\log(n)/\eps^2)$ phases. This modification intuitively renders the algorithm less adaptive than the original, where each choice depends heavily on previous ones. We then introduce a novel walk-between-solutions algorithm that enables us to return a solution with at most $k + O(\log(n)/\eps^2)$ centers, with no additional loss in the approximation factor (up to $1 + \epsilon$). Specifically, we achieve the following result.

\begin{restatable}[]{theorem}{pseudoapprox}
   For every constant $\eps \in (0, 1/6)$, there is a polynomial-time algorithm {for $k$-median} that returns a solution {containing} at most $k+ O(\log n/\eps^2)$ many centers and {of} cost at most $(2+\eps)\sopt$.
    \label{thm:pseudo-approx}
\end{restatable}
Theorem~\ref{thm:pseudo-approx} is proved in Section~\ref{sec:log_adaptivity}, and an overview of the main ideas is given in Section~\ref{sec:pseudo-overview}. %

Combining this result with the \cite{LiS16} reduction from $k$-median {pseudo-solutions} to true $k$-median solutions already yields a $(2+\varepsilon)$-approximation in quasi-polynomial time.
In order to improve this to a polynomial time algorithm, we use a completely different algorithm. We say that an instance is $\beta$-stable if {$\sopt_{k-1} \geq (1+\beta)\cdot \sopt_{k}$}. Our result is as follows.
\begin{restatable}[]{theorem}{stableapprox}
For any two constants $\eps, \zeta>0$, there exists a randomized polynomial-time algorithm that, given a $(\zeta/\log n)$-stable $k$-median instance,  returns a solution of cost at most $(2+\eps)\sopt$ with high probability. 
    \label{thm:mainadditivecenters}
\end{restatable}
The algorithm for stable instances is inspired by a recent fixed-parameter-tractable approach that essentially identifies, by enumeration, small-radius balls guaranteed to contain the optimal centers and then solves the problem by maximizing a submodular function subject to a partition matroid constraint. However, there are two major challenges: first, the running time of~\cite{Cohen-AddadG0LL19} is not polynomial even when $k = O(\log n)$; and second, because $k$ may be large (potentially polynomial in $n$), we need to identify almost all clusters of the optimal solution before proceeding with any guessing/enumeration steps. We address these challenges through several novel ideas, including a sampling approach based on the $k$-means++ algorithm.  The proof of Theorem~\ref{thm:mainadditivecenters} is presented in Section~\ref{sec:centerremoval}, with an overview of the techniques provided in Section~\ref{sec:removal-overview}.

Let us see why Theorems \ref{thm:pseudo-approx} and \ref{thm:mainadditivecenters} combine to yield Theorem \ref{thm:main}. We compute a few feasible solutions. Let $\Delta$ be the surplus number of centers used in Theorem \ref{thm:pseudo-approx}. Let $k'=\max\{1,k-\Delta\}$. One solution is obtained by computing the optimum solution with one center for $k'=1$, and otherwise by running the algorithm from Theorem \ref{thm:pseudo-approx} with the number of centers being $k' = k-\Delta$. 
If $\sopt_{k'} \leq \sopt_k\cdot (1+\varepsilon)$, this is a good enough approximation. Then we run the algorithm from Theorem \ref{thm:mainadditivecenters} for every $k'' \in (k',k]$. If the solution with $k'$ centers computed before is not good enough, then there exists some $k'' \in (k',k]$ such that $\sopt_{k''-1} \geq (1+\beta)\cdot \sopt_{k''}$ with $\beta \in \Omega(\eps/\Delta) = \Omega(\varepsilon^3/\log n)$ and $\sopt_{k''}\leq (1+\varepsilon)\cdot \sopt_{k}$.
In this case, Theorem \ref{thm:mainadditivecenters} yields the desired result for a proper choice of the parameters. Outputting the best of all computed solutions yields Theorem \ref{thm:main}.

For Euclidean spaces, the set of candidate centers is infinite. By a standard reduction, we can find a polynomial-sized set $F\subset \mathbb{R}^d$ in polynomial time such that the cost of any constant factor approximate solution is approximated by a set of centers obtained from $F$, see for example Lemma 5.3 of \cite{KumarSS10}. Thus we also obtain a $(2+\varepsilon)$-approximation in Euclidean spaces, improving on $2.406$ \cite{Cohen-AddadEMN22}.

 Finally, let us remark that Theorem~\ref{thm:main} essentially closes the gap between the approximation guarantee of the best-known algorithm that opens $k$ centers in expectation~\cite{Cohen-AddadLS23} and algorithms that consistently open at most $k$ centers. However, a major open problem remains: obtaining a tight approximation result for $k$-median, as it is only known that it is hard to approximate the problem within a factor less than $(1 + 2/e)$~\cite{JMS02}. Building on this work, a natural next step would be to resolve this question or to significantly improve the approximation guarantee, even in the setting where $k$ centers are opened in expectation.

\subsection{Related Work}
\label{sec:relatedWork}

We already reviewed the LMP-based approximation algorithms earlier. But there are several other important results for $k$-median and closely related problems which we cover here. 

There also exist other approaches for designing approximation algorithms for metric $k$-median, most notably local search \cite{AryaGKMMP04,Cohen-AddadGHOS22,GT08,KanungoMNPSW04}, combinatorial algorithms \cite{MeP03}, metric embeddings \cite{Bartal98,Bartal96,CharikarCGG98}, or other LP-based rounding approaches \cite{ArcherRS03,CGTS02,CharikarLi12}.
There has also been a lot of work for the related problems $k$-means and $k$-center, which minimize the sum of squared distances and the maximum distance, respectively. Specifically for $k$-center, \cite{Gon85,HoS86} gave a $2$-approximation which is optimal unless $P\neq NP$.
$k$-means admits a $(9+\varepsilon)$-approximation in general metrics \cite{AhmadianNSW20}. In addition, both $k$-median and $k$-means have been studied in $d$-dimensional Euclidean spaces, where a PTAS is possible when either $d$ or $k$ are fixed \cite{ARR98,KumarSS10,FriggstadRS19,CFS19}. If neither $d$ nor $k$ are fixed, the approximation factors for Euclidean $k$-median and $k$-means are also better than their general metric counterparts \cite{AhmadianNSW20,GrandoniORSV22}, with the best known approximation ratio being $5.913$ for $k$-means and, prior to our work, being $2.406$ for $k$-median by \cite{Cohen-AddadEMN22}.
Both in high dimensional Euclidean spaces as well as general metrics, all of these aforementioned problems are APX hard \cite{AwasthiCKS15,CKL21,Cohen-AddadS19,GI03,JMS02,LeeSW17}.

LMP-based approaches are also very relevant for the related facility location problem.
Mahdian et al.~\cite{MahdianYZ06} showed that an LMP $(1+2/e)$-approximation yields a $(\gamma_{GK} \approx 1.463)$-approximation for UFL, which matches the hardness result of Guha and Khuller~\cite{GuK99} where $\gamma_{GK}$ is the solution of $\gamma = 1 + 2e^{-\gamma}$. 
Nevertheless, most improvements do not rely on an LMP approach. 
Following a long line of  research~\cite{ByrkaA10,CharikarG05,ChudakS03,JainMMSV03,JMS02,
KPR00,MahdianYZ06,ShmoysTA97}, Li presented an algorithm with the currently best known approximation ratio of roughly $1.488$~\cite{Li13}.

\subsection{Preliminaries}
\label{sec:prelim}
We are given a set of $n$ clients $D$ and set of $m$ potential centers or facilities $F$. Given a set $S\subseteq F$, we use $d(j,S)=\min_{i\in S} d(j,i)$, where $d$ is the underlying distance function.
$\tilde{O}(x)$ suppresses $\text{polylog}(x)$ factors, and $O_\eps(x)$ suppresses polynomial factors in $\eps$.
We use $[a]^+ = \max(a,0)$.
We already defined the $k$-median problem. The uncapacitated facility location problem $(UFL)$ with uniform opening costs is given a non-negative number $f$ and asks to find a set $S\subseteq F$ minimizing
$$\text{cost}_{UFL}(S)=\sum_{j\in D}d(j,S) + f\cdot |S|.$$
We refer to $\sum_{j\in D} d(j,S)$ as the \emph{connection cost} and $f\cdot |S|$ as the \emph{opening cost}.
Let $x_{i,j}$ denote in the indicator variable for serving client $j$ with facility $i$ and let $y_i$ be the indicator variable for opening facility $i$. Relaxing these to continuous variables in $[0,1]$, the LP relaxation to $k$-median and UFL are given by the following equations:
\begin{multicols}{2}
\begin{align*}
\min &  \sum_{i\in F,j\in D}d(i,j)x_{i,j} & (LP_{km})\\
s.t. &  \sum_{i\in F}x_{i,j}\geq 1 & \forall j\in D\\
& y_i - x_{i,j}\geq 0 & \forall j\in D,i\in F\\
&  \sum_{i\in F} y_i \leq k  \\
& x, y \geq 0
\end{align*}
\columnbreak
\begin{align*}
\\ &\nonumber \\
\min & \sum_{i\in F,j\in D}d(i,j)x_{i,j} + f\cdot \sum_{i\in F}y_i & (LP_{UFL}(f))\\
s.t. & \sum_{i\in F}x_{i,j}\geq 1  & \forall j\in D \\
&  y_i - x_{i,j}\geq 0 & \forall j\in D,i\in F \\
& x, y \geq 0
 \end{align*}
\end{multicols}
The constraint $\sum_{i\in F}x_{i,j}\geq 1$ ensures that every client is served by some facility and the constraint $y_i-x_{i,j}\geq 0$ ensures that a client $j$ can only be served by facility $i$ if $i$ is open.
Moreover, the dual for the facility location LP relaxation is given by the equations:
\begin{align*}
\max & \sum_{j\in D}\alpha_j & (DP_{UFL}(f))\\
s.t. & \sum_{j\in D} [\alpha_j - d(i,j)]^+ \leq f & \forall i\in F \\
& \alpha \geq 0
 \end{align*}
The usual interpretation of the $\alpha_j$ is that they are the budgets for each client, accounting for both connecting the client to a facility, as well as opening a facility.
The opening contribution dual variables $\beta_{i,j}$ obtained from the $y_i-x_{i,j}\geq 0$ constraints, as well as the dual constraints $\alpha_j - \beta_{i,j}\leq d(i,j)$ are removed here, as we can always set $\beta_{i,j} = [\alpha_j-d(i,j)]^+$.
Recall that $\sopt=\sopt_k$ is the cost of an optimal $k$-median clustering.
We use $\sopt_{LP}$ and $\sopt_{LP}(f)$ to denote the cost of an optimal solution to $LP_{km}$ and  $LP_{UFL}(f)$, resp. Trivially $\sopt\geq\sopt_{LP}$. Since any feasible solution for $LP_{km}$ is also feasible for $LP_{UFL}(f)$, one has $\sopt_{LP}\geq \sopt_{LP}(f)-kf$. We will use the fact that, by weak duality, any feasible solution to $DP_{UFL}(f)$ induces a lower bound on $\sopt_{LP}(f)$. In particular, for any feasible solution $\alpha$ to $DP_{UFL}(f)$, we have $\sopt \geq \sopt_{LP}(f) - kf\geq  \sum_{j\in \clients} \alpha_j - kf$.

We say that an algorithm is a Lagrangian Multiplier Preserving (LMP) $\alpha$-approximation algorithm for UFL if it returns a solution $S$ with
$$\text{cost}(S)\leq \alpha\cdot (\sopt_{LP}(f)-f|S|).$$

Finally, we will make the assumption that the distances are integers in $[1,n^3/\eps]$. The justification is given by the following standard reduction.
\begin{lemma}
\label{lem:aspectratio}
For any constants $\eps > 0$ and $\alpha>1$, given a polynomial-time $\alpha$-approximation algorithm for $k$-median on instances with distances in $\{1,\ldots,n^3/\eps\}$, there exists a polynomial-time $\alpha(1+O(\eps))$-approximation algorithm for $k$-median on general instances. 
\end{lemma}
\begin{proof}
Consider any input instance $(\clients\cup \facilities,d)$ of $k$-median. Assume that $n:=|\clients|$ is large enough w.r.t. $\alpha$, otherwise the problem can be solved by brute force in polynomial time.  
We guess $M:=\max_{j\in \clients}d(j,\opt)$, where $\opt$ is some optimal $k$-median solution, by trying all the $n\cdot m$ possibilities. Next, for each $(i,j)\in \facilities\times \clients$, define  $d'(j,i)=d'(i,j)=\max\{1,\lceil\frac{d(i,j)}{M}\frac{n}{\eps}\rceil\}$ if $d(i,j)\leq M$. Set all the remaining 
$d'(a,b)$, $a\neq b$, to $\frac{n^3}{\eps}$. Finally, replace the $d'(i,j)$'s with the corresponding metric closure. We run the given $\alpha$-approximation algorithm on the $k$-median instance $(D\cup F,d')$, hence obtaining a solution $S$. It is easy to check that $S$ is a good enough approximation for the input instance.  %
\end{proof}

\section{Technical Overview}

In this section, we give an overview of our result. 
Section~\ref{sec:pseudo-overview} introduces ideas behind Theorem~\ref{thm:pseudo-approx} that achieves a $(2+\eps)$-approximation while opening $k + O(\log n / \eps^2)$ centers, and Section~\ref{sec:removal-overview} discusses Theorem~\ref{thm:mainadditivecenters} for stable instances. 

\subsection{$(2+\eps)$-Approximation with 
$O_{\eps}(\log n)$ Extra Centers}
\label{sec:pseudo-overview}

Let us motivate our high-level plan by revisiting the previous framework. 
The previous best approximation algorithms~\cite{lis13, ByrkaPRST17, GowdaPST23, Cohen-AddadLS23} for $k$-median start by the following initialization step: starting from two LMP approximate solutions for UFL $S, S' \subseteq F$ with the respective facility costs $f, f'$ such that $|S| < k, |S'| > k$\footnote{Throughout this overview, we assume no solution opens exactly $k$ centers, because we are done if this happens.} and $f$ possibly much larger than $f'$, the algorithm uses binary search to obtain two new solutions $(S, f), (S', f')$ such that $|S| < k, |S'| > k$ and $|f - f'| \leq 2^{-n}$. This step can be already viewed as a {\em merging} step that takes two solutions whose number of facilities {\em sandwich} $k$ and makes them (specifically the facility cost) closer. Then, the natural question is {\em can we merge $S$ and $S'$ into the (almost) same set of open facilities so that both open (almost) exactly $k$ facilities?}

One naive implementation of this idea might be the following (from now on, let us assume $f = f'$). Just like binary search, suppose that we find an LMP approximate solution $S''$ which is {\em strictly between} $S$ and $S'$ in some formal sense. If $|S''| > k$, we will let $S' \leftarrow S''$ and otherwise $S \leftarrow S''$. 
Then the invariant $|S| < k$, $|S'| > k$ is maintained while the two solutions got {\em strictly closer}. So we can continue this process until they become the same!

Of course, the main question is whether one can find such $S''$ given $S$ and $S'$. Just like $S$ and $S'$, $S''$ must be an outcome of the LMP approximation algorithm, or at least close to such an outcome in some formal sense. Since {\em being a valid outcome of an algorithm} is a highly complex property, especially when the algorithm has a large number of {\em adaptive} steps whose action depends on the previous actions, this task remains challenging for a general LMP approximation algorithm. 

However, there is an example of success for 
Jain and Vazirani's primal-dual~\cite{JaiV01} algorithm, which gives a LMP $3$-approximation. At a high level, this algorithm consists of two phases, where the first phase yields a set $T$ of tentative open facilities, and the second phase is allowed to choose an arbitrary $S \subseteq T$ that forms a maximal independent set in an auxiliary graph, so there is a freedom in choosing the independent set. 
Ahmadian et al.~\cite{AhmadianNSW20} used this freedom to give a walking procedure that finds a maximal independent set in the auxiliary graph of cardinality exactly $k$, which led to an improved approximation ratio for Euclidean $k$-means/$k$-median and metric $k$-means.

However, the \JMS algorithm~\cite{JMS02, JainMMSV03} with dual-fitting analysis, which guarantees a better $2$-LMP approximation, is more adaptive; in fact, it can be a completely deterministic algorithm without leaving any freedom. Therefore, in order to implement our idea, it would be better to find a low-adaptive version of the \JMS algorithm, where the algorithm runs in a small ($O_{\eps}(\log n)$ in our case) number of {\em phases}, and each phase presents some freedom of choices for the algorithm.

\paragraph{The \JMS Algorithm.}
Let us recall the \JMS algorithm~\cite{JMS02, JainMMSV03}. 
The algorithm maintains the set $A$ of \emph{active} clients and the set $S$ of open facilities. Initially, all clients are active, i.e., $A = D$, and no facilities are opened, i.e., $S = \emptyset$. Additionally, the algorithm maintains dual variables $(\alpha_j)_{j \in D}$, initialized to $0$.

While there is at least one active client, the algorithm uniformly increases the $\alpha$-values of the active clients until there is a facility $i$ whose "bids" are sufficient to open it (or an active client becomes ``close'' to an already opened facility). The bids from a client $j$ to a facility $i$ depend on whether $j$ is active or not. If $j \in A$, then a portion of its budget $\alpha_j$, specifically $d(i,j)$, is allocated to cover the connection cost, and the remainder, $\alpha_j - d(i,j)$, is used to bid towards the facility's opening cost $\hat f := 2f$. For a client $j \notin A$, which is already connected to an open facility at a distance $d(j,S)$, the bid towards $i$ is the difference in distance, i.e., $d(j,S) - d(i,j)$ (or $0$ if $i$ is farther).
Finally, the algorithm marks any client $j$ with $\alpha_j \geq d(j,S)$ as inactive, ensuring that at any point in time, active clients $j$ satisfy $\alpha_j \leq d(j,S)$.

In order to prove that the algorithm yields an LMP $2$-approximation, let $(\alpha^*_j)_{j\in D}$ be the final $\alpha$ values and $S^*$ be the set of opened facilities. 
The design of the algorithm guarantees that the $\alpha^*$ values pay for both the opening costs of the facilities in $S^*$ and the connection costs of the clients, leading to $\sum_{j\in D} \alpha^*_j = \sum_{j\in D} d(j,S^*) + \hat f \cdot |S^*|$. 

The more interesting part of the analysis is to show that $\alpha^*/2$ is a dual feasible solution for the dual facility location LP with facility cost $f$, namely 
$\sum_{j \in D}[ \alpha^*_j - 2d(i,j)]^+ \leq \hat f$ for every facility $i$. 
The dual feasibility follows from two key properties that, at any point in time, (a) $\alpha_j \leq d(j, S)$ for every $j \in A$ and (b) for every $i \in F$, no facility is overbid: 
    \begin{equation}
     \sum_{j\in A} [\alpha_j - d(i,j)]^+ + \sum_{j\in D \setminus A} [d(j,S) - d(j,i)]^+ \leq \hat f\,.
     \label{eq:no-overbid-overview}
    \end{equation}

\subsubsection{Logarithmic Adaptivity}
\label{sec:log_overview}

As stated, the \JMSalg is a highly adaptive algorithm; for instance, 
consider the case when two facilities $i$ and $i'$ both become paid for at the same time $\theta$. Whether we open $i$ or $i'$ can dramatically change the bids to future facilities and completely change the trace of the algorithm. 

To reduce such high adaptivity, we limit the number of time steps where such decisions are made. Specifically, our high-level goal is to only consider the time steps of the form $(1+\eps^2)^0, (1+\eps^2)^1, (1+\eps^2)^2, \ldots$. Then, as the minimum nonzero distance is $1$ and the maximum distance is $\poly(n)$, the total number of time steps that we need to consider is $O(\log n /\eps^2)$, and this is what we refer to as $O(\log n)$ adaptivity.

The main difficulty in implementing this idea is to ensure~\eqref{eq:no-overbid-overview}. Consider the situation where all active clients have $\alpha$ value $\theta = (1+\eps^2)^t$. If their $\alpha$ values jump from $\theta$ to $(1+\eps^2)\theta$, it is possible that a previously underbid facility becomes strictly overbid. 

We fix this issue in two steps. The first easy fix is to {\em scale down the metric} by a factor of $(1-\delta)$ where $\delta = O(\eps)$; that is, we consider a facility $i$ paid for if
\begin{equation}
     \sum_{j\in A} [\alpha_j - (1-\delta) d(i,j)]^+ + (1-\delta) \sum_{j\in D \setminus A} [d(j,S) - d(j,i)]^+  \geq \hat f\,.
     \label{eq:paid-overview}
\end{equation}
At the cost of losing a factor $1/(1-\delta)$ in the approximation guarantee, this has the following benefit. If $i$ is not paid according to~\eqref{eq:paid-overview} when all active clients have $\alpha_j = \theta$, then even when $\alpha_j$'s are increased to $(1+\eps^2)\theta$, for a client $j$ with $d(j, i) \geq \eps \theta$, its contribution $[(1+\eps^2)^{t+1} - d(i,j)]^+$ to~\eqref{eq:no-overbid-overview} is at most its previous contribution 
$[(1+\eps^2)^{t} - (1-\delta) d(i,j)]^+$ to~\eqref{eq:paid-overview}, which makes~\eqref{eq:no-overbid-overview} unlikely to be violated. 

Then our second fix deals with active clients close to co-located with $i$. We do so by allowing the clients $j$ in the ball $B(i, \eps \theta)$ to increase their bids for $i$ up to $\min((1+\eps^2)\theta, (1-\delta)d(j,S))$, so that if $i$ could not be opened even with increased bids from nearby clients, then increasing active clients' $\alpha$ values from $\theta$ to $(1+\eps^2)\theta$ should not make $i$ overbid. These increased bids bring other considerations (e.g., ensuring dual feasibility of other facilities), leading to the following definition of {\em openable} facilities, which can be efficiently checked by a linear program.

\begin{definition}
    A facility $i$ is \emph{openable} with respect to the algorithm's state \algstatex if there are bids $(\tau_j)_{j\in A}$ of active clients that satisfy the following conditions.
    \begin{itemize}
        \item Only nearby clients are increased:
        \begin{gather*}
            \tau_j = \alpha_j \qquad \mbox{for every $j \in A - B(i, \eps \theta)$.}
        \end{gather*}  
        \item Nearby clients are only increased slightly:  
        \begin{gather*}
            \alpha_j \leq \tau_j \leq \min\{(1-\delta) d(j,S), (1+\eps^2)\theta\} \qquad \mbox{for every $j\in A \cap B(i, \eps \theta)$.}
        \end{gather*}
        \item Facility $i$ is paid for: 
        \begin{gather*}
            \sum_{j\in A} [\tau_j - (1-\delta)d(i,j)]^+ + (1-\delta) \sum_{j\in D \setminus A} [d(j,S) - d(i,j)]^+ \geq \hat f\,.
        \end{gather*}
        \item Dual feasibility: for every facility $i_0$ and $\ell \in A$,
        \begin{gather*}
            \sum_{j \in A} [\tau_j - d(i_0, j)]^+ + \sum_{j \in D \setminus A} [\tau_\ell - 2d(j, i_0) - d(\ell, i_0)]^+ \leq \hat f.
        \end{gather*}
    \end{itemize}
\label{def:openable_overview}
\end{definition}

\begin{figure*}[h!]
\begin{center}
\begin{minipage}{1.0\textwidth}
\begin{mdframed}[hidealllines=true, backgroundcolor=gray!15]

\begin{algorithm}[\logadaptalg] \ \\[0.2cm]
\textbf{{Initialization:}} Set $\hat f = 2f, S = \emptyset, A=D, \theta = 1$ and $\alpha_j = \theta$ for every $j\in A=D$. \\[0.1cm]

{Repeat the following while $A\neq \emptyset$:}

\begin{enumerate}
\item While there is an unopened facility $i$ that is openable (chosen arbitrarily if there are many),
\begin{quote} Calculate $(\tau_j)_{j\in A}$ so that $i$ is openable with bids $(\tau_j)_{j\in A}$. Open $i$ (i.e., add it to $S$), 
let $\alpha_j \leftarrow \tau_j$ for every $j \in A$ (note that only the values of the clients in $A \cap B(i, \eps \theta)$ change), mark all clients $j\in A$ such that $\alpha_j \ge (1-\delta) d(j,S)$ as served (i.e., remove them from $A$).
\end{quote}

\item Perform the following step and proceed to the next phase. 
\begin{quote}
For every remaining $j \in A$, let $\alpha_j \leftarrow \min((1 + \eps^2)\theta, (1-\delta)d(j, S))$, remove it from $A$ when $\alpha_j = (1 - \delta)d(j, S)$.  Update $\theta = (1+\eps^2)\theta$.
\end{quote}
\end{enumerate}
\label{alg:log_adapt_alg_overview}
\end{algorithm}
\end{mdframed}
\end{minipage}
\end{center}
\end{figure*}

Given this definition of openability, 
Section~\ref{subsec:log_adaptive} presents our modified \JMS algorithm with $O(\log n)$ phases, \logadaptalg, which is restated here as Algorithm~\ref{alg:log_adapt_alg_overview}. 
Here a {\em phase} is defined to be the run of the algorithm for a fixed value of $\theta = 1, (1+\eps^2), (1+\eps^2)^2, \dots$, which is broken down to two {\em stages} according to the description of  Algorithm~\ref{alg:log_adapt_alg_overview}. Note that at least in the beginning, stage 1 allows opening of any openable facility, providing the desired freedom to walk between two solutions within the phase. While stage 1 still presents some adaptivity as an opening of a facility might impact the openability of another, it is sufficient for the merging procedure described in Section~\ref{sec:walking_overview}. 

For the analysis, in Section~\ref{sec:log_adaptivity_analysis}, we prove that \logadaptalg yields an LMP $(2+O(\eps))$-approximate solution for UFL. It is a nontrivial extension of the dual-fitting analysis of the \JMS algorithm for our definition of openability and design of the algorithm. %

\subsubsection{Merging Two Solutions}
\label{sec:walking_overview}
We describe our final algorithm \mergealg that maintains two LMP approximate solutions whose number of facilities sandwich $k$ and gradually merges them until one of them opens almost exactly $k$ facilities. 

\paragraph{Setup.}
Based on the discussion about \logadaptalg, we will describe a solution $\calH$ as $(H_1, \dots, H_L)$ with $L = O(\log n / \eps^2)$, where $H_i$ is a {\em valid execution} of stage 1 of phase $i$ (i.e., when $\theta = (1+\eps^2)^{i-1}$), which is formally represented as a sequence of opened facilities. $H_i$'s validity also requires {\em maximality}, meaning that no openable facility remains in the phase.
In order to facilitate the merging, given a phase, we extend the notion of a valid execution to include the following actions. 

\begin{itemize}
    \item We allow opening a {\em free facility}, at most three times per phase. A free facility $\tilde i$ is based on a regular facility $i \in F$, and when it is created, we can choose the distance parameter $u(\tilde i) = d(\tilde i, i)$. All other distances $d(\tilde i, x)$ are defined to be $u(\tilde i) + d(i, x)$. 

    \item We slightly relax the notion of openability so that one can open a facility $i$ by paying $\hat f - n \eta$ instead of $\hat f$ for parameter $\eta := 2^{-n}$ (in the third bullet of Definition~\ref{def:openable_overview}). 
\end{itemize}

Call $\calH = (H_1, \dots, H_L)$ a {\em solution of $\eta$-valid sequences} if $H_i$ is a valid execution (i.e., an $\eta$-valid sequence) of the $i$th phase in the above relaxed notion given $H_1, \dots, H_{i-1}$. 
We extend the analysis of \logadaptalg in Appendix~\ref{sec:robust_analysis_full} to show that such $\calH$ in the above sense still yields an LMP $(2+O(\eps + n\eta))$-approximation, if we do not count the opening cost of the free facilities (while using the connections they provide). 
Therefore, a solution of $\eta$-valid sequences that opens exactly $k$ regular facilities, and at most $3$ free facilities per phase,  is a $(2+O(\eps + n\eta))$-approximation to $k$-median that opens at most $O(L) = O(\log n/\eps^2)$ extra facilities. 

Initially, by the standard binary search on the facility cost and running \logadaptalg, we start with two solutions $\calH = (H_1, \dots, H_L)$ and $\calH' = (H'_1, \dots, H'_L)$ with facility costs $f$ and $f'$ such that $|f - f'| \leq \eta$ and their numbers of opened facilities sandwich $k$. (We will associate parameters $(f, u)$ for $\calH$ and $(f', u')$ for $\calH'$. Also, note that every facility in $\calH$ (resp. $\calH'$) has cost $f$ (resp. $f'$).)
Throughout the algorithm, we will maintain the invariants that (1) $\calH$ and $\calH'$ open the same set of free facilities (so that $u$ and $u'$ are defined on the same set), (2) their number of opened regular facilities sandwich $k$, and (3) $(f, u) = (f', u')$ except for one parameter $f$ or $u(\tilde i)$ for one free $\tilde i$, whose values in the two solutions differ by at most $\eta$. Call such a parameter the {\em difference parameter}. 

For each $i = 1, \dots, L$, we will outline a procedure that, given that the prefixes $(H_1, \dots, H_{p-1})$ and 
$(H'_1, \dots, H'_{p-1})$ are identical, modifies the suffixes $(H_p, \dots, H_L)$ and $(H'_p, \dots, H'_L)$ such that $H_p$ and $H'_p$ become identical while satisfying the above invariants, unless we already find a solution that opens exactly $k$ regular facilities. This will allow us to gradually merge the solution until one of them opens exactly $k$ regular facilities (and at most $O(L)$ free facilities). 

\paragraph{How to walk in each phase. }
Finally, we focus on phase $p$ and introduce ideas behind ensuring $H_p = H'_p$. Our first step is to change the difference parameter value of $\calH'$ to that of $\calH$; formally, we let $\calH'' = \completesolx_{(f, u)} (H'_1, \dots, H'_p)$, which is the solution resulting from running \logadaptalg from a given partial solution $(H'_1, \dots, H'_p)$ with parameters $(f, u)$ instead of $(f', u')$. Such a change of the difference parameter is the reason for introducing the $n\eta$ slack for $\eta$-valid sequences; for instance, if $H_i$ was the exact execution with parameters $(f' ,u')$, our Claim~\ref{claim:1ststep_etavalid} shows that it is still $\eta$-valid with respect to $(f, u)$.

If $\calH''$ opens $k$ regular facilities we are done. Also, if the number of open regular facilities of $\calH''$ and $\calH'$ sandwich $k$, then we already made progress as they now coincide on the $p$th phase while still having one difference parameter.
Therefore, we only need to handle the case where $\calH''$ and $\calH$ sandwich $k$, where the difference parameters already became identical but $H_p$ and $H'_p$ are different. Let us rename $\calH''$ to $\calH'$. 

Now our goal is to gradually change $H'_p$ to $H_p$ (i.e., {\em walk} from $H'_p$ to $H_p$). 
Instead of describing the full procedure, which is outlined in Section~\ref{sec:merging_algorithm}, we illustrate the procedure on a particular example. 
Suppose that $H_p = \langle i_1 \rangle$, $H'_p = \langle i_2, \dots, i_{6} \rangle$ where one can open at most one between $i_1$ and $i_j$ for any $j \in \{2, \dots, 6 \}$; this constraint can be expressed as a star graph with $i_1$ as the center vertex and $i_2, \dots, i_6$ are only connected to $i_1$, and a valid solution (i.e., $\eta$-valid sequence) here means a maximal independent set. (Suppose that this is the only constraint on openings, so $H_p$ and $H'_p$ can be thought of as sets instead of sequences.)

In this example, $H_p$ and $H'_p$ are the only two maximal independent sets, so it is impossible to {\em smoothly} walk from $H'_p$ to $H_p$. Here is where the notion of a free facility becomes useful because we can assume that having a free facility helps {\em maximality} while not hurting {\em independence}. Then, from $H'_p$, by (Step 1) adding a free copy of $i_1$ called $\tilde i_1$ with $u(\tilde i_1) = 0$, (Step 2) removing $i_2, \dots, i_6$ one-by-one, and (Step 3) converting $\tilde i_1$ to a regular $i_1$, we obtain a sequence of maximal independent sets where a consecutive pair of independent sets differ by at most one facility. Figure~\ref{figure:walking-overview} illustrates this process. 
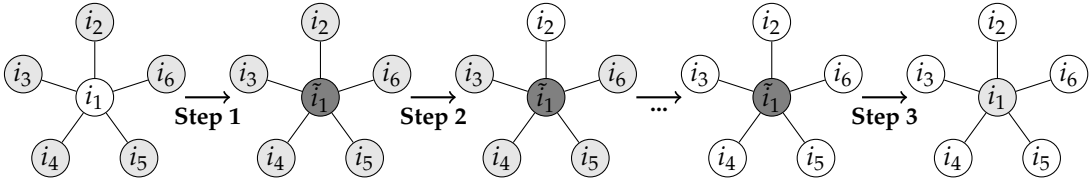
\begin{figure*}
\centering
\begin{tikzpicture}
    \def \radius {1}

    \node[circle, draw, minimum size=0.5cm, inner sep=0pt] (1a) at (0,0) {\(i_1\)};
    \foreach \i [count=\j from 2] in {90, 162, 234, 306, 18} {
        \node[circle, draw, fill=black!10, minimum size=0.5cm, inner sep=0pt] (\j a) at (\i:\radius) {\(i_{\j}\)};
        \draw (1a) -- (\j a); %
    }

    \node[circle, draw, fill=black!50, minimum size=0.5cm, inner sep=0pt] (1b) at (3,0) {\(\tilde{i}_1\)};
    \foreach \i [count=\j from 2] in {90, 162, 234, 306, 18} {
        \node[circle, draw, fill=black!10, minimum size=0.5cm, inner sep=0pt] (\j b) at ($(1b) + (\i:\radius)$) {\(i_{\j}\)};
        \draw (1b) -- (\j b); %
    }

    \node[circle, draw, fill=black!50, minimum size=0.5cm, inner sep=0pt] (1c) at (6,0) {\(\tilde{i}_1\)};
    \node[circle, draw, minimum size=0.5cm, inner sep=0pt] (2c) at ($(1c) + (90:\radius)$) {\(i_2\)};
    \draw (1c) -- (2c);
    \foreach \i [count=\j from 3] in {162, 234, 306, 18} {
        \node[circle, draw, fill=black!10, minimum size=0.5cm, inner sep=0pt] (\j c) at ($(1c) + (\i:\radius)$) {\(i_{\j}\)};
        \draw (1c) -- (\j c); %
    }

    \node[circle, draw, fill=black!50, minimum size=0.5cm, inner sep=0pt] (1d) at (9,0) {\(\tilde{i}_1\)};
    \node[circle, draw, minimum size=0.5cm, inner sep=0pt] (2d) at ($(1d) + (90:\radius)$) {\(i_2\)};
    \draw (1d) -- (2d);
    \foreach \i [count=\j from 3] in {162, 234, 306, 18} {
        \node[circle, draw, minimum size=0.5cm, inner sep=0pt] (\j d) at ($(1d) + (\i:\radius)$) {\(i_{\j}\)};
        \draw (1d) -- (\j d); %
    }

    \node[circle, draw, fill=black!10, minimum size=0.5cm, inner sep=0pt] (1e) at (12,0) {\(i_1\)};
    \foreach \i [count=\j from 2] in {90, 162, 234, 306, 18} {
        \node[circle, draw, minimum size=0.5cm, inner sep=0pt] (\j e) at ($(1e) + (\i:\radius)$) {\(i_{\j}\)};
        \draw (1e) -- (\j e); %
    }

    \draw[->, thick] ($(1a) + (1.2, 0)$) -- ($(1b) + (-1.2, 0)$) node[midway, below] {\small\textbf{Step 1}};
    \draw[->, thick] ($(1b) + (1.2, 0)$) -- ($(1c) + (-1.2, 0)$) node[midway, below] {\small\textbf{Step 2}};
    \draw[->, thick] ($(1c) + (1.2, 0)$) -- ($(1d) + (-1.2, 0)$) node[midway, below] {\small\textbf{...}};
    \draw[->, thick] ($(1d) + (1.2, 0)$) -- ($(1e) + (-1.2, 0)$) node[midway, below] {\small\textbf{Step 3}};

\end{tikzpicture}
\caption{How we walk from $H'_p$ to $H_p$. Dark grey color denotes a free facility and light grey color denotes regular open facilities.}
\label{figure:walking-overview}
\end{figure*}

For each intermediate $H'_p$, we consider its completion (i.e., $\completesolx_{(f, u)} (H'_1, \dots, H'_p)$) and see how many regular facilities it opens. Since the completions from the initial $H'_p$ and $H_p$ sandwich $k$, there must be a consecutive pair in the above sequence of independent sets (say $H^{a}_p$ and $H^b_p$) that sandwich $k$, and they differ by only one facility $i_q$ (say $H^{a}_p$ contains $i_q$ while $H^b_p$ does not). For future execution of the algorithm, $H^{a}_p$ and $H^{b}_p$ are the same as the solutions that have the free copy $\tilde i_q$ (instead of $i_q$) with the $u(\tilde i_q)$ value $0$ (for $H^{a}_p$) and $\infty$ (or sufficiently large number, for $H^{b}_p$). Therefore, by performing binary search on the $u(\tilde i_q)$, just like we do on the facility cost $f$, we yield two solutions that are now identical up to phase $p$ except the $\eta$ difference on the only one different parameter $u(\tilde i_q)$, which archives our goal of merging two solutions at phase $p$!%

In general, the fact that $H_p$ and $H'_p$ are not just sets but sequences of opened facilities requires more technical consideration. 
However, the full procedure is still an extension of this idea, trying to increase the length of the common prefix between $H_p$ and $H'_p$ by walking from $H'_p$ to $H_p$ by one coordinate at a time, until two consecutive solutions are sandwiching $k$. 

\subsection{$(2+\varepsilon)$-Approximation for $O_\eps (\frac{1}{\log n})$-Stable Instances}
\label{sec:removal-overview} 

The other component of our algorithm is an improved approximation to $k$-median assuming that the instance has a weak version of a structural property known as stability, initially formulated by \cite{ORSS12}. Recall that an instance is $\beta$-stable if $\sopt_{k-1}\geq (1+\beta)\sopt_k$.
For any $\beta$-stable instance, 
\cite{ABS10,Cohen-AddadS17} showed how to obtain 
a PTAS for the problem in time $n^{\beta^{-1}\eps^{-O(1)}}$ and~\cite{ABS10} left as an open question
as whether one could obtain a PTAS with running 
time $f(\beta,\eps)n^{O_{\eps}(1)}$ for any computable 
function $f$ for general metrics (they show it is possible for Euclidean metrics). We would indeed
like to have an algorithm running in time
$2^{\text{poly}(\beta, \eps)} n^{O_{\eps}(1)}$ to 
solve our $(\eps / \log n)$-stable instances.

Unfortunately, since any $k$-median instance is 
$\Omega(1/k)$-stable, obtaining a PTAS with the above
running time would violate Gap-ETH~\cite{Cohen-AddadG0LL19}: The $k$-median
problem cannot be approximated better than $1+2/e$ in time $f(k,\eps) n^{O_\eps(1)}$ assuming Gap-ETH. 

The saving grace here is that we are not after 
a $(1+\eps)$-approximation algorithm but we can use
any $(2+O(\eps))$-approximation with the above running time. 
Since a connection between $\beta$-stable and 
FPT algorithm exists (as explained in the above discussion), it is tempting to make the journey 
from FPT algorithms back to $\beta$-stable instances and use
the state-of-the-art approximation algorithm of~\cite{Cohen-AddadG0LL19} to solve $\beta$-stable 
instances. Unfortunately the algorithm presented
in~\cite{Cohen-AddadG0LL19} runs in time
$(\log n)^{\tilde{O}(k/\eps)} n^{O(1)}$ and so
importing the techniques directly would result in a
running time of the form $(\log n)^{\tilde{O}((\beta/\eps)^{-1})} n^{O(1)}$ and  not to a polynomial
running time for our setting.
We must thus go beyond the state-of-the-art techniques for FPT and $\beta$-stable $k$-median 
instances.

\subsubsection{Local Optima and Stability}
\label{sec:removal-overview-local}

One particularly important tool when dealing with stable solutions are properties of local optima. The local search procedure starts with any feasible solution $S$, and searches for an alternative solution $S'$ such that $|S\triangle S'|= 2$ (i.e., $S$ and $S'$ differ by one center) and $\text{cost}(S')< \text{cost}(S)$. If such an $S'$ exists, $S$ is replaced by $S'$ and the process is iterated. At the end of the procedure we say that $S$ is locally optimal. The running time of this procedure is polynomial, where we exploit the assumption that distances are polynomially bounded integers (by Lemma \ref{lem:aspectratio}).

We first introduce a bit of notation. Suppose we are given a locally optimal solution $S$. We say that a cluster $C^*$ of the optimal solution $\opt$ is \emph{pure} with respect to $S$, if there exists a cluster $C'$ induced by $S$ such that both clusters agree on all but  $O(\varepsilon)\cdot \min(|C^*|,|C'|)$ of the points. 
We begin by examining structural properties for $\beta$-stable solutions, summarized by in 
Section~\ref{sec:localSearchAnalysis}.
\begin{itemize}
    \item For any given constant-factor approximation $S$, at most $O_\eps(\beta^{-1})$ many clusters are not pure, see Lemma \ref{lem:numnonpure}.
    \item In any locally optimal solution $S$, we have
    $$\sum_{p\in C^*,~C^* \text{ pure}} d(p,S) \leq \sum_{p\in C^*,~ C^* \text{ pure}} d(p,\opt) + O(\varepsilon)\cdot \sopt,$$ that is the cost of the clients in pure clusters is well approximated.
\end{itemize}
More precisely, in the second bullet we show that all the points in each pure cluster $C^*$ with center $c^*$ can be assigned to the \emph{same} center $t(c^*)$ in $S$ (i.e., the center $c'$ of $C'$) with a moderate increase of their cost (see Lemma \ref{lem:purecost}).

The first property relies on the fact that for every non-pure cluster, $S$ must merge at least two clusters of the optimum. If $S$ is a constant-factor approximation, this cannot happen more than $O_{\eps}(\beta^{-1})$ many times before the approximation guarantee no longer holds.
The second property is implied by the local optimality of $S$. Specifically, swapping out the center $c'\in S$ serving $C'$ for the center $c^*$ of the pure cluster $C^*$ cannot change the cost by much as these two clusters agree on most of their points.

In a nutshell, a locally optimal solution $S$ approximates all but $O_{\eps}(\beta^{-1})$ many clusters of the optimal solution very well. In what follows, we will describe how we may extend $S$ to a $(2+\varepsilon)$-approximation overall.

\subsubsection{Identifying the Non-Pure Clusters}
To find a good approximation for the non-pure 
clusters, we would like to identify these clusters.
By the above discussion, we may assume that we have a locally optimal solution $S$ for which all but $O_\eps(\beta^{-1}) = O_\eps(\log n)$ many clusters are well approximated.
Immediate ideas, such as determining the badly approximated clusters by brute force, will not improve over the quasi-polynomial running time.
Instead, we look towards fixed parameter tractable (FPT) approaches, parameterized by the number of centers. In fact,
we would like an FPT approximation algorithm \emph{parameterized
by the number of non-pure clusters}.
Starting point is an FPT algorithm based on submodular maximization, initially due to \cite{Cohen-AddadG0LL19}, which we first review. However, recall that the running time of this
algorithm is $(\log k)^{\tilde{O}(k\eps^{-1})}$ and so will
not be enough for our purpose. We thus introduce new ideas to
speed things up.

\paragraph{Review of the FPT approximation algorithm 
of~\cite{Cohen-AddadG0LL19}.}
Suppose $c^*$ is a center of cluster $C^*$ in an optimal solution $\opt$. \cite{Cohen-AddadG0LL19} introduce the notion of a \emph{leader} client $\ell\in C^*$ which minimizes the distance to $c^*$. 
Using coresets to first reduce the number of clients
to $\tilde{O}(k\log n/\eps^{2})$, the leaders of all clusters, as well as the distances $d(\ell,c^*)$ can be identified in FPT-time, i.e. in time $(\log n)^{\tilde{O}(k \varepsilon^{-1})} \cdot n^{O(1)}$.
A solution $S$ obtained by simply selecting an arbitrary facility within distance $d(\ell,c^*)$ from $\ell$ yields a $3$-approximation.
\cite{Cohen-AddadG0LL19} improve upon this by showing that for any two candidate sets of centers $S'$ and $\dummyset$
$$ \text{cost}(\dummyset)-\text{cost}(\dummyset\cup S')$$
is a monotone submodular function. 

The first part which consists in guessing the leaders is the 
computational bottleneck since optimizing a monotone submodular
function can be done in polynomial time.

\medskip We next describe how we deal with the clusters that are not well approximated. Recall that we first compute a locally optimal solution $S$.
By the properties of local optima, $S$ is already a sufficiently good solution for the pure clusters. For the remaining $O_\eps(\log n)$ non-pure clusters, we only have coarser approximation bounds. 

\paragraph{Leader Identification in Polynomial Time via $D^2$-Sampling -- Sections~\ref{sec:dsampleproc} and~\ref{sec:ballguesses}.}
Our next high-level goal is to find leaders of the non-pure clusters and guess their
(approximate) distances to the respective optimum centers in time $2^{\beta^{-1}\eps^{-O(1)}} n^{O(1)}$ for
$\beta$-stable instances.

\cite{Cohen-AddadG0LL19} identify leaders in time ${k\log n \choose k}$ by sparsifying the instance via coresets and then determining the distance of leaders to the optimal solution in time $(\log n)^k$.
Our overall approach significantly speeds up these two steps by way of distance sampling similar to $k$-means++ \cite{ArV07}.

To identify the leaders faster, we lift the stringent requirement that the leader of cluster $C^*$ is the client closest to the center $c^*$ in the optimal solution. Instead, we consider any client among the $\varepsilon\cdot |C^*|$ closest points to $c^*$ as a valid leader with $\text{avg}_{C^*, \opt}$ being the maximum distance of any leader of $C^*$ to $c^*$. We later show that this only induces
a minor cost increase in the total approximation bound.

The second relaxation is that we only aim to target clusters that are sufficiently expensive in our locally stable solution $S$. More precisely, a non-pure cluster $C^*$ with center $c^*$ is \emph{basic-cheap} if (a) the total cost of its leaders in $S$ is at most $ \varepsilon^{O(1)}\cdot\beta \cdot \sopt = \frac{\varepsilon^{O(1)}}{\log n} \cdot \sopt$ or (b) $d(c^*,S)\leq (1+\varepsilon)\text{avg}_{C^*, \opt} +\frac{\eps}{|C^*|} \cdot \sum_{p\in C^*}(d(p,c^*)+d(p,S))$. 
Then we can show that $S$ is already an almost $2$-approximation for these clusters modulo some additive slack, namely \[
\sum_{p\in C^*,~C^* \text{ basic-cheap}} d(p, S) \leq 
2 \sum_{p\in C^*,~C^* \text{ basic-cheap}} d(p, \opt) + O(\eps \cdot \sopt).\] 
More precisely, similarly to the case of pure clusters, we can show that all the points of a basic-cheap cluster can be reassigned to the center $t(c^*)\in S$ which is closest to $c^*$ with an increase of the cost of the above type (see Lemma \ref{lem:cheapcost}). Indeed, by triangle inequality $d(p, S) \leq d(p, c^*) + d(c^*, S)$ for any $p$. If $C^*$ is basic-cheap by condition (b), one directly obtains the desired upper bound on $d(c^*, S)$. Otherwise, i.e. if $C^*$ is basic-cheap by  condition (a), we use the fact that the number of basic-cheap clusters (which is at most the number of non-pure clusters) is $\frac{\log n}{\eps^{O(1)}}$ so the sum of their costs in $S$ is $O(\eps \cdot \sopt)$. This essentially allows us to focus only on non-pure and non-basic-cheap clusters.

We use these two relaxations as follows. 
Let $W$ be a sample of clients picked independently and proportionately to $\frac{d(p,S)}{\text{cost}(S)}$. 
With probability $\varepsilon^{O(1)}/\log n$, we are picking a leader from a given non-pure, non-basic cheap cluster: in this case we say that $C^*$ is hit by the sample.
Thus, when picking $O_\eps(\log n)$ samples, we have picked a leader from all but a small fraction of non-pure, non-basic cheap clusters. Thereafter, we can show that with good probability, we can assign all the points of a non-pure non-basic-cheap non-hit cluster $C^*$, with center $c^*$, to the center $t(c^*)\in S$ which is closest to $c^*$ with a small cost (see Lemma \ref{lem:probcost} in Section~\ref{sec:dsampleproc}). We next call \emph{expensive} the non-pure non-basic-cheap clusters which are hit, and \emph{cheap} the non-pure clusters which are either basic-cheap or non-basic-cheap and non-hit. Thus the clusters are partitioned into pure, cheap and expensive ones. Since we have sampled $O_\eps(\log n)$ clients, we can enumerate all the subsets in polynomial time
and focus on a subset that contains exactly one leader of each expensive cluster.

Using a constant number
of potential guesses per leader, we then show how to get a sufficiently good approximation of the distance $\rho$ from 
each leader to its closest optimum center, see Lemma \ref{lem:ballguesses} in Section~\ref{sec:ballguesses}. Intuitively, the cost of the leader in solution $S$ is a rough approximation of $\rho$, so we can interpolate around it.

\paragraph{Existence of a Good Solution.}
We now use $W$ from last section to identify a modified solution from $S$ that is $(2+O(\eps))$-approximate. Let $\hat{OPT}$ denote the set of centers of expensive clusters.
The cost of all the other optimal clusters, i.e., the pure and cheap ones, is within a factor roughly $2$ from the optimum.
We thus wish to add to $S$ the optimal centers of $\hat{\opt}$ and remove an equally sized set $S_0\subseteq S$.
An important constraint on $S_0$ is that it contains no center 
$t(c^*)$, where $c^*$ is the center of an optimal cluster which is pure or cheap. Intuitively, such centers $t(c^*)\in S$ guarantee that the cost of points in pure and cheap clusters is within a factor roughly $2$ from the optimum.
More precisely, in Lemma \ref{lem:structSminusS0} we show that such an $S_0$ exists so that 
$M_O := S - S_0 \cup \hat{\opt}$ is a $(2+O(\eps))$-approximation:
If a cluster of $\opt$ is pure or cheap, there is one center
$t(c^*)$ in $S-S_0$ serving their clients up to a $2+O(\eps)$ factor; for the remaining ones, we would have their optimum 
center (in $\hat{\opt}$).

So our goal from now on is to approximate the solution 
$M_O$, where $|S_0| = \hat{\opt} = O_{\eps}(\log n)$, and this final goal has two key components: (1) removing some existing centers from $S$ and (2)
adding the centers from $\hat{\opt}$. The issue is that stated
as such this seems as hard as solving the $k$-median problem.
We first address (1), and we will address (2) by making use
of the leaders we have correctly guessed for the centers in
$\hat{\opt}$ at the previous steps.

\paragraph{Center Removal.}
Let us first try to address (1). Remember we start with solution $S$, and the leader and radii guesses from the previous steps. The key issue with (1) is that
the clients in some of the clusters we would like to remove
(i.e., served by $S_0$ in solution $S$) are actually served by a center in $\hat{\opt}$, but we do not know which are these 
centers yet! We will thus make use of some approximate loose bound
on their cost in solution $S - S_0 \cup \hat{\opt}$ by introducing
\emph{dummy} centers $\dummyset$. For each leader $\ell$ and guessed distance/radii $\rho$,
we will introduce a dummy center at distance $\rho$ from $\ell$ 
and at distance $\dist(p, \ell) + \rho$ from any other vertex $p$.

If we promise to open a center in the ball of radii $\rho$ around
$\ell$ -- which solution  $S - S_0 \cup \hat{\opt}$ does -- then
we will guarantee that there will be a center not further away
than the dummy centers in the final solution. We now work with
this solution $S \cup \dummyset$ which provides a worst-case fall 
back on where the centers in $\hat{\opt}$ could be. We can show
that $M_D := S \cup \dummyset - S_0$ is a $(3+O(\eps))$-approximate solution.

Equipped with this we proceed to our 3-step center removal. 
We know that we have at most $|S_0| = O_{\eps}(\log n)$ centers
we need to remove. Our first procedure (see Section~\ref{sec:removalofExpensive}) aims at removing the centers of clusters of high cost in $S_0$.
Start with solution $S \cup \dummyset$. For $O_{\eps}(\log n)$ times, we sample a center $c\in S$ proportionally to the cost of its cluster in solution $S \cup \dummyset$, remove $c$ from $S$ with probability 
$1/2$, and repeat on the resulting solution. Let $\calQ$ be the centers removed from $S$. Then 
with 
probability at least $2^{-O_{\eps}(\log n)}$ the procedure 
succeeds by either obtaining $\calQ=S_0$
or ensuring that the total cost of the clusters of $S_0-\calQ$ in the solution is 
at most $O(\eps \cdot \sopt)$. This indeed follows from the fact that, 
if we only removed centers $\calQ\subseteq S_0$ up to a certain step,
the solution is never more than a $
(3+O(\eps))$-approximation (since $S 
\cup \dummyset - S_0$ is a $(3+O(\eps))$-approximation) and so the probability
of sampling a new center of $S_0-\calQ$ is $\Omega(\eps)$, unless the total cost of the clusters of $S_0-\calQ$ is small. Therefore, repeating the process $n^{O_\eps(1)}$
times, we have that one of the runs succeeds with high probability.

Thus, after this step, we have removed a set $\calQ$ of centers
of $S_0$ and the total cost of the clusters of $S_0-\calQ$ is $O(\eps \cdot \sopt)$.
Following this procedure, we only have to remove inexpensive clusters. While here it is not possible to identify the correct ones, making a mistake is not too costly. For any $c\in S_0-\calQ$, we can now think of its whole cluster
as being assigned entirely to a unique center in $S - S_0 \cup
\hat{\opt}$ since by triangle inequality this would only increase
its cost by $O(\eps \cdot \sopt)$ in total.
We can thus think of the remaining centers $S_0-\calQ$ as either being assigned to $\hat{\opt}$ (we call this set $\calR$), or to a center in $S - S_0$ (we call this set $\calU$).
Suppose we know how many centers fall into these two types, a figure we determine by trying all possibilities. The next step would be ideally to remove $\calU$. We instead remove a \emph{proxy} $\bcalU$ of $\calU$, $|\bcalU|=|\calU|$(see Section \ref{sec:removalofCheap}).
A natural starting point to find a center for reassignment is to do so by greedily removing the center $c$ for which the reassignment cost is minimized. However, we must be careful. First, the reassignment costs are dynamic, that is when removing a center, it affects reassignment costs of subsequently removed ones. Second, if a center $c$ was used to bound the reassigment cost of an deleted center, we cannot delete it later in order to prevent the error from adding up. We design a procedure that  produces a $\bcalU$ satisfying all these constraints (see Lemma \ref{lem:successcheapremove}). We remark that $\bcalU\cap \calR=\emptyset$, hence removing $\calR$ is still a valid option in the later stages.

Thus, for the purpose of analysis (the algorithm does not know $\calR$), we obtain the solutions $M_D' = S- (\calQ \cup \bcalU\cup \calR) \cup \dummyset$ and $M_O' = S- (\calQ \cup \bcalU \cup \calR) \cup \hat{\opt}$. Moreover, the costs satisfy $\text{cost}(M_O')+\text{cost}(M_D')\leq (4+O(\varepsilon))\opt$, which is the same as summed up costs of $M_O$ and $M_D$, see Lemma \ref{lemma:costboundofMOandMDwithassignments} and Claim \ref{claim:Mprimebounds}.

\paragraph{Improvement via Submodular Maximization.}
We can now use submodular maximization techniques (specifically 
algorithms to maximize a submodular function over a partition matroid constraint) in a similar but more general setting as in the work of~\cite{Cohen-AddadG0LL19}. 
In more detail, we wish to find the best $|\hat{\opt}|$ centers $X$ to open (instead of the dummy centers $\dummyset$), one at the prescribed distance from each leader, to
improve the cost of $M_D'$ as much as possible. This is almost 
identical to the work of~\cite{Cohen-AddadG0LL19} except that we also need to close $|\calR|$ other centers $\widetilde{\calR}$ in $S-\calQ-\bcalU$.  Ideally, we would like to open $\hat{\opt}$ and close $\calR$. We show that selecting a good enough pair $(X,\widetilde{\calR})$ can be reduced to a submodular maximization problem under a partition matroid constraint. In particular, we carefully define our submodular function so that its value, for a given subset $X$ of centers, is calculated by optimally selecting the set $\widetilde{\calR}$ to close (for the choice $X$).  Then we can use state-of-the-art algorithms to optimize this and \emph{simultaneously} approximate the centers $\hat{\opt}$ to open and the set $\calR$ to close.

\section{$(2+\eps)$-Approximation with $O(\log n /\eps^2)$ Extra Centers}
\label{sec:log_adaptivity}

In this section, %
we present a polynomial-time algorithm that opens $k+ O(\log n/\eps^2)$ many centers and returns a solution of cost at most $(2+\eps)\sopt$. We restate the theorem for convenience of the reader.
\pseudoapprox*

Our algorithm is a non-trivial extension of the \JMS algorithm with dual fitting analysis~\cite{JMS02,JainMMSV03}, next called \JMSalg, so as to guarantee that we only open $k+ O(\log n/\eps^2)$ centers instead of only doing so in expectation.
We first recall the basics of \JMSalg\, in Section~\ref{sec:basicJMS}. Section~\ref{subsec:log_adaptive} presents a modified algorithm \logadaptalg that runs in $O(\log n / \eps^2)$ adaptive steps (that we call {\em phases}) and Section~\ref{sec:log_adaptivity_analysis} analyzes the latter algorithm.

Finally, building on additional setups in Section~\ref{sec:robust_analysis},
Section~\ref{sec:merging_algorithm} presents our final algorithm \mergealg achieving the guarantee of Theorem~\ref{thm:pseudo-approx}.
It is an extension of \logadaptalg maintaining two solutions that open at least and at most $k$ centers respectively, while opening at most three {\em free} centers per phase; this replaces the {\em bipoint rounding}~\cite{JaiV01, LiS16, ByrkaPRST17, GowdaPST23} and directly returns a $(2+\eps)$ approximation that opens at most $k + O(\log n /\eps^2)$ centers. Its analysis, which is an extension of the analysis of Section~\ref{sec:log_adaptivity_analysis}, is presented in Appendix~\ref{sec:robust_analysis_full}.

\subsection{Review of the \JMS Algorithm}
\label{sec:basicJMS}
In this section, we provide a brief overview of the classic $2$-approximate LMP algorithm from~\cite{JMS02, JainMMSV03}. Our algorithms are derived by modifying this procedure, and these modifications become more intuitive when we revisit the basic concepts of the original algorithm. We present it in Algorithm~\ref{alg:JMS} as a primal-dual algorithm to emphasize the changes we introduce.

\begin{figure}[h!]
    \begin{center}
        \begin{minipage}{1.0\textwidth}
            \begin{mdframed}[hidealllines=true, backgroundcolor=gray!15]

                \begin{algorithm}[\JMSalg] \ \\[0.2cm]
                    \textbf{{Initialization:}} Set $\hat f = 2f, S = \emptyset, A=D$ and $\alpha_j = 0$ for every $j\in A=D$. \\[0.1cm]

                    {While $A\neq \emptyset$:}\\[-0.1cm]

                    Increase $\alpha$-values of active clients  $(\alpha_j)_{j\in A}$ uniformly until one of the following events happen:

                    \begin{enumerate}
                        \item There is an unopened facility $i$ whose bids pay for it:
                              \[
                                  \sum_{j\in A} [\alpha_j - d(i,j)]^+ + \sum_{j\in D - A} [d(j,S) - d(j,i)]^+ = \hat f\,.
                              \]
                              Open $i$ (i.e., add it to $S$) and mark all clients $j\in A$ such that $\alpha_j \ge d(j,S)$ as served (i.e., remove them from $A$).
                        \item There is a client $j\in A$ so that \[\alpha_j = d(j,S).\]
                              The client $j$ is removed from the set $A$ of active clients.
                    \end{enumerate}
                    \label{alg:JMS}
                \end{algorithm}
            \end{mdframed}
        \end{minipage}
    \end{center}
\end{figure}

The algorithm maintains the set $A$ of \emph{active} clients and the set $S$ of opened facilities. Initially, all clients are active, i.e., $A = D$, and no facility is open, i.e., $S = \emptyset$. Additionally, the algorithm maintains an approximately feasible dual solution $\alpha$, where each client $j \in D$ has a corresponding value $\alpha_j$. These values are initialized to $0$.

While there is at least one active client, the algorithm uniformly increases the $\alpha$-values of the active clients until there is a facility $i$ whose "bids" are sufficient to open it or an active client ``gets connected'' to an already opened facility. The bids from a client $j$ to a facility $i$ depend on whether $j$ is active or not. If $j \in A$, then a portion of its budget $\alpha_j$, specifically $d(i,j)$, is allocated to cover the connection cost of $j$, and the remainder, $\alpha_j - d(i,j)$, is used to bid towards the facility's opening cost $\hat f:=2f$. For a client $j \notin A$, which is already connected to an open facility at a distance $d(j,S)$, the bid towards $i$ is the difference in distance, i.e., $d(j,S) - d(i,j)$ (or $0$ if $i$ is farther away from $j$ than $S$).

Finally, the algorithm marks any client $j$ with $\alpha_j \geq d(j,S)$ as inactive, ensuring that at any point in time, active clients $j$ satisfy $\alpha_j \leq d(j,S)$.

An overview of the analysis of \JMSalg is as follows.  We consider the final $\alpha$-values denoted by $(\alpha^*_j)_{j\in A}$, and the final set $S^*$ of opened facilities.
The definition of the bids is to guarantee that the $\alpha^*$-values pay for both the opening costs of the facilities in $S^*$ and the connection costs of the clients:
\begin{lemma}
    $\sum_{j\in D} \alpha^*_j = \sum_{j\in D} d(j,S^*) + \hat f \cdot |S^*|.$
    \label{lemma:JMS-cost}
\end{lemma}

This is complemented by showing that $\alpha^*/2$ is a feasible dual solution.
\begin{lemma}
    For a facility $i$, $\sum_{j \in D}[ \alpha^*_j - 2d(i,j)]^+ \leq \hat f$.
    \label{lemma:JMS-dual-feasible}
\end{lemma}
To prove dual-feasibility, two key properties are used:
\begin{description}
    \item[No overbidding:] At any point in time, for every $i\in F$,
          \begin{gather*}
              \sum_{j\in A} [\alpha_j - d(i,j)]^+ + \sum_{j\in D - A} [d(j,S) - d(j,i)]^+ \leq \hat f\,.
          \end{gather*}
    \item[No active clients bid on open facilities:] At any point in time, $\alpha_j \leq d(j, S)$ for $j\in A$.
\end{description}
Specifically, if we give a total order of the clients so that $j\leq \ell$ if $\alpha^*_j \leq \alpha^*_\ell$, the above two properties are used to prove the following technical lemma; intuitively, at time right before $\alpha^*_{\ell}$, client $j \in D-A$ satisfies $d(j, S) + d(j, i) + d(i, \ell) \geq \alpha^*_{\ell}$, which implies $d(j, S) - d(j, i) \geq \alpha^*_{\ell} - 2d(i, j) - d(i, \ell)$.
\begin{lemma}
    For any $\ell \in D$,
    \begin{equation*}
        \sum_{j \in D : j \geq \ell} [\alpha^*_\ell - d(i, j)]^+ +
        \sum_{j \in D : j < \ell} [\alpha^*_\ell - 2d(i, j) - d(i, \ell)]^+ \leq \hat f.
    \end{equation*}
    \label{lemma:JMS-dual-feasible-technical}
\end{lemma}
Lemma~\ref{lemma:JMS-dual-feasible} then follows by summing up the above inequality for all $\ell\in D$.
Let $\optlpfl(f)$ be the value of an optimal solution to the linear programming relaxation when opening costs are $f$.
Then, we get the following theorem by combining Lemma~\ref{lemma:JMS-cost}  and Lemma~\ref{lemma:JMS-dual-feasible}.
\begin{theorem}
    The set $S^*$ of facilities output by \JMSalg satisfies
    \[
        \sum_{j\in D} d(j, S^*) \leq 2 \cdot \left( \optlpfl(f) - f \cdot |S^*| \right)\,.
    \]
    \label{thm:JMS}
\end{theorem}
Recall that $\optlpfl(f) - f\cdot k$ is a lower bound on the optimal solution to the $k$-median problem (see Section~\ref{sec:prelim}). Hence, if \JMSalg outputs a solution $S^*$ such that $|S^*| = k$, then it is a 2-approximate solution to $k$-median. Moreover, we can always set $f$ high enough so that any good solution opens at most $k$ facilities and $f$ low enough so that any good solution opens at least $k$ facilities. Thus, a simple binary search can find values $f$ and $f'$ such that running \JMSalg with opening costs $f$ opens $k_1 \geq k$ facilities, while running it with opening costs $f'$ opens $k_2 \leq k$ facilities. However, even though we can ensure that $|f - f'|$ is exponentially small, say $1/2^n$, it turns out that \JMSalg is quite unstable, and we may have $k_1 \ll k$ and $k_2 \gg k$, which does not directly imply anything useful for $k$-median. In other works, this has been addressed through an additional rounding algorithm called bi-point rounding, which introduces an extra factor into the approximation guarantee. In contrast, we modify \JMSalg into a similar algorithm that is ``smoother'' and less adaptive in its selection of facilities to open.

\subsection{\logadaptalg}
\label{subsec:log_adaptive}

It is natural to think of the \JMSalg as increasing the $\alpha$-values of active clients with a time $\theta$. That is, we always have $\alpha_j = \theta$ for $j\in A$. A difficulty in ensuring that \JMSalg opens $k$ facilities is that it is a very adaptive algorithm. To see this, consider the case when two facilities $i$ and $i'$ both become paid for at the same time $\theta$. Whether we open $i$ or $i'$ can dramatically change the bids to future facilities and completely change the trace and what facilities are opened by the algorithm. Our idea to overcome this difficulty is to limit the number of time steps where such decisions are made, thus lowering the adaptivity of the algorithm. Specifically, our high-level goal is to only consider the time steps of the form $(1+\eps^2)^0, (1+\eps^2)^1, (1+\eps^2)^2, \ldots $. Then, as the minimum nonzero distance is $1$ and the maximum distance is $\poly(n)$, the standard arguments (formalized in Section~\ref{sec:merging_algorithm}) imply that the total number of time steps that we need to consider is $O(\log(n)/\eps^2)$, and this is what we refer to as $O(\log n)$ adaptivity.

The difficulty in only considering those time steps is that we easily lose the property that we never overbid a facility $i$, which is crucial in the analysis of \JMSalg. Indeed, consider a time step $\theta = (1+\eps^2)^t$ and let $A$ and $I$ be the current active and inactive clients, respectively. Now, suppose that there is a facility $i$ whose bids
\begin{gather*}
    \sum_{j\in A} [\alpha_j - d(i,j)]^+ + \sum_{j\in I} [d(j,S) - d(j,i)]^+
\end{gather*}
are just barely below $\hat f$. However, in that case, increasing the $\alpha$-value of active clients by a factor $1+\eps^2$ may cause the bids to be strictly above $\hat f$, leading to difficulties in the analysis. (Specifically, in the proof that $\alpha^*/2$ is dual feasible.)
We intuitively change \JMSalg in two ways to handle this. First, we open centers if the bids pay for them in the metric scaled down by a factor of $(1-\delta)$. That is, we consider a facility $i$ paid for if
\begin{gather*}
    \sum_{j\in A} [\alpha_j - (1-\delta) d(i,j)]^+ + (1-\delta) \sum_{j\in I} [d(j,S) - d(j,i)]^+  \geq \hat f\,.
\end{gather*}
As we scale down the metric (i.e., the connection cost) by a factor $1-\delta$, this loses a factor $1/(1-\delta)$ in the approximation guarantee.
However, the benefit is the following.
Let $X = \{j\in A \mid \alpha_j \geq (1-\delta) d(j,i)\}$ be the  active clients that bid for $i$. Then
\[
    \sum_{j\in X} \left(\alpha_j - (1-\delta) d(i,j) \right)  \geq  \sum_{j\in X} \left((1+\eps^2) \alpha_j  - d(i,j) \right)
\]
assuming that $\sum_{j\in X} \eps^2 \alpha_j = \eps^2 \theta |X|$ is smaller than $\delta \sum_{j\in X} d(i,j)$. We will select $\delta = 3\eps$ and so this means that it is safe to increase the $\alpha$-values of active clients by a factor $(1+\eps^2)$ \emph{as long as} the average distance from bidding active clients is at least $\eps^2 \theta/\delta = \eps \theta/3$.
Our second change deals with the remaining case when the clients in $X$ are close to co-located with $i$. We do so by allowing the clients $j$ in the ball $B(i, \eps \theta)$ to increase their bids for $i$ up to $\min((1+\eps^2)\theta, (1-\delta)d(j,S))$. These changes are captured in the following definition of openable facilities. The bids $(\tau_j)_{j\in A}$ from active clients are only allowed to increase for nearby clients (the first two properties). A facility is considered paid for with respect to the scaled down metric (the third property). Finally, for the proof of dual feasibility even after increasing the bids of nearby clients, we need a technical fourth condition in the definition. We remark that we use \algstatex to denote the algorithm's state, where $\alpha$ is the approximate feasible dual solution, $S$ the set of open facilities, $A$ the set of active clients (which determines $I = D \setminus A$), and $\theta$ the time. Both \JMSalg and \logadaptalg (that we describe below) satisfy $\alpha_j = \theta$ for all $j\in A$. %

\begin{definition}
    A facility $i$ is \emph{openable} with respect to the algorithm's state \algstatex if there are bids $(\tau_j)_{j\in A}$ of active clients that satisfy the following conditions.
    \begin{itemize}
        \item Only nearby clients are increased:
              \begin{gather*}
                  \tau_j = \alpha_j \qquad \mbox{for every $j \in A - B(i, \eps \theta)$.}
              \end{gather*}
        \item Nearby clients are only increased slightly:
              \begin{gather*}
                  \alpha_j \leq \tau_j \leq \min\{(1-\delta) d(j,S), (1+\eps^2)\theta\} \qquad \mbox{for every $j\in A \cap B(i, \eps \theta)$.}
              \end{gather*}
        \item Facility $i$ is paid for:
              \begin{gather*}
                  \sum_{j\in A} [\tau_j - (1-\delta)d(i,j)]^+ + (1-\delta) \sum_{j\in I} [d(j,S) - d(i,j)]^+ \geq \hat f\,.
              \end{gather*}
        \item Dual feasibility: for every facility $i_0$ and $\ell \in A$,
              \begin{gather*}
                  \sum_{j \in A} [\tau_j - d(i_0, j)]^+ + \sum_{j \in I} [\tau_\ell - 2d(j, i_0) - d(\ell, i_0)]^+ \leq \hat f.
              \end{gather*}
    \end{itemize}
    \label{def:openable}
\end{definition}
We note that when the state of the algorithm is clear from the context, we simply say that "$i$ is openable" without explicitly specifying the state.

The following lemma follows from that we can write the selection of the bids as a linear program. Indeed, given the state $(\alpha, S, A, \theta)$ of the algorithm, each of the four conditions on the bids $(\tau_j)_{j\in A}$ can be formulated as linear inequalities.
\begin{lemma}
    There is a polynomial-time algorithm that, given the state $(\alpha, S, A, \theta)$ and a facility $i$, either outputs bids $(\tau)_{j\in A}$ that satisfy all the conditions of Definition~\ref{def:openable} or certifies that $i$ is not openable.
\end{lemma}
\begin{proof}
    Fix $i$ and $(\alpha, S, A, \theta)$ and consider a linear program with variables $(\tau_j)_{j \in A}$. We claim that each of the four bullets in Definition~\ref{def:openable} can be expressed as linear constraints. Indeed, the first two bullets are immediate. As $\tau_j = \alpha_j$ for every $j \in A - B(i, \eps \theta)$, we can treat them as constants and only have $(\tau_j)_{j \in A \cap B(i, \eps \theta)}$ as the true variables. Then, the third bullet,
    \begin{gather*}
        \sum_{j\in A} [\tau_j - (1-\delta)d(i,j)]^+ + (1-\delta) \sum_{j\in I} [d(j,S) - d(i,j)]^+ \geq \hat f\,
    \end{gather*}
    becomes a linear inequality in the true variables, since the only terms involving the true variables are $[\tau_j - (1 - \delta)d(i,j)]^+$ for $j \in A \cap B(i, \eps \theta)$, and for those $j$,
    $[\tau_j - (1 - \delta)d(i,j)]^+$ is indeed equal to $(\tau_j - (1 - \delta)d(i,j))$ without $[\cdot]^+$ because $(\tau_j - (1 - \delta)d(i,j)) \geq \theta - (1 - \delta)\eps\theta > 0$.

    Finally, the fourth bullet
    \begin{gather*}
        \sum_{j \in A} [\tau_j - d(i_0, j)]^+ + \sum_{j \in I} [\tau_\ell - 2d(j, i_0) - d(\ell, i_0)]^+ \leq \hat f
    \end{gather*}
    is equivalent to introducing variables $t_j$ for $j \in D$ and adding the following constraints:  (1) $t_j \geq \tau_j - d(i_0, j)$ for $j \in A$, (2) $t_j \geq \tau_\ell - 2d(j,i_0) - d(\ell, i_0)$ for $j \in I$, (3) $t_j \geq 0$ for $j \in D$, and (4)
    $
        \sum_{j \in A} t_j + \sum_{j \in I} t_j \leq \hat f.
    $
\end{proof}

The description of \logadaptalg is given in Algorithm~\ref{alg:log_adapt_alg}.
\begin{figure}[h!]
    \begin{center}
        \begin{minipage}{1.0\textwidth}
            \begin{mdframed}[hidealllines=true, backgroundcolor=gray!15]

                \begin{algorithm}[\logadaptalg] \ \\[0.2cm]
                    \textbf{{Initialization:}} Set $\hat f = 2f, S = \emptyset, A=D, \theta = 1$ and $\alpha_j = \theta$ for every $j\in A=D$. \\[0.1cm]

                    {Repeat the following while $A\neq \emptyset$:}

                    \begin{enumerate}
                        \item While there is an unopened facility $i$ that is openable (chosen arbitrarily if there are many),
                              \begin{quote} Calculate $(\tau_j)_{j\in A}$ so that $i$ is openable with bids $(\tau_j)_{j\in A}$. Open $i$ (i.e., add it to $S$),
                                  let $\alpha_j \leftarrow \tau_j$ for every $j \in A$ (note that only the values of the clients in $A \cap B(i, \eps \theta)$ change), mark all clients $j\in A$ such that $\alpha_j \ge (1-\delta) d(j,S)$ as served (i.e., remove them from $A$).
                              \end{quote}

                        \item Perform the following step and proceed to the next phase.
                              \begin{quote}
                                  For every remaining $j \in A$, let $\alpha_j \leftarrow \min((1 + \eps^2)\theta, (1-\delta)d(j, S))$, remove it from $A$ when $\alpha_j = (1 - \delta)d(j, S)$.  Update $\theta = (1+\eps^2)\theta$.
                              \end{quote}
                    \end{enumerate}
                    \label{alg:log_adapt_alg}
                \end{algorithm}
            \end{mdframed}
        \end{minipage}
    \end{center}
\end{figure}
It takes a parameter $\eps\in [0, 1/6]$ and $\delta = 3\eps$.
The time  $\theta$ is now increased in steps by a multiplicative factor of $(1+\eps^2)$. This is done in each iteration of the outermost loop, "Repeat the following while $A \neq \emptyset$," which we refer to as a \emph{phase} of the algorithm.
Inside a phase, we refer to the two points as \emph{stages}, so stage 1 corresponds to the while-loop where we open openable facilities, and stage 2 corresponds to the increase of $\alpha$-values of the active clients $A$.
At any point in time, the state of the algorithm is described, as aforementioned, by the tuple $(\alpha, S, A, \theta)$.
The careful reader may observe that $A$ is redundant because the active clients $A$ are determined by $S$. Indeed,  the same holds for the set of \emph{inactive} clients $I = D- A = \{j\in D: \alpha_j < (1-\delta) d(j,S) \}$. However, it will be notationally convenient to explicitly keep track of $A$ in the algorithm's state.

The parameter $\eps >0$ regulates the number of phases the algorithm will make.
Indeed, all active clients have the same $\alpha$-value, and after the completion of $p$ phases, the $\alpha$-value of active clients equals $(1+\eps^2)^p$.
This ensures that the number of phases until the algorithm terminates is $O(\log(n)/\eps^2)$.
Hence, the increase of $\alpha$-values only depends on $O(\log(n)/\eps^2)$ states of the algorithm, the states before the second stage in each of the phases. This is why we refer to the algorithm to be of logarithmic adaptivity.

Apart from impacting the number of phases, the choice of $\eps$ also impacts the approximation guarantee. Specifically, we prove the following theorem in the subsequent section. (Recall that $\delta = 3\eps$.)
\begin{theorem}
    The set $S^*$ of facilities output by \logadaptalg satisfies
    \[
        \sum_{j\in D} d(j, S^*) \leq \frac{2}{1-\delta} \cdot \left( \optlpfl(f) - f \cdot |S^*| \right)\,.
    \]
    \label{thm:log-adapt}
\end{theorem}

\subsection{Analysis of \logadaptalg}
\label{sec:log_adaptivity_analysis}

An {\em atomic step} refers to the set of operations in the indented paragraphs in \logadaptalg that change the state $(\alpha, S, A, \theta)$. (I.e., an atomic step happens when a new facility $i$ is opened and at the end of each phase.) By ``at any point of the algorithm'', we mean any point in the algorithm's execution that is not in the middle of an atomic step. Throughout the analysis, we let $\alpha^*_j$ be the final $\alpha$-value of a client $j\in D$, and we let $S^*$ be the set of opened facilities.
We start by analyzing the approximation guarantee of the algorithm with respect to  $\alpha^*$, and we then prove that $\alpha^*/2$ is a feasible solution to the dual.

\subsubsection{Approximation Guarantee.}

In the following lemma, we prove that $\sum_{j\in D} \alpha^*_j$ pays for the opening costs of the facilities in $S^*$ and the connection cost of the clients scaled by a factor $(1-\delta)$.
\begin{lemma}
    We have $\sum_{j\in D} \alpha^*_j \geq (1-\delta) \sum_{j\in D} d(j, S^*) +  \hat f\,|S^*|$.
    \label{lemma:approx_guarantee}
\end{lemma}
In the next subsection, we prove the following lemma showing that $\alpha^*/2$ is a feasible dual solution.

\begin{restatable}{lemma}{lemmadualfeasibility}
    For every facility $i$,
    $\
        \sum_{j \in D} [\alpha^*_j - 2d(i,j)]^+ \leq \hat f.
    $
    \label{lemma:dual_feasibility}
\end{restatable}

These two lemmas imply
\[
    \optlpfl(f) \geq \sum_{j\in D} \frac{\alpha^*_j}{2} \geq  \frac{1-\delta}{2}\sum_{j\in D} d(j, S^*) + f\,|S^*|,
\]
where we used that $\hat f = 2 \cdot f$. Therefore, we have an LMP $2/(1-\delta)$-approximation: the connection cost of the output solution $S^*$ satisfies
\begin{gather*}
    \sum_{j\in D} d(j,S^*) \leq \frac{2}{1-\delta} \left( \optlpfl(f) - f\,|S^*| \right)\,.
\end{gather*}
In other words, Lemma~\ref{lemma:approx_guarantee} and Lemma~\ref{lemma:dual_feasibility} imply Theorem~\ref{thm:log-adapt}.

We proceed to prove Lemma~\ref{lemma:approx_guarantee} (and prove Lemma~\ref{lemma:dual_feasibility} in the next subsection).
\begin{proof}[Proof of Lemma~\ref{lemma:approx_guarantee}]
    We prove by induction that at any point of the algorithm
    \begin{gather}
        \sum_{j\in I} \alpha_j \geq   \sum_{j\in I} (1-\delta)d(j, S)  + \hat f\,|S|.
        \label{eq:IH_approx_guarantee}
    \end{gather}
    The equality is initially true since $A= D$ (thus $I = \emptyset$) and $S = \emptyset$. We now analyze each of the two stages separately.

    \paragraph{Stage 1.} Consider what happens when we open a facility $i$, i.e., add it to $S$.
    Let $(\alpha, S, A, \theta)$ be the state right before opening $i$, $A'$ be the clients that are removed from $A$ by the opening, which means $A' = \{ j \in A : \tau_j \geq (1 - \delta)d(i, j) \}$. (Recall that the opening lets $\alpha_j \leftarrow \tau_j$ for $j \in A$.)
    Further, let $X := \{j \in I: d(j, i) < d(j, S) \}$ be the subset of clients of $I$ that make positive bids to $i$.
    The change of cost of the right-hand side of~\eqref{eq:IH_approx_guarantee} is at most
    \[
        \hat f +   \sum_{j\in A'} (1-\delta) d(i,j) + \sum_{j \in X}  (1 - \delta)(d(i,j) - d(j, S))\,.
    \]
    Since $i$ is fully paid (the third bullet of Definition~\ref{def:openable}), we also have
    \begin{align*}
        \hat f & \leq
        \sum_{j\in A'} ( \tau_j  - (1-\delta) d(i,j)    )
        +
        \sum_{j \in X}  (1-\delta)(d(j, S) -  d(i,j)) \,.
    \end{align*}
    We thus get that the change of cost of the right-hand side is at most $\sum_{j\in A'} \tau_j$, which is the change of the left-hand side.

    \paragraph{Stage 2.}
    No facility is open at this stage, and every client $j$ that becomes inactive at this stage has $\alpha_j = (1 - \delta)d(j ,S)$, which implies that the left-hand side and right-hand side both increase with the same amount when a client $j$ becomes inactive.
\end{proof}

\subsubsection{Dual Feasibility}

We start by showing that we never overbid a facility $i$. This is rather immediate in the analysis of \JMSalg that we outlined in Section~\ref{sec:basicJMS}.
While we prove this remains true for \logadaptalg, proving it requires more work, especially because of the discrete increase of $\alpha$-values in the second stage of a phase.
\begin{claim}[No Over-Bidding]\label{clm:nooverbid}
    At any point of the algorithm, for every facility $i$,
    \begin{equation}
        \sum_{j \in A} [\alpha_j - d(i,j)]^+ + \sum_{j \in I} [(1 - \delta)d(j, S) - d(j, i)]^+ \leq \hat{f}.
        \label{eq:dual-feasible-induction}
    \end{equation}
    \label{claim:dual-feasible-induction}
\end{claim}
\begin{proof}
    In the beginning, the claim is true because every client is active and $\alpha_j = 1 \leq d(i, j)$ for every $j \in D$.
    Let us prove that each step of the algorithm maintains this invariant. We first consider the steps that happen during the first stage and then those during the second stage.

    \paragraph{Stage 1.}
    Suppose that facility $i'$ is open with the bids $(\tau_j)_{j \in A}$.
    It does (1) potentially increase the $\alpha$-values of clients in $B(i', \eps \theta) \cap A$, and (2) make clients in $A \cap B(i', \theta/(1 - \delta))$ inactive, which includes $B(i', \eps \theta) \cap A$ considered in (1).

    Consider any facility $i$ and see how~\eqref{eq:dual-feasible-induction} is impacted by this change.
    First, observe that a client $j$ that was inactive before the opening of $i'$ does not increase its contribution to~\eqref{eq:dual-feasible-induction}. This is because $d(j,S)$ monotonically decreases in $S$.
    We now turn our attention to the more interesting case.
    For a client $j$ becoming inactive from active, the contribution to~\eqref{eq:dual-feasible-induction} changes from $[\alpha_j - d(i,j)]^+$ (where $\alpha$-values are measured right before the opening) to
    $[(1 - \delta)d(j, S) - d(j, i)]^+$ (where $S$ includes $i'$).

    If $\alpha_j$ does not change by the opening of $i'$, the fact that $j$ becomes inactive means $(1 - \delta)d(j, S) \leq \alpha_j$, so $j$'s contribution to~\eqref{eq:dual-feasible-induction} cannot increase.
    If $\alpha_j$ does strictly increase by the opening of $i'$, it means that $d(j, S) \leq d(j, i') \leq \eps \theta$ after the opening, so the contribution to~\eqref{eq:dual-feasible-induction} decreases by at least $\theta - (1 - \delta)\eps\theta$ or becomes $0$.

    We have thus proved that no client increases their contribution to~\eqref{eq:dual-feasible-induction} during the steps of the first stage, and so these steps maintain the inequality.

    \paragraph{Stage 2.}
    Assume towards contradiction that at the end of the phase, if we increase $\alpha_j$ of every $j \in A$ to $\min((1 + \eps^{2})\theta, (1 - \delta)d(j, S))$, the claim is violated for a non-empty subset of facilities $F'$.
    We select a ``minimal'' such counter example in the following way: let $\tau' \leq (1+\eps^2) \theta$ be the smallest value so that $\alpha'_j := \min(\tau', (1 - \delta)d(j, S))$ satisfies
    \begin{equation}
        \sum_{j \in A} [\alpha_j' - d(i,j)]^+ +  \sum_{j \in I} [(1 - \delta)d(j, S) - d(j, i)]^+ = \hat{f},
        \label{eq:dual-feasible-induction-tauprime}
    \end{equation}
    for some $i\in F'$.
    We remark that $\tau' \geq \theta$ since~\eqref{eq:dual-feasible-induction} is satisfied for all facilities after the first stage by the previous arguments.
    Furthermore,  $i\not \in S$ is not yet open since if it was open then there would be no client $j\in A$ for which $\alpha'_j-d(i,j)$ is strictly positive. Hence the increase in $\alpha$-values cannot cause $i$ to violate~\eqref{eq:dual-feasible-induction} if it were already opened, which would contradict $i\in F'$.

    We now show that this $i$  with $\tau_j = \alpha'_j$ for $j \in A \cap B(i, \eps \theta)$ and $\tau_j = \alpha_j$ for $j \in A - B(i, \eps \theta)$ satisfies the conditions of Definition~\ref{def:openable} for $i$. In other words, $i$ is an openable facility  (that was not yet opened), which contradicts the completion of the first stage.  We verify the conditions of Definition~\ref{def:openable} one-by-one. The first two bullets are satisfied by the definition of the bids $(\tau_j)_{j\in A}$.
    For the third bullet,
    \begin{align*}
                   & \sum_{j \in A}  [\tau_j - (1 - \delta)d(i,j)]^+
        + (1 - \delta)\sum_{j \in I} d(j, S) - d(i, j)]^+                                        \\
        = \quad    & \sum_{j \in A \cap B(i, \eps \theta)}  [\alpha'_j - (1 - \delta)d(i,j)]^+ +
        \sum_{j \in A - B(i, \eps \theta)}  [\alpha_j - (1 - \delta)d(i,j)]^+
        + (1 - \delta)\sum_{j \in I} [d(j, S) - d(i, j)]^+                                       \\
        \geq \quad & \sum_{j \in A \cap B(i, \eps \theta)}  [\alpha'_j - d(i,j)]^+ +
        \sum_{j \in A - B(i, \eps \theta)}  [\alpha'_j - d(i,j)]^+
        + \sum_{j \in I} [(1 - \delta)d(j, S) - d(i, j)]^+\,.
    \end{align*}
    To upper bound the second summation we used that $\alpha_j - (1 - \delta)d(i,j) \geq (1 + \eps^{2})\alpha_j - d(i, j) \geq \alpha'_j - d(i,j) $ when $d(i, j) \geq \frac{\eps^2\alpha_j}{\delta}=\frac{\eps \theta}{3}$ (recall that $\alpha_j = \theta$ for all $j\in A$), condition which is satisfied by the clients not in $B(i,\eps \theta)$.

    It remains to verify the fourth bullet of Definition~\ref{def:openable}. Suppose toward contradiction that there is a facility $i_0$ and  $\ell \in A$ such that
    \begin{gather*}
        \sum_{j\in A } [\tau_{j} - d(i_0, j)]^+ + \sum_{j\in I} [\tau_\ell  - d(i_0,\ell) - 2d(i_0, j)]^+ > \hat f.
    \end{gather*}
    Note that $\ell \in A$ implies $\tau_\ell \leq (1 - \delta)d(\ell, S)$ and for every $j \in I$,
    $
        (1 - \delta)d(\ell, S)
        \leq (1 - \delta)(d(\ell, i_0) + d(i_0, j) + d(j, S)).
    $
    Thus
    \begin{equation}\label{eqn:clm:nooverbid:stage2}
        \tau_\ell - d(i_0, \ell) - 2d(i_0, j) \leq (1-\delta) d(j,S) - d(i_0,j)\,.
    \end{equation}
    We conclude that
    \begin{align*}
             & \sum_{j \in A} [\alpha'_j - d(i_0,j)]^+ +  \sum_{j \in I} [(1 - \delta)d(j, S) - d(i_0, j)]^+ \\
        \geq & \sum_{j \in A} [\tau_j - d(i_0,j)]^+ +  \sum_{j \in I} [(1 - \delta)d(j, S) - d(i_0, j)]^+    \\ \overset{\eqref{eqn:clm:nooverbid:stage2}}{\geq}  & \sum_{j\in A} [\tau_{j} - d(i_0, j)]^+ + \sum_{j\in I} [\tau_\ell  - d(i_0,\ell) - 2d(i_0, j)]^+ > \hat f
    \end{align*}
    This implies that we could have further lowered $\tau'$ while keeping equality ~\eqref{eq:dual-feasible-induction-tauprime} true for $i_0$. This contradicts the minimality of $\tau'$.

\end{proof}

Having proved that \logadaptalg as \JMSalg maintains the invariant that there is no overbidding, we are ready to prove the dual-feasibility of $\alpha^*/2$.  (This corresponds to Lemma~\ref{lemma:JMS-dual-feasible} in the analysis of \JMSalg.)

\lemmadualfeasibility*

We remark that this implies dual feasibility of $\alpha^*/2$ by dividing both sides by $2$ (recall $\hat f = 2 \cdot f$).

The rest of the section is devoted to the proof of Lemma \ref{lemma:dual_feasibility}, and from now on, we focus on an arbitrary but fixed facility $i\in F$. Let
$D^* = \{ j \in D : \alpha^*_j > 2d(i, j) \}$ be those clients that contribute to the left-hand side of the inequality in the claim. Further, order the clients in $D^*$ according to $\alpha^*$ values by the time they are removed from $A$ (in the algorithm) and break ties according to $\alpha^*$ values (in increasing order). We will use $j, \ell$ to denote clients in $D^*$, and let us slightly abuse notation and use $j \leq \ell$ with respect to this ordering. In other words, we have $j\leq \ell$ if $j$ was removed from $A$ in an atomic step before $\ell$ was removed, or $j$ and $\ell$ were removed during the same atomic step and $\alpha^*_j \leq \alpha^*_\ell$.
We also use
the terminology that a facility $i_0$ {\em freezes} $i$ if $i_0$ becomes open when $i$ is not, and $d(i, i_0) \leq (1+3\eps)\theta/2$. (Of course, $i$ freezes itself when it is open.) This notion is important as for a facility $i$, its $D^*$ and their $\alpha^*$ values do not change after it is frozen.

The following lemma, analogous to Lemma~\ref{lemma:JMS-dual-feasible-technical} in the analysis of \JMSalg, is key to proving dual feasibility.
\begin{lemma}
    For any $\ell \in D^*$ that becomes inactive strictly before $i$ becomes frozen,
    \begin{equation}
        \sum_{j \in D^* : j \geq \ell} [\alpha^*_\ell - d(i, j)]^+ +
        \sum_{j \in D^* : j < \ell} [\alpha^*_\ell - 2d(i, j) - d(i, \ell)]^+ \leq \hat f.
        \label{eq:dual-feasible-inactive}
    \end{equation}
    \label{lemma:dual-feasible-inactive}
\end{lemma}
\begin{proof}
    Let us do the following case analysis depending on the stage $\ell$ becomes inactive.

    \paragraph{Stage 1.}
    $\ell$ becomes inactive because of the opening of some $i_0$. By the statement of the lemma, it suffices to handle the case that $i_0$ does not freeze $i$. So $d(i, i_0) > (1 + 3\eps)\theta / 2$. The opening of $i_0$ might have strictly increased the $\alpha$-value of the clients in $B(i_0, \eps \theta)$, but no such client $j$ will be in $D^*$, since they immediately become inactive while
    \[
        d(i, j) \geq d(i, i_0) - d(i_0, j) > ((1 + 3\eps) / 2 - \eps) \theta \geq (1+\eps^2) \theta/2 \geq \alpha^*_j / 2.
    \]
    Applying this argument to every facility that is open in this phase (which did not freeze $i$), we can conclude that, right before $i_0$ is open, every $j \in D^*$ has $\alpha_j \leq \theta$ and all active ones have $\alpha_j = \theta$. (So, $\alpha^*_\ell = \theta$ as well since the opening of $i_0$ does not strictly increase $\alpha_\ell$.)
    Claim~\ref{claim:dual-feasible-induction}, applied right before $i_0$ is open, ensures that
    \[
        \sum_{j \in A} [\alpha_j - d(i, j)]^+
        + \sum_{j \in I} [(1 - \delta)d(j, S) - d(i, j)]^+ \leq \hat f.
    \]
    This satisfies~\eqref{eq:dual-feasible-inactive} as (1) every $j \in D^*$ with $j \geq \ell$ is in $j \in A$ at that point by our ordering of the clients within $D^*$, and (2) for every $j < \ell$ with $j \in I$,
    \[
        \alpha^*_\ell < (1 - \delta)d(\ell, S)
        \leq (1 - \delta)d(j, S) + d(j, i) + d(i, \ell)
    \]
    implies
    \[
        \alpha^*_\ell - 2d(i, j) - d(i, \ell)
        \leq (1 - \delta)d(j, S) - d(i, j).
    \]
    (Note that this $S$ does not contain $i_0$ and $\ell$ was not connected yet.)

    \paragraph{Stage 2.}
    Suppose that $\ell$ becomes inactive by the increase of the $\alpha$-values at the end of a phase.
    Then Claim~\ref{claim:dual-feasible-induction} ensures that, at the end of this phase (after $\ell$ becomes inactive),
    \begin{equation}
        \sum_{j \in A} [\alpha_j - d(i,j)]^+ + \sum_{j \in I} [(1 - \delta)d(j, S) - d(j, i)]^+ \leq \hat{f}.
        \label{eq:dual-feasible-induction-again}
    \end{equation}

    Let us see that the left-hand side of~\eqref{eq:dual-feasible-induction-again} is at least that of~\eqref{eq:dual-feasible-inactive} by comparing the contribution of $j \in D^*$ to both. If $j \in A$, its contribution to the former is definitely at least that to the latter, since any active $j$ has $\alpha_j = (1 + \eps^{2})\theta \geq  \alpha^*_\ell$.

    For $j \in I$, we must have
    \[
        \alpha^*_\ell \leq (1-\delta) d(\ell,S)  \leq  (1-\delta) ( d(\ell,i) +d(i,j)+ d(j,S))\leq  (1-\delta)d(j,S)+d(\ell,i) +d(i,j)\,.
    \]
    (since $\alpha_\ell$ is never increased above $(1-\delta) d(\ell,S)$). Thus
    \[
        \alpha^*_\ell  - 2d(i,j) - d(i,\ell)  \leq (1-\delta) d(j,S) - d(i,j)
    \]
    which satisfies our goal when $j<\ell$.

    Finally, when $j\in I$ but $j \geq \ell$, it means that $j$ was removed at the same time as $\ell$ at the end of the phase. In that case $(1-\delta) d(j,S) = \alpha^*_j$ and, as we break ties with $\alpha^*$-values, $\alpha^*_j \geq \alpha^*_\ell$,  which again satisfies our goal.
\end{proof}

Equipped with the above lemma, we are ready to complete the proof of dual feasibility.

\begin{proof}[Proof of Lemma~\ref{lemma:dual_feasibility}]

    Consider the case that $i$ is frozen by $i_0$ with $(\tau_j)_{j \in A}$. We first prove the following claim that says that after $i$ is frozen the $\alpha$-values of clients in $D^*$ are finalized.

    \begin{claim*}
        We have $\alpha^*_j = \tau_j$ for every $j \in D^* \cap A$.
    \end{claim*}
    \begin{proof}[Proof of Claim]
        Consider $j \in A$.
        First, note that if $\tau_j \geq (1-\delta) d(j,i_0)$ then $j$ is removed from $A$ when $i_0$ is opened and so $\alpha^*_j = \tau_j$.

        In the other case, when $\alpha_j \leq \tau_j < (1-\delta) d(j,i_0)$ we show that $j \not \in D^*$.
        Consider the (future) time right after $j$ is removed from the active clients. The $\alpha$-value of $j$ then (which is equal to the final $\alpha^*_j$) cannot be strictly greater than $(1 - \delta)d(j, i_0)$; whether it is increased in stage 1 (as $\alpha \leftarrow \tau$) or stage 2, since $i_0 \in S$ (where $S$ is the set of open facilities right before $j$ becomes inactive), the algorithm ensures that it cannot be strictly more than $(1 - \delta)d(j, S) \leq (1 - \delta)d(j, i_0)$.
        By the triangle inequality
        \begin{gather*}
            d(j, i) \geq
            d(j, i_0) - d(i_0, i) =
            \left(1 - \frac{d(i_0, i)}{\theta}\frac{\theta}{d(j,i_0)}\right)d(j,i_0)\,,
        \end{gather*}
        where $\theta$ is considered at the time when $i_0$ was open. As $i_0$ freezes $i$, we have $d(i_0, i) \leq (1+3\eps)\theta/2$ and, by the assumption of the case we have $\theta = \alpha_j \leq (1-\delta) d(j,i_0)$. Plugging in those bounds to the above inequality yields,
        \begin{gather*}
            d(j,i) \geq
            \left(1 - \frac{(1+3\eps)}{2}(1-\delta)\right)d(j,i_0)\,.
        \end{gather*}
        We rewrite the above inequality (multiplying both sides by two and using that $\delta = 3\eps$) to obtain
        \begin{gather*}
            2d(j,i) \geq
            \left(2 - {(1+3\eps)}(1-\delta)\right)d(j,i_0) = \left(1+ 9\eps^2 \right) d(j,i_0) \,.
        \end{gather*}
        In other words we have $\alpha^*_j \leq d(j,i_0) < 2d(j,i)$, so  $j$ cannot be in $D^*$.
    \end{proof}

    Having proved the above claim, we proceed to analyze the dual feasibility.
    As the $\alpha$-values of clients in $A\cap D^*$ remain unchanged after the removal of $i_0$, the fourth bullet of Definition~\ref{def:openable} ensures that, for any $\ell \in A \cap D^*$,
    \begin{equation}
        \sum_{j \in A \cap D^*} [\alpha^*_j - d(i, j)]^+ +
        \sum_{j \in I \cap D^*} [\alpha^*_\ell - 2d(i, j) - d(\ell, i)]^+ \leq \hat f.
        \label{eq:openable-fourth-again-again}
    \end{equation}

    Moreover, for every client $\ell\in D^* \cap I$ that became inactive strictly before $i_0$ was opened, Lemma~\ref{lemma:dual-feasible-inactive} goves

    \begin{equation}
        \sum_{j \in D^* : j \geq \ell} [\alpha^*_\ell - d(i, j)]^+ +
        \sum_{j \in D^* : j < \ell} [\alpha^*_\ell - 2d(i, j) - d(i, \ell)]^+ \leq \hat f.
        \label{eq:dual-feasible-inactive-again}
    \end{equation}
    Let us consider the summation of~\eqref{eq:dual-feasible-inactive-again} for every $\ell \in I \cap D^*$ and~\eqref{eq:openable-fourth-again-again} for every $\ell \in A \cap D^*$,
    and consider how many times each term appears. Let $a = |A \cap D^*|, b = |D^*|$.

    \begin{itemize}
        \item $-d(i, j')$ when $j' \in I\cap D^*$: Say $j'$ is the $q$th client in $D^*$ for some $q \leq b-a$. Then $-d(i, j')$ appears $b + (b - q)$ times in~\eqref{eq:openable-fourth-again-again} ($2a$ times) and~\eqref{eq:dual-feasible-inactive-again} ($2(b-a-q)+q$ times) as $-d(i, j)$ and $(q - 1)$ times as $-d(i, \ell)$ in~\eqref{eq:dual-feasible-inactive-again} when $\ell = j'$, so the total is $2b - 1$.

        \item $-d(i, j')$ when $j' \in A\cap D^*$:
              It appears $b$ times in~\eqref{eq:openable-fourth-again-again} and~\eqref{eq:dual-feasible-inactive-again} as $-d(i, j)$ and $b - a$ times as $-d(i, \ell)$ in~\eqref{eq:dual-feasible-inactive-again} when $\ell=j'$, so the total is $2b - a$.

        \item $\alpha^*_{j'}$ for $j' \in I\cap D^*$: It appears $b$ times in~\eqref{eq:dual-feasible-inactive-again} as $\alpha^*_\ell$ when $\ell = j'$.

        \item $\alpha^*_{j'}$ for $j' \in A\cap D^*$: It appears $b -a$ times in~\eqref{eq:openable-fourth-again-again} as $\alpha^*_\ell$ when $\ell = j'$ and once in~\eqref{eq:openable-fourth-again-again} as $\alpha^*_j$ for each $\ell \in A \cap D^*$, so $b$ times in total.
    \end{itemize}
    Summarizing we have,
    \begin{gather*}
        b \sum_{j\in D^*} \alpha^*_j - (2b-a)\sum_{j\in D^* \cap A} d(i,j) - (2b-1) \sum_{j\in D^* \cap I} d(i,j) \leq b \hat f,
    \end{gather*}
    and so in particular $\sum_{j\in D^*} \left(\alpha^*_j - 2 d(i,j)\right) \leq \hat f$, which proves the lemma in the case the opening of a facility $i_0$ freezes $i$.

    To complete the proof, when no facility freezes $i$, Lemma~\ref{lemma:dual-feasible-inactive} holds for every client in $D^*$, so we can simply sum up~\eqref{eq:dual-feasible-inactive-again} over all clients in $D^*$, which corresponds to the special case of the above analysis when $D^*\cap A = \emptyset$.
\end{proof}

\subsection{Walking Between Two Solutions: Setup} %
\label{sec:robust_analysis}

Given the description of the \logadaptalg, we present our final algorithm \mergealg, resulting in a solution that opens $k+O_{\eps}(\log n)$ centers by maintaining two partial solutions that lead to opening at least $k$ and at most $k$ centers respectively and {\em gradually merge them} into a single solution that opens exactly $k+O_{\eps}(\log n)$ facilities, exploiting the $O_{\eps}(\log n)$ adaptivity of \logadaptalg.

Before detailing the merging procedure, which we call \mergealg, it is helpful to observe that we may generalize the analysis to make it more ``robust''. We introduce this more robust analysis in this subsection, where we also define the notation that will be used later to describe \mergealg in Section~\ref{sec:merging_algorithm}.

Starting from the two solutions obtained in the standard way whose facility costs differ by at most $\eta := 2^{-n}$,
the two solutions we maintain will have an exponentially small difference $\eta$ in one of their parameters.
(These parameters will be either the facility cost $f$ or some distance related to a {\em free facility} that will be introduced later.)
For this reason, when we merge the two solutions, it may occur that some facilities in one solution are almost completely paid for, rather than fully paid for. To address this, we generalize the definition of "openable" to "$\eta$-openable," where the only difference is in the third condition: we now require each facility that is opened to be paid $\hat{f} - \eta n$ instead of $\hat{f}$ (as in Definition~\ref{def:openable}).
We also note that being openable only depends on the time $\theta$ and the set $S$ of open facilities so far. This is because $\alpha_j = \theta$ for $j\in A$, $I = \{j\in D: \theta \geq  (1-\delta)d(j,S)\}$ and $A = D- I$. This allows us to simplify notation and we have the following generalization of Definition~\ref{def:openable}.

\begin{definition}
    Consider the time $\theta$ and let $S$ be the set of opened facilities.
    Then set of inactive and active clients are $I= \{j\in D: \theta \geq  (1-\delta) d(j,S) \}$ and $A = D - I$, respectively.
    We say that a facility $i\in F$ is \emph{$\eta$-openable} (with respect to $\theta$ and $S$) if there are bids $(\tau_j)_{j\in A}$ of active clients that satisfy the following conditions.
    \begin{itemize}
        \item Only nearby clients are increased:
              \begin{gather*}
                  \tau_j = \alpha_j \qquad \mbox{for every $j \in A - B(i, \eps \theta)$.}
              \end{gather*}
        \item Nearby clients are only increased slightly:
              \begin{gather*}
                  \alpha_j \leq \tau_j \leq \min\{(1-\delta) d(j,S), (1+\eps^2)\theta\} \qquad \mbox{for every $j\in A \cap B(i, \eps \theta)$.}
              \end{gather*}
        \item Facility $i$ is paid for (up to the error parameter $\eta$):
              \begin{gather*}
                  \sum_{j\in A} [\tau_j - (1-\delta)d(i,j)]^+ + (1-\delta) \sum_{j\in I} [d(j,S) - d(i,j)]^+ \geq \hat f - n \eta\,.
              \end{gather*}
        \item Dual feasibility: for every facility $i_0$ and $\ell \in A$,
              \begin{gather*}
                  \sum_{j \in A} [\tau_j - d(i_0, j)]^+ + \sum_{j \in I} [\tau_\ell - 2d(j, i_0) - d(\ell, i_0)]^+ \leq \hat f.
              \end{gather*}
    \end{itemize}
    \label{def:robust_openable}
\end{definition}

We also introduce the concept of \emph{free facilities}. During the process of merging the two solutions, we introduce free facilities, whose opening costs are not necessarily paid. While we will introduce at most $O(\log n)$ such free facilities, we do not restrict the number of free copies a regular facility may have. Let $\Sf \subseteq S$ denote the multiset of free facilities that have been opened.

Moreover, for each free copy $\tilde{i} \in \Sf$ of $i \in F$, we define a parameter $u(\tilde{i})$, which represents the distance between the freely opened facility $\tilde{i}$ and its original copy $i$. So for any point $x$ in the metric space (either a facility or a client), the distance is given by $d(\tilde{i}, x) = u(\tilde{i}) + d(i, x)$.

We further introduce the following terminology: as mentioned, we refer to the facilities in $\Sf$ as \emph{free facilities}, the facilities in $F$ and $\Sr = S - \Sf$ as \emph{regular} facilities, and when we simply mention "facilities," we are referring to their union.

Finally, we will use the letters $i$ and $h$ to denote facilities. Specifically, we use $i \in F$ to denote a regular facility and use $\tilde{i}$ to denote a free copy of $i$. When referring to a facility that can be either free or regular, we use the letter $h$.

We now introduce the definition of valid sequences, which aims to capture the sequence of facilities that are opened during the while-loop of stage 1 in \logadaptalg.
\begin{definition} [$\eta$-valid sequence]
    Consider the time $\theta$ and let $S$ be the set of opened facilities.
    A sequence $\langle h_1, \ldots, h_\ell\rangle$ of facilities is $\eta$-valid (with respect to $\theta$ and $S$), if the following conditions hold.
    \begin{itemize}
        \item Each $h_t$ is either free or $\eta$-openable with respect to the time $\theta$ and  a superset $S' \supseteq S \cup \{h_1, \ldots, h_{t-1}\}$ of opened facilities.
        \item  The sequence is maximal, i.e., there is no openable facility with respect to the time $\theta$ and  $S \cup \{h_1, \ldots, h_\ell \}$.
    \end{itemize}
    \label{def:valid_sequence}
\end{definition}

We remark that \logadaptalg generates $0$-valid sequences in each phase. Specifically, it generates sequences $\calH = (H_1, H_2, \ldots, H_L)$, where each $H_p$ is $0$-valid (with respect to $\theta = (1+\eps^2)^{p-1}$ and $S = \cup_{i \in [p-1]} H_i$) and contains only regular facilities. The set $S$ of opened facilities is such that $A = \emptyset$ at the final time $\theta$.

For a general sequence of sequences $\calH = (H_1, H_2, \ldots, H_L)$, let us say it consists of $\eta$-valid sequences if, for every $p \in [L]$, $H_p$ is an $\eta$-valid sequence with respect to $\theta = (1+\eps^2)^{p - 1}$ and $S = \cup_{q \in [p-1]} H_q$. (Here we abuse notation and treat $H_q$ as a set of facilities.)
Furthermore, we say that $\calH$ is a \emph{solution} if the set $S$ of opened facilities ensures that $A = \emptyset$ at the end, meaning that at the time $\theta = (1+\eps^2)^L$ of the $(L+1)$st phase, the set of inactive clients $I = \{ j \mid \theta \geq (1-\delta) d(j,S) \}$ equals all clients in $D$.
Otherwise, we say that $\calH$ is a partial solution.

Theorem~\ref{thm:log-adapt} can thus be stated as: if $\calH$ is a solution of $0$-valid sequences consisting of regular facilities then the set $S$ of opened facilities satisfies
\[
    \sum_{j\in D} d(j, S) \leq \frac{2}{1-\delta} \cdot \left( \optlpfl(f) - f \cdot |S| \right)\,.
\]
We now allow to open regular facilities by paying at least $\hat f- \eta n$ (instead of $\hat f$) and incorporate the notion free facilities, which leads to the following generalization of the analysis, proved in Appendix~\ref{sec:robust_analysis_full}.
\begin{theorem}
    Consider a solution $\calH$ of $\eta$-valid sequences and let $S$ be the opened facilities. Further, let $\Sf$ and  $\Sr = S - \Sf$ be the free copies and regular facilities, respectively. Then
    \[
        \sum_{j\in D} d(j, S) \leq \frac{2}{1-\delta} \cdot \left( \optlpfl(f) - (f-\eta n) \cdot |\Sr| \right)\,.
    \]
    \label{thm:robust}
\end{theorem}

It will be useful to use similar terminology as we did in Section~\ref{sec:log_adaptivity_analysis}, let us consider the {\em execution} of the solution $\mathcal{H}$ where for $p = 1, \dots, L$, we run the $p$-th phase according to the sequence $H_p$; in Stage 1, each facility in the sequence becomes open one by one with the corresponding $(\tau_j)_j$ values, and in Stage 2 at the end of the phase, $\theta \leftarrow \theta(1+\eps^2)$.

Then an {\em atomic step} can still be defined in the same way as in \logadaptalg; it corresponds to opening a facility in Stage 1 (and updating $\alpha, S, A$ accordingly) or increasing $\alpha$ values in Stage 2 simultaneously (and updating $S$ accordingly). By ``at any point of the algorithm'', we mean any point in the algorithm's execution that is not in the middle of an atomic step.

With this notation, we have the following generalization of Claim~\ref{clm:nooverbid}:
\begin{quote}
    At any point of the algorithm, for every facility $i$,
    \begin{equation}
        \sum_{j \in A} [\alpha_j - d(i,j)]^+ + \sum_{j \in I} [(1 - \delta)d(j, S) - d(j, i)]^+ \leq \hat{f}.
        \label{eq:dual-feasible-local-robust}
    \end{equation}
\end{quote}
This is formally proved in Claim~\ref{claim:dual-feasible-induction_robust} in Appendix~\ref{sec:robust_analysis_full}, and the proof is essentially the same as that of Claim~\ref{clm:nooverbid}. We use it to prove Lemma~\ref{lemma:payableimpliesopenable}, which we will use in the merging procedure to ensure maximality.
At a point of the execution, we say that a facility $i$ is \emph{payable} if it is openable (with $\eta = 0$) except that the dual feasibility condition is not necessarily satisfied. The payable notion is thus a weaker notion than openable, and a facility $i$ may be payable but not openable. However, we have the following lemma that says that if there is a payable facility, there is also an openable facility. We remark that it does not guarantee that a payable facility is also openable but it does guarantee that if there exists a payable facility, there is also an openable facility, which is sufficient for our arguments.
\begin{lemma}
    At any point of the algorithm, if there is a payable facility, then there is also an openable facility (that is not yet opened).
    \label{lemma:payableimpliesopenable}
\end{lemma}
\begin{proof}
    The argument is similar to that of the ``Stage 2'' proof of Claim~\ref{clm:nooverbid}.
    Suppose there is a non-empty subset of facilities $F'$ that are payable. To prove the lemma, we need to show that there is an openable facility.

    We choose a ``minimal'' payable facility as follows. Let $\tau' \in [\theta,(1+\eps^2)\theta]$ be the smallest value such that there exists some $i\in F'$, for which the values  $\tau_j$, defined by
    $$ \tau_j := \begin{cases} \min\{\tau', (1-\delta)d(j,S)\}, & \text{if } j \in A \cap B(i,\eps \theta),\\ \alpha_j, & \text{otherwise}, \end{cases} $$
    satisfy
    \begin{equation}
        \sum_{j \in A} [\tau_j - (1-\delta)d(i,j)]^+ +  (1 - \delta)\sum_{j \in I} [d(j, S) - d(j, i)]^+ \geq \hat{f},
        \label{eq:dual-feasible-induction-tauprime-local}
    \end{equation}
    for some $i\in F'$.
    We remark that $i\not \in S$ is not yet open since if it was open then there would be no client $j\in A$ for which $\tau_j-(1-\delta)d(i,j)$ is strictly positive and the bids of inactive clients would be $0$. (Strictly speaking, here we assume $\hat{f} >0$ since otherwise all facilities are opened at the very beginning of the algorithm and there is nothing to prove.)

    We now show that this $\tau$   satisfies the conditions of Definition~\ref{def:robust_openable} for $i$. In other words, $i$ is an openable facility  (that was not yet opened), which implies the lemma.  We verify the conditions of Definition~\ref{def:robust_openable} with $\eta =0$ one-by-one. The first two bullets are satisfied by the definition of the bids $(\tau_j)_{j\in A}$.
    The third bullet is satisfied by~\eqref{eq:dual-feasible-induction-tauprime-local}.
    Finally, we verify the dual feasibility condition.
    Suppose toward contradiction that there is a facility $i_0$ and  $\ell \in A$ such that
    \begin{gather*}
        \sum_{j\in A } [\tau_{j} - d(i_0, j)]^+ + \sum_{j\in I} [\tau_\ell  - d(i_0,\ell) - 2d(i_0, j)]^+ > \hat f.
    \end{gather*}
    Note that $\ell \in A$ implies $\tau_\ell \leq (1 - \delta)d(\ell, S)$ and for every $j \in I$,
    $
        (1 - \delta)d(\ell, S)
        \leq (1 - \delta)(d(\ell, i_0) + d(i_0, j) + d(j, S)).
    $
    Thus
    \begin{equation}\label{eqn:clm:nooverbid:payableopenable}
        \tau_\ell - d(i_0, \ell) - 2d(i_0, j) \leq (1-\delta) d(j,S) - d(i_0,j)\,.
    \end{equation}
    We conclude that
    \begin{align*}
        \sum_{j \in A} [\tau_j - d(i_0,j)]^+ & +  \sum_{j \in I} [(1 - \delta)d(j, S) - d(i_0, j)]^+ \\ &   \overset{\eqref{eqn:clm:nooverbid:payableopenable}}{\geq}   \sum_{j\in A} [\tau_{j} - d(i_0, j)]^+ + \sum_{j\in I} [\tau_\ell  - d(i_0,\ell) - 2d(i_0, j)]^+ > \hat f\,.
    \end{align*}
    We remark that this implies that $\tau' > \theta$ since otherwise, as then $\tau_j = \alpha_j =\theta$ for every $j\in A$, the above inequality would contradict~\eqref{eq:dual-feasible-local-robust}.

    We proceed to prove that the above inequality implies that we could have further lowered $\tau'$ while keeping equality ~\eqref{eq:dual-feasible-induction-tauprime-local} true for $i_0$, which  contradicts the minimality of $\tau'$ (using that $\tau'>\theta$). (It implies that $i_0$ is also payable with a smaller $\tau'$, which contradicts the minimality of $\tau'$.)
    To see this, define the values $\tau^{(i_0)}_j$ for every $j\in A$ by
    $$ \tau^{(i_0)}_j := \begin{cases} \min\{\tau', (1-\delta)d(j,S)\}, & \text{if } j \in A \cap B(i_0,\eps \theta),\\ \alpha_j, & \text{otherwise}, \end{cases} $$
    We claim
    \[
        \tau^{(i_0)}_j - (1-\delta)d(i_0,j) \geq  \tau_j - d(i_0,j) \qquad \mbox{for every $j\in A$}.
    \]
    If $j\in B(i_0, \eps \theta)$, then $\tau^{(i_0)}_j \geq \tau_j$ and the inequality is immediate. Else if $j\not \in  B(i_0, \eps \theta)$, then $\tau^{(i_0)}_j = \alpha_j = \theta$ and $\tau_j \leq \tau'\leq (1+\eps^2) \theta$, so the inequality is also satisfied in this case because $\delta d(i_0,j) \geq \delta \eps \theta = 3\eps^2 \theta \geq \eps^2 \theta$.

    By the above calculations, we have
    \begin{align*}
        \sum_{j \in A} & [\tau^{(i_0)}_j - (1-\delta)d(i_0,j)]^+ +  (1 - \delta)\sum_{j \in I} [d(j, S) - d(j, i_0)]^+ \\
                       & \geq
        \sum_{j \in A}  [\tau_j - d(i_0,j)]^+ +  (1 - \delta)\sum_{j \in I} [d(j, S) - d(j, i_0)]^+                    \\
                       & \geq
        \sum_{j \in A}  [\tau_j - d(i_0,j)]^+ +  \sum_{j \in I} [(1 - \delta)d(j, S) - d(j, i_0)]^+                    \\
                       & > \hat f\,,
    \end{align*}
    which implies that $\tau'$ could have been further lowered while keeping~\eqref{eq:dual-feasible-induction-tauprime-local} true for $i_0$, contradicting the minimality of $\tau'$. This concludes the proof that there is an openable facility, which proves the lemma.
\end{proof}

In the next subsection we will merge two solutions into one with $|\Sr| =k$ and $\Sf =O(\log n/\eps^2)$. To our help, we will have two basic routines $\completesol$ and $\completesequence$:
\begin{itemize}
    \item $\completesol_{f,u}(\calH)$ takes a partial solution $\calH$ as input  and returns a solution by running the while-loop of stage 1 in \logadaptalg using the opening cost $f$ and the distances $u$ for free facilities already opened.
    \item $\completesequence_{f,u}(\calH = (H_1, \ldots, H_{p-1}), H_p = \langle h_1, \ldots, h_\ell\rangle )$ where $\calH$ is  partial solution and $H_p$ is an $\eta$-valid sequence except that it may not be maximal. $\completesequence$ then runs the while-loop of stage $1$ in \logadaptalg starting with $H_p$ to return a maximal sequence whose prefix is $H_p$ (with respect to opening cost $f$ and distances $u$).
\end{itemize}
We note that all facilities added by \completesol and \completesequence are $0$-openable and regular. Additionally, for intuition, observe that, with the above notation, the set of facilities opened by $\completesol_{f,u}(\emptyset)$ is equivalent to that opened by \logadaptalg. (There is no free facility so $u$ is meaningless here.)
To simplify notation, we only use $\completesequence$ when the partial solution $\calH$ is clear from the context, so we simply write $\completesequence_{(f,u)}(\langle h_1, \ldots, h_\ell \rangle)$. Furthermore, we omit $(f,u)$ when the parameters are clear from the context.

\subsection{\mergealg}
\label{sec:merging_algorithm}

We now present our procedure \mergealg that {\em walks} between two solutions. We will maintain two solutions $(\calH, f, u)$ and $(\calH', f', u')$ whose numbers of open regular facilities {\em sandwich $k$}. (Formally, let us say $a$ and $b$ sandwich $k$ if $a<k<b$ or $b<k<a$.) We will gradually make them closer until one of them opens exactly $k$ regular facilities.
We start from the two initial solutions obtained using the standard method. Note that the minimum nonzero distance is $1$ and the maximum distance is $\Maxdist = \poly(n)$, which implies that $\sopt \leq n\Maxdist$.

\paragraph{Initialization.} We consider executions of $\completesol_{f, u}(\emptyset)$ with empty $u$. When $f = 1/n^2$, every client $j \in D$ satisfies $d(j, S) = 0$ if $d(j, F) = 0$ and $d(j, S) \leq d(j, F) + 2/n^2$ otherwise, as each client can open its closest facility by itself when $\alpha _j = d(j, F) + 2/n^2$. If this solution opens at most $k$ centers, then it is already an $(1+2/n)$-approximation and we return it. On the other hand, when $f = 4n\Maxdist$, the algorithm will open exactly one facility, since the first facility is open when $\theta \geq 2\Maxdist$, so every client becomes immediately inactive when it is open. We then perform a binary search in the interval $[1/n^2, 4n\Maxdist]$ of values of $f$, keeping the invariant that for the current research interval $[f,f']$ the corresponding two solutions $(\calH, f, u)$ and $(\calH', f', u)$ are such that $(\calH, f, u)$ opens less than $k$ facilities and $(\calH', f', u)$ opens more than $k$ facilities. We halt the process when $f \leq f'  \leq f + \eta$.

\paragraph{General Step.}
\mergealg works on phase $p = 1, 2, \dots$ iteratively, and ensures that the two solutions are identical up to the $p$-th phase. Formally, \mergealg takes two solutions $(\calH, f, u)$ and $(\calH', f', u')$ as well as the current phase $p$ as input. The following promises on them will be satisfied whenever we call \mergealg.
\begin{itemize}
    \item Both solutions consist of $\eta$-valid sequences.
    \item The numbers of open regular facilities by the two solutions sandwich $k$.
    \item Exactly one parameter, which we call the {\em difference parameter}, differs by at most $\eta$ between the two solutions, i.e, if $f\neq f'$ then $|f-f'| \leq \eta$ and $u(\tilde i) = u'(\tilde i)$ for every $\tilde i \in \Ff$; otherwise if $f=f'$ there is exactly one free facility $\tilde i\in \Ff$ so that $u(\tilde i) \neq u'(\tilde i)$ and $|u(\tilde i) - u'(\tilde i)| \leq \eta$.
    \item $\calH = (H_1, H_2, \ldots, H_L)$ and $\calH' = (H'_1, H'_2, \ldots, H'_{L'})$ have a common prefix of length $p-1$, i.e., $H_i = H'_i$ for $i=1,2, \ldots, p-1$. Furthermore, the remaining sequences $H_p, \ldots, H_L$ and $H'_p,\ldots, H'_{L'}$ only contain regular facilities and are $0$-valid with respect to the parameters $f, u$ and $f', u'$, respectively. (This implies that $\calH$ and $\calH'$ open the same set of free facilities.)
\end{itemize}

Given the promises, the goal of \mergealg is to
\begin{itemize}
    \item[(1)] either find a solution of $\eta$-valid sequences that opens exactly $k$ regular facilities and $O_{\eps}(\log n)$ free facilities, or
    \item[(2)] produce two new solutions that satisfy the above promises for phase $p + 1$.
\end{itemize}
Since the maximum distance between any two points is $\Maxdist \leq \poly(n)$ and $f \leq 4n\Maxdist$,
before $\theta$ reaches $6M$, at least one facility will open and every client will become inactive.
Therefore, for some threshold $p^* = O(\log n / \eps^2)$, once $\calH$ and $\calH'$ agree on the first $p^*$ phases, they open exactly the same set of facilities which violates the promise that the numbers of open regular facilities sandwich $k$. Therefore, case (1) must happen for some $p \leq p^*$, which yields a solution of $\eta$-valid sequences that opens exactly $k$ regular facilities and $O(\log n)$ free facilities.
Combined with Theorem~\ref{thm:robust}, this proves Theorem~\ref{thm:pseudo-approx}. Indeed, we then have a solution consisting of open centers $S$ with $|S| = k  + O(\log(n)/\eps^2)$ such that, by Theorem~\ref{thm:robust},
\begin{align*}
    \sum_{j\in D} d(j, S) & \leq \frac{2}{1-\delta} \cdot \left( \optlpfl(f) - (f-\eta n) \cdot k \right) \leq (2+O(\eps)) \sopt\,,
\end{align*}
where, for the last inequality, we used that $\eta = 2^{-n}$ and $\sopt\geq 1$ (since any distance is at least $1$). Theorem~\ref{thm:pseudo-approx} then follows by scaling $\eps$ appropriately.

In the remainder of this section, we present \mergealg for two solutions and a given phase $p$. It consists of three subroutines, presented in Section~\ref{sec:parameters_identical}~\ref{sec:common_prefix}, and~\ref{sec:extra_facilities} respectively. Each subroutine achieves (1) or (2) above (which ends the whole procedure or lets us move on to the next phase), or sets up the stage for the next subroutine. The final third subroutine always achieves (1) or (2).

\subsubsection{Making Parameters Identical}
\label{sec:parameters_identical}
In this subsection, we try to make the parameters of the two solutions exactly identical.
By renaming,  we assume that if $f\neq f'$  is the difference parameter then $f' < f$, and if the difference parameter is $u(\tilde i) \neq u'(\tilde i)$ then $u(\tilde i) < u'(\tilde i)$.

\begin{claim}
    $H'_1, H'_2, \ldots, H'_p$ are $\eta$-valid sequences with respect to parameters $f, u$.
    \label{claim:1ststep_etavalid}
\end{claim}
\begin{proof}
    Suppose first that $f\neq f'$ is the difference parameter.
    As the first $p-1$ sequences are identical, we have that $H'_1, \ldots, H'_{p-1}$ are $\eta$-valid with respect to parameters $f, u$. We further have that $H'_p$ is $\eta$-valid with respect to parameters $f, u$ since it was $0$-valid with respect to parameters $f', u'$ and the only difference is that $f$ is slightly bigger than $f'$ by at most $\eta$. For maximality, we have by Lemma~\ref{lemma:payableimpliesopenable} that there is no payable facilities after opening $H'_p$ with respect to opening cost $f'$ and increasing the opening cost $f$ cannot introduce any new payable facilities, so $H'_p$ is also maximal with respect to $f$.

    Now assume that $u(\tilde i) \neq u'(\tilde i)$ is the difference parameter. Then by the same arguments as above, we have $H'_1, \ldots, H'_{p-1}$ are $\eta$-valid with respect to $f,u$. Moreover, $H'_p$ was $0$-valid with respect to $f',u'$ and we have that $u'(\tilde i) - \eta \leq u(\tilde i) \leq u'(\tilde i)$.
    It remains to show that $H'_p$ is $\eta$-valid with respect to $f,u$ in this case. The maximility of $H'_p$ with respect to $f,u$  is implied by Lemma~\ref{lemma:payableimpliesopenable} by the same arguments as above when $f\neq f'$: since lowering $u(\tilde i) \leq u'(\tilde i)$ only decreases the payments toward opening facilities and, hence, it  does not introduce any new payable facilities.
    We proceed to show that the facilities remain $\eta$-openable.

    Let $H'_p = \langle h_1, \dots, h_{\ell} \rangle$ and consider $h_t$ for some $t \in [\ell]$, which is $0$-openable with respect to $f', u', \theta = (1 + \eps^2)^{p - 1}$ and some superset $S' \supseteq (\cup_{q \in [p-1]} H_q) \cup \{ h_1, \dots, h_{t-1} \}$. When we change from $f', u'$ to $f, u$, since the only difference is $u(\tilde i) < u'(\tilde i)$, the set of active clients $A'$ for $f', u'$ is a superset of the set of active clients $A$ for $f, u$. Also, let $d'(\cdot, \cdot)$ and $d(\cdot, \cdot)$ be the distances with respect to $u'$and $u$ respectively.
    Based on the fact that $h_t$ was $0$-openable with respect to the former setting with bids $(\tau'_j)_{j \in A'}$, let
    $(\tau_j)_{j \in A}$ be the new bids defined as $\tau_j = \min(\tau'_j, (1-\delta)d(j, S))$ for $j \in A$.
    Let us show that $h_t$ was $\eta$-openable with respect to the latter setting
    with bids $(\tau_j)_{j \in A}$ by checking the conditions in Definition~\ref{def:robust_openable}:

    \begin{itemize}
        \item $\tau_j = \alpha_j$ for every $j \in A - B(h_t, \eps \theta)$: Since it was true for $A' \supseteq A$, it remains true.
        \item $\alpha_j \leq \tau_j \leq \min\{(1-\delta)d(j,S),(1+\eps^2)\theta\}$ for every $j \in A \cap B(h_t, \eps \theta)$: Fix $j \in A \cap B(h_t, \eps \theta)$.
              Since $A \subseteq A'$, $\tau'_j \leq \min\{(1-\delta)d'(j,S),(1+\eps^2)\theta\}$. Since $\tau_j$ is defined as $\min(\tau'_j, (1-\delta)d(j,S))$, the condition is satisfied. (Note that $j \in A$ implies $\alpha_j \leq \min(\tau'_j, (1-\delta)d(j, S))$ as well.)
        \item $
                  \sum_{j \in A} [\tau_j - (1-\delta)d(h_t,j)]^+ + (1-\delta)\sum_{j \in I} [d(j,S)-d(h_t,j)]^+ \geq \hat f - n \eta$: Let us argue that the contribution of each client $j$ to the left-hand side changes by at most $\eta$ between $(A', d', \tau')$ and $(A, d, \tau)$.
              For $j \in A$, note that $\tau_j < \tau'_j$ implies $\tau_j = (1-\delta)d(j, \tilde i) \geq (1-\delta)(d'(j, \tilde i) - \eta) \geq \tau'_j - \eta$, so $\tau_j \geq \tau'_j - \eta$ in any case. For $j \in I'$, note that $d(j, S)$ and $d'(j, S)$ only differ by at most $\eta$. Finally for $j \in A' \cap I$, $(1-\delta)d(j, S) \geq \tau_j \geq \tau'_j - \eta$. Therefore, the fact that this condition was satisfied for $(A', d', \tau')$ with $\hat f$ on the right-hand side implies  that $(A, d, \tau)$ satisfies this condition with $\hat f - n \eta$ on the right-hand side.

        \item For every regular facility $i_0$ and $\ell \in A$, $
                  \sum_{j \in A} [\tau_j - d(i_0,j)]^+ + \sum_{j \in I} [\tau_\ell - 2d(j,i_0) - d(\ell, i_0)]^+\leq \hat f$: Since $\ell \in A$ implies $\ell \in A'$, we have
              \[
                  \sum_{j \in A'} [\tau'_j - d(i_0,j)]^+ + \sum_{j \in I'} [\tau'_\ell - 2d(j,i_0) - d(\ell, i_0)]^+
                  \leq \hat f.
              \]
              (Since $i_0$ is a regular facility, $d(i_0, j)$ for a client $j$ does not change by changing $u(\tilde i)$.)
              For $j \in A' \cap I$, since $j \in I$ and $\ell \in A$, one has
              \[
                  \tau'_j \geq \alpha_j \geq d(j, S) \geq d(\ell, S) - d(\ell, j)
                  \geq \tau_\ell - d(j, i_0) - d(\ell, i_0).
              \]
              Thus $\tau'_j - d(i_0, j) \geq \tau_\ell - 2d(j, i_0) - d(\ell, i_0)$. Therefore we have
              \[
                  \sum_{j \in A} [\tau_j - d(i_0,j)]^+ + \sum_{j \in I} [\tau_\ell - 2d(j,i_0) - d(\ell, i_0)]^+
                  \leq \hat f.
              \]
    \end{itemize}

\end{proof}

\paragraph{Outcome of this section of \mergealg.}
Let
\[
    \calH'' =  \completesol_{(f,u)}(H'_1, \ldots, H'_p)\,,
\]
which is a valid call by Claim \ref{claim:1ststep_etavalid}.
If $\calH''$ opens $k$ regular facilities, we return this solution. Otherwise,
as $(\calH, f,u)$ and $(\calH', f', u')$ sandwich $k$, we must have that either $(\calH, f, u)$ and $(\calH'', f, u)$, or $(\calH'', f, u)$ and $(\calH', f', u')$ sandwich $k$.
On the one hand, if $(\calH'', f, u)$ and $(\calH', f',u')$ sandwich $k$, then we made progress in that we now have two solutions with a common prefix of length $p$ (instead of $p-1$) and only one parameter differs by at most $\eta$. In this case, we are done with phase $p$ and can continue \mergealg with these two solutions and phase $p+1$.

On the other hand, if $(\calH, f,u)$ and $(\calH'', f, u)$ sandwich $k$, then these two solutions have the same parameters and their first $\eta$-valid sequences that differ are $H_p$ and $H'_p$. In this case, we rename $\calH''$ to $\calH'$ and proceed to the next step.

\subsubsection{Maximize Common Prefix of $H_p$ and $H'_p$}
\label{sec:common_prefix}
We are given two solutions $(\calH, f,u)$ and $(\calH', f, u)$ on the same parameters $f,u$ that sandwich $k$. Moreover, if we let $\calH = (H_1, \ldots, H_L)$ and $\calH' = (H'_1, \ldots, H'_{L'})$, then $H_i = H'_i$ for $i= 1, \ldots, p-1$ and $H_p \neq H'_p$. In addition $H_p$ and $H'_p$ only contain regular facilities.

In this step of \mergealg, we increase the common prefix of  $H_p$ and $H'_p$ by ``slowly'' modifying $H'_p$. Recall that $H_p$ only contains regular facilities: let $H_p = \langle i_1, i_2, \ldots, i_\ell \rangle$. Moreover, as the parameters $(f,u)$ are clear from the context, we omit them in the following when, e.g., calling the procedures \completesol and \completesequence. %

\begin{mdframed}[hidealllines=true, backgroundcolor=gray!15]
    \vspace{-5mm}
    \paragraph{Transforming $H'_p$  step-by-step so that $H_p$ becomes a prefix of $H'_p$.}\ \\

    \noindent Repeat the following until all of $H_p$ becomes a prefix of $H'_p$.\\

    \begin{itemize}
        \item Let $q$ be the largest index so that $H'_p = \langle i_1, i_2, \ldots, i_q, h'_{q+1}, \ldots, h'_{\ell_1}\rangle$, i.e., $q$ equals the length of the common prefix $i_1, i_2, \ldots, i_q$ of $H_p$ and  $H'_p$.
        \item Update $H'_p$ by inserting a free copy of $i_{q+1}$, $\tilde i_{q+1}$, at the end of $H'_p$ with $u(\tilde i_{q+1}) = 0$. That is, we get the sequence
              \begin{gather*}
                  H'_p = \langle i_1, i_2, \ldots, i_q, h'_{q+1}, \ldots, h'_{\ell_1}, \tilde{i}_{q+1} \rangle,.
              \end{gather*}

        \item
              For $t = q+1 , \ldots, \ell_1$,
              \begin{itemize}
                  \item  Update $H'_p$ by first removing $h'_t$ and then completing the sequence by a call to \completesequence. More explicitly, if
                        \begin{align*}
                            H'_p =  & \langle i_1, \ldots, i_q, h'_t, h'_{t+1} \ldots, h'_{\ell_1}, \tilde i_{q+1}, h'_{\ell_1+1}, \ldots, h'_{\ell_2}\rangle                    \\
                            \intertext{then we delete $h'_t$ to obtain}
                            H''_p = & \langle i_1, \ldots, i_q,  h'_{t+1} \ldots, h'_{\ell_1}, \tilde i_{q+1}, h'_{\ell_1+1}, \ldots, h'_{\ell_2}\rangle                         \\
                            \intertext{and finally update $H'_p = \completesequence(H''_p)$ which will be of the form}
                                    & \langle i_1, \ldots, i_q, h'_{t+1}, \ldots, h'_{\ell_1}, \tilde i_{q+1}, h'_{\ell_1+1}, \ldots, h'_{\ell_2}, \ldots, h'_{\ell_3}\rangle\,,
                        \end{align*}
                        where  $h'_{\ell_2+1},\ldots, h'_{\ell_3}$ are the facilities  added by \completesequence (which could be empty).
              \end{itemize}

        \item After removing $h'_{q+1}, \ldots, h'_{\ell_1}$, we have
              \begin{align*}
                  H'_p & = \langle i_1, i_2, \ldots, i_q, \tilde{i}_{q+1}, h'_{x}, \ldots, h'_{y}\rangle\,. \\
                  \intertext{We update  $H'_p$ by replacing $\tilde{i}_{q+1}$ with its regular copy $i_{q+1}$, i.e., we obtain}
                  H'_p & = \langle i_1, i_2, \ldots, i_q, {i}_{q+1}, h'_{x}, \ldots, h'_{y}\rangle\,.
              \end{align*}
    \end{itemize}
\end{mdframed}
We have the following claims.

\begin{claim}
    Throughout the updates to $H'_p$, $H'_1, \ldots, H'_p$ remains an  $\eta$-valid sequence, and $H'_p$ contains at most one free facility.
    \label{claim:2ndstep_valid_free}
\end{claim}
\begin{proof}
    As we never update $f$ or the value $u(\cdot)$ of a free facility in $H'_1, \ldots, H'_{p-1}$, these sequences remain $\eta$-valid.

    We continue to argue that $H'_p$ remains $\eta$-valid throughout its updates. By Claim~\ref{claim:1ststep_etavalid}, this is true before any updates to $H'_p$. We now analyze the three types of updates to $H'_p$.
    \begin{itemize}
        \item Updating  $H'_p$ by inserting a free copy $\tilde{i}_{q+1}$ of $i_{q+1}$ maintains that it is $\eta$-valid as every regular facility stays $\eta$-openable in the sequence (they have the same prefix as before) and the sequence remains maximal (by Lemma~\ref{lemma:payableimpliesopenable} as adding a free facility does not make any new facility payable).
        \item Updating $H'_p$ by removing $h'_t$ and calling \completesequence.  Removing a $h'_t$ can only decrease the prefix of regular facilities in the sequence, and so each regular facility still has a (superset) of facilities for which it is $\eta$-openable. Finally, the call \completesequence ensures that the sequence is maximal.
        \item Updating $H'_p$ by replacing $\tilde i_{q+1}$ with its regular copy $i_{q+1}$. First note that, since $u(\tilde i_{q+1}) = 0$, this does not impact on the following regular facilities being $\eta$-openable (with respect to a superset). By the same argument, the sequence stays maximal (as $\tilde{i}_{q+1}$ with $u(\tilde i_{q+1}) = 0$ and  $i_{q+1}$ has the same impact on following facilities).  Finally, to see that $i_{q+1}$ is $\eta$-openable  (with respect to a potential superset) notice that the prefix $i_1, \ldots, i_q$ is the same as that in $H_p$ which is an $\eta$-valid sequence.
    \end{itemize}

    We complete the proof of the claim by arguing that $H'_p$ contains at most one free facility at any point of time. This follows from the facts that (1) $H'_p$ initially contains no free facility, (2) in one repetition, the free facility $\tilde i_{q+1}$ is inserted into $H'_p$, and (3) $\tilde i_{q+1}$ is replaced by its regular copy $i_{q+1}$ before continuing to the next repetition.
\end{proof}

\newcommand{\Hbefore}{\ensuremath{H^{\text{before}}_p}\xspace}
\newcommand{\Hafter}{\ensuremath{H^{\text{after}}_p}\xspace}

\paragraph{Outcome of this section of \mergealg.} We distinguish three cases.

\paragraph{Case 1.} If at some point $H'_p$ is such that $\completesol(H'_1, H'_2, \ldots, H'_p)$ opens exactly $k$ regular facilities, then we return that solution.

\paragraph{Case 2.} Otherwise,  suppose there is an update to $H'_p$ so that if we let \Hbefore and \Hafter denote $H'_p$ before and after this update, respectively, then \completesolx$(H_1, H_2, \ldots, \Hbefore)$ and \completesolx$(H_1, H_2, \ldots, \Hafter)$ sandwich $k$. There are three kinds of updates to $H'_p$. However, note that the update where $\tilde i_{q+1}$ is replaced by its regular copy $i_{q+1}$ cannot satisfy the conditions of this case. Indeed, as $u(\tilde i_{q+1}) = 0$, we have
\[
    \completesol(H_1, \ldots, \Hbefore) = \completesol(H_1, \ldots, \Hafter)
\]
in this case, which contradicts that they sandwich $k$.

We proceed to analyze the two other kind of updates.
If the update to $H'_p$ is to insert a free facility $\tilde{i}_{q+1}$ with $u(\tilde{i}_{q+1}) = 0$. Then
\begin{align*}
    \Hbefore & = \langle i_1, i_2, \ldots, i_q, h'_{q+1}, \ldots, h'_{\ell_1} \rangle                  \\
    \Hafter  & = \langle i_1, i_2, \ldots, i_q, h'_{q+1}, \ldots, h'_{\ell_1}, \tilde{i}_{q+1} \rangle
\end{align*}
Note that the solution \completesolx$(H_1, H_2, \ldots, \Hbefore)$ is equivalent to \completesolx$(H_1, H_2, \ldots, \Hafter)$ if we set $u(\tilde{i}_{q+1}) = 10\Maxdist$ (except for the free facility $\tilde i_{q+1}$). We can thus do binary search on $u(\tilde{i}_{q+1})$ to find two values $u'(\tilde{i}_{q+1})$ and $u(\tilde{i}_{q+1})$ with $|u'(\tilde{i}_{q+1}) - u(\tilde{i}_{q+1})| \leq \eta$ so that the solutions
$\completesol_{(f,u')}(H_1, \ldots, \Hafter)$ and $\completesol_{(f,u)}(H_1, \ldots, \Hafter)$  sandwich $k$ (where $u$ and $u'$ are identical except for $\tilde{i}_{q+1}$).
We have thus finished phase $p$, and we can restart \mergealg with the two solutions and phase $p+1$, since they have the first $p$, instead of $p-1$, $\eta$-valid sequences in common and exactly one difference parameter.
Their maximality follows from that of $\Hbefore$.

Finally, consider when the update to $H'_p$ is of the second type, i.e., $u(\tilde{i}_{q+1}) = 0$ and
\begin{align*}
    \Hbefore  =
              & \langle i_1, \ldots, i_q, h'_t, h'_{t+1} \ldots, h'_{\ell_1}, \tilde i_{q+1}, h'_{\ell_1+1}, \ldots, h'_{\ell_2}\rangle                 \\
    \Hafter = & \langle i_1, \ldots, i_q, h'_{t+1}, \ldots, h'_{\ell_1}, \tilde i_{q+1}, h'_{\ell_1+1}, \ldots, h'_{\ell_2}, \ldots, h'_{\ell_3}\rangle
\end{align*}
In this case, $\completesol(H_1, \ldots, \Hbefore)$ and $\completesol(H_1, \ldots, \Hafter)$ are two solutions with the same parameters $f, u$, that sandwich $k$ and we proceed to the next step with the following properties of $\Hbefore$ and $\Hafter$:
\begin{itemize}
    \item Both sequences contain one free facility, $\tilde{i}_{q+1}$.
    \item $\Hbefore$ and $\Hafter$ are identical except that (1) $\Hbefore$ contains a regular facility $h'_t$ not present in $\Hafter$, and (2) $\Hafter$ contains a (potentially empty) suffix of regular facilities not present in $\Hbefore$.
\end{itemize}

\paragraph{Case 3.} None of the previous cases apply. Then,  completing the solution with the initial $H'_p$ and the final $H'_p$ cannot sandwich $k$. Since then (assuming the first case does not apply) there must be an update to $H'_p$ that results in two solutions that sandwich $k$: this happens because the completion of $H_p$ and the completion of the initial $H'_p$ sandwich $k$. The above implies that if we let $\Hbefore = H_p$ and $\Hafter$ equal the final $H'_p$ then \completesolx$(H_1, \ldots, \Hbefore)$ and \completesolx$(H_1, \ldots, \Hafter)$ sandwich $k$. Further note that, since $\Hafter$ equals $H'_p$ after the final updates, $\Hbefore$ appears as a prefix of $\Hafter$ and both sequences only contain regular facilities.   To summarize,
in this case, \completesolx$(H_1, \ldots, \Hbefore)$ and $\completesol(H_1, \ldots, \Hafter)$ are two solutions with the same parameters $f, u$, that sandwich $k$ and we proceed to the next step with the following properties of $\Hbefore$ and $\Hafter$:
\begin{itemize}
    \item Both sequences contain no free facilities.
    \item $\Hbefore$ and $\Hafter$ are identical except that $\Hafter$ contains a suffix of regular facilities not present in $\Hbefore$.
\end{itemize}

\subsubsection{Removing Extra Facilities}
\label{sec:extra_facilities}
The input to this step consists of two solutions $\completesol(H_1, \ldots, \Hbefore)$ and \completesolx$(H_1,\ldots, \Hafter)$ that have the same parameters $f, u$ and sandwich $k$. Moreover,  $\Hbefore$ and $\Hafter$ both contain at most one free facility and they are identical except for potentially two differences:
\begin{itemize}
    \item $\Hbefore$ may contain a facility not present in $\Hafter$;
    \item $\Hafter$ contains a suffix of regular facilities (which may be empty) that is not present in $\Hbefore$.
\end{itemize}

To simplify the notation in this section, we let $i_*$ denote the regular facility in $\Hbefore$ that is not present in $\Hafter$ if it exists (so $i_*$ corresponds to $h'_t$ in the previous section). We further let $h_1, \ldots, h_\ell$ denote the suffix of $\Hafter$ not present in $\Hbefore$.

We now transform $\Hafter$ into $\Hbefore$ one change at a time:

\begin{mdframed}[hidealllines=true, backgroundcolor=gray!15]
    \vspace{-5mm}
    \paragraph{Generation of $\eta$-valid sequences from $\Hafter$ to $\Hbefore$.}\ \\
    \begin{itemize}
        \item If $\Hbefore$ contains the regular facility $i_*$ not present in $\Hafter$, then let $\tilde{i}_*$ be a free copy of $i_*$ with $u(\tilde{i}_*) =0$ and obtain $H^{0}_p$ from $\Hafter$ by adding $\tilde{i}_*$ at the end. Otherwise, we let $H^0_p$ equal $\Hafter$.
        \item For $r= 1, \ldots, \ell$, obtain $H^r_p$ from $H^{r-1}_p$ by removing $h_r$.
    \end{itemize}
\end{mdframed}
We remark that $H^\ell_p$ is identical to $\Hbefore$ since we removed the whole suffix $h_1, \ldots h_\ell$ present in $\Hafter$, except for potentially the facility $i_*$. That is, the only difference in $\Hbefore$ and $H^\ell_p$ may be the placement of $\tilde{i}_*$ in the sequence and the fact that $\tilde{i}_*$ is a free copy of the regular facility $i_*$ present in $\Hbefore$. However, since $u(\tilde{i}_*) = 0$, we have that the two sequences $\Hbefore$ and $H^\ell_p$ have an identical impact on \completesol, which is a fact we will use in the proof of Claim~\ref{claim:3rdstep_outcome}.

\begin{claim}
    For $r= 0, 1 \ldots \ell$, $H^r_p$ is an $\eta$-valid sequence with at most two free facilities.
\end{claim}
\begin{proof}
    The sequence $H^r_p$ has at most two free facilities because  $\Hafter$  has at most one free facility, and we potentially add one new free facility $\tilde i_*$.

    We now verify that $H^r_p$ is $\eta$-valid. For any regular facility $h'$ in the sequence, the prefix of $H^r_p$ before $h'$ is a subset of the prefix before $h'$ in $\Hafter$. Thus, as $h'$ was openable with respect to a superset $S(h')$ in \Hafter, $h'$ must be $\eta$-openable with respect to the same superset $S(h')$ in $H^r_p$. This ensures that the first property of Definition~\ref{def:valid_sequence} is satisfied. To verify maximality, notice that since we added $\tilde i_t$ with $u(\tilde i_t) =0$ in the end, we have that the set of active clients after opening the facilities in $H^r_p$ is a subset of the active clients after opening the facilities in $\Hbefore$. Hence, the maximality of $\Hbefore$ implies the maximality of $H^r_p$ (by Lemma~\ref{lemma:payableimpliesopenable} as any payable facility in $H^r_p$ is also payable in $\Hbefore$).
\end{proof}

\begin{claim}
    One of the following is true.
    \begin{itemize}
        \item  For one of the sequences $H^r_p$,  \completesolx$(H_1, H_2, \ldots, H^r_p)$ opens $k$ regular facilities.
        \item  $\completesol(H_1, H_2, \ldots, \Hafter)$ and $\completesol(H_1, H_2, \ldots, H^{0}_p)$ sandwich $k$.
        \item  \completesolx$(H_1, H_2, \ldots, H^r_p)$ and \completesolx$(H_1, H_2, \ldots, H^{r+1}_p)$ sandwich $k$ for some  $r\in \{0, 1, \ldots, \ell-1\}$.
    \end{itemize}
    \label{claim:3rdstep_outcome}
\end{claim}
\begin{proof}
    Assume the first bullet point does not hold, i.e., the completion of no $H^r_p$ leads to the opening of $k$ regular facilities. Now notice that the second and third bullet points ask whether consecutive solutions sandwich $k$, starting with $\Hafter$ and ending with $H^{\ell}_p$.
    By definition, $H^{\ell}_p$
    equals  $\Hbefore$ except  potentially for the placement of $\tilde{i}_*$ and that $\tilde{i}_*$ is the free copy of $i_*$ present in $\Hbefore$. However, as  $u(\tilde{i}_t) = 0$, the two sequences have an identical impact when we complete the solution, which implies
    \[
        \completesol(H_1, \ldots, H^{\ell}_p) = \completesol(H_1, \ldots, \Hbefore).
    \]
    Thus, as $\completesol(H_1, \ldots, \Hafter)$ and $\completesol(H_1, \ldots, \Hbefore)$ sandwich $k$, either the second or third bullet point of the statement must be valid (since no solution opened exactly $k$ regular facilities).
\end{proof}

\paragraph{Outcome of this section of \mergealg.} We distinguish three cases based on the first bullet point that is true in Claim \ref{claim:3rdstep_outcome}.

\paragraph{Case 1.} If one of the sequences $H^r_p$ is such that $\completesol(H_1, H_2, \ldots, H^r_p)$ opens  $k$ regular facilities, then we return that solution.

\paragraph{Case 2.} Otherwise, if the solutions $\completesol(H_1, H_2, \ldots, \Hafter)$ and \completesolx$(H_1, H_2, \ldots, H^{0}_p)$ sandwich $k$, then $H^0_p$ is equal to $\Hafter$ except that $\tilde{i}_p$ was added in the end of $\Hafter$ with $u(\tilde{i}_p) = 0$.
Note that the solution \completesolx$(H_1, H_2, \ldots, \Hafter)$ is equivalent to \completesolx$(H_1, H_2, \ldots, H^0_p)$ if we set $u(\tilde{i}_{q+1}) = 10\Maxdist$ (except the presence of $\tilde{i}_*$). We can thus do binary search on $u(\tilde{i}_*)$ to find two values $u'(\tilde{i}_*)$ and $u(\tilde{i}_*)$ with $|u'(\tilde{i}_*) - u(\tilde{i}_*)| \leq \eta$ so that the two solutions
$\completesol_{(f,u')}(H_1, \ldots, H^0_p)$ and $\completesol_{(f,u)}(H_1, \ldots, H^0_p)$ sandwich $k$  (where $u$ and $u'$ are identical except for $\tilde{i}_{q+1}$).
We have thus finished phase $p$, and we can restart \mergealg with the two solutions and phase $p+1$, since they have the first $p$, instead of $p-1$, $\eta$-valid sequences in common and exactly one difference parameter.

\paragraph{Case 3.} Finally, consider the case when none of the above cases apply and thus (by Claim \ref{claim:3rdstep_outcome})   there is an  $r\in \{0, 1, \ldots, \ell-1\}$ so that \completesolx$(H_1, H_2, \ldots, H^r_p)$ and\\ \completesolx$(H_1, H_2, \ldots, H^{r+1}_p)$ sandwich $k$.
The only difference between $H^r_p$ and $H^{r+1}_p$ is that $H^r_p$ contains $h_{r+1}$ which is not present in $H^{r+1}$. Let $H_p$ be the sequence obtained by adding a free copy $\tilde i_{r+1}$ of $h_{r+1}$ at the end of $H^{r+1}$ (which clearly maintains that is $\eta$-valid), then
\begin{align*}
    \completesolx(H_1, H_2, \ldots, H_p) & = \completesolx(H_1, H_2, \ldots, H^r_p) \mbox{ if $u(\tilde i_{r+1}) = 0$}               \\
    \intertext{and}
    \completesolx(H_1, H_2, \ldots, H_p) & = \completesolx(H_1, H_2, \ldots, H^{r+1}_p) \mbox{ if $u(\tilde i_{r+1}) = 10\Maxdist$.}
\end{align*}
Hence, doing a binary search on $u(\tilde{i}_{r+1})$ as in the previous case finishes phase $p$. We can restart \mergealg with the two solutions and phase $p+1$, since they have the first $p$, instead of $p-1$, $\eta$-valid sequences in common and exactly one difference parameter.

\newcommand{\costinc}{\frac{\eps}{\log k}}

\newcommand{\invinc}{\frac{\log \log^2 k}{\eps}}

\newcommand{\tol}{\frac{\eps}{\log^{3/2} k}}

\newcommand{\rad}{\text{Rad}}
\newcommand{\oK}{\overline{K}}
\newcommand{\avg}{\text{avg}}

\newcommand{\E}{\text{E}}
\newcommand{\ho}{\hat{\opt}}
\newcommand{\Sdummy}{S_{\text{dummy}}}

\section{$(2+\eps)$-Approximation for  $(\frac{\zeta}{\log n})$-{Stable} Instances}
\label{sec:centerremoval}
In this section, we give a $(2+\eps)$-approximation algorithm for 
$(\frac{\zeta}{\log n})$-stable $k$-median instances. The proof of the theorem
is presented in Section~\ref{sec:stable:everythingtogether}.
We restate the theorem for the reader's convenience.
\stableapprox*

To simplify notation, we let $\eps$ be the minimum of $\eps$ and $\zeta$ in the above theorem and further assume that $\eps<1/12$ is sufficiently small. With this notation, we present a $(2+O(\eps))$-approximation algorithm for $\eps/\log(n)$-stable instances. This implies the above theorem as any $\zeta/\log(n)$-stable instance is also $\zeta'/\log(n)$-stable with $\zeta' \leq \zeta$.

The algorithm and the proof of the above theorem follow several steps where we iteratively simplify the instance by adding more structure. Finally, we obtain enough structure to solve the problem by maximizing a submodular function subject to a matroid partition constraint. We have not optimized constants in favor of simplifying the description, and throughout this section we let $O_\eps(\log n)$ denote $O(\log(n)/\eps^{O(1)})$. An overview of this section and the algorithm is as follows:
\begin{mdframed}[hidealllines=true, backgroundcolor=gray!15]
\vspace{-5mm}
\paragraph{Overview of $2$-Approximation for $k$-median on $(\eps/\log n)$-
Stable Instance $k, \clients, \facilities, \dist$.}\ \\
\begin{enumerate}
    \item We start, in Section~\ref{sec:localSearchAnalysis}, by running a standard Local Search on $(k, \clients, \facilities, \dist)$ to obtain a set of centers $S$. It is well-known that $S$ is a constant-factor approximation. Moreover, for a stable instance, we show that $S$ ``correctly'' implicitly identifies all but $O_\eps(\log n)$ centers of $\opt$. Specifically, this ensures that all but $O_\eps(\log n)$ clusters have the same cost as in $\opt$ up to a $(1+O(\eps))$ factor. 
    \item In Section~\ref{sec:dsampleproc}, we use a technique inspired from the 
    classic $k$-means++ algorithm and its
    so called $D^2$-sampling procedure to
    sample clients proportional to their cost in solution $S$, and we show that if we take  $O_\eps(\log n)$ samples $W$ then with good probability we ``hit'' all but an insignificant fraction of clusters of high cost in $S$. 
    \item In Section~\ref{sec:ballguesses}, we then, guess a subset of the  sampled clients $W$ so that each hit cluster has exactly one client, which we refer to as the leader of that cluster. In addition, for each client $\ell$ that is a leader, we guess the (approximate) distance from $\ell$ to the optimal center in the optimal cluster it hits. This defines a set $\calB$ of balls with the guarantee that the optimal solution opens one center in each of these balls. Moreover, if we let $\ho$ be the optimal centers corresponding to these balls, we prove in Section~\ref{sec:dummycenters} that there is a swap of centers $M_O = S-S_0  \cup \ho$ that removes $S_0$ and adds the optimal centers $\ho$ so that the cost of the solution $M_O$  is good and in particular a $2+O(\eps)$-approximation.  
    \item The remaining steps are then focused on approximating this swap, i.e., intuitively to that of finding $S_0$ and $\ho$:  
   \begin{enumerate}
    \item We first find the set $\calQ$ of expensive centers of $S_0$ so that the centers $S_0-\calQ$ has total cost $O(\eps \sopt)$ in Section~\ref{sec:removalofExpensive}. This allows us to think of the centers in $S_0 - \calQ$ as contracted, i.e., as single weighted points, and in particular, that all clients of these clusters are assigned to the same center when their center in $S_0$ is removed. 
     We then, in Section~\ref{sec:consistentlyassigning}, use this structure to partition $S_0 - \calQ$ into two sets $\calU$ and $\calR$ where  $\calU$ are those centers that are reassigned to other centers in $S - S_0$ and $\calR$ are those centers that are assigned to $\ho$. 

    \item The set $\calU$ is then approximated with a set $\bcalU$ (with close to identical properties) in Section~\ref{sec:removalofCheap}.
    \item At this stage, we have thus ``guessed'' all centers $\calQ \cup \bcalU$ to remove from $S$ except for those in $\calR$. This is achieved in Section~\ref{sec:submodularopt} where we reduce the problem of both finding $\calR$ and an approximate version of $\ho$ to that of submodular function optimization subject to a partition matroid (since we should open one center per ball in $\calB$). 
    \end{enumerate}
\end{enumerate}

\end{mdframed}

We remark that in each of Steps 3, 4a, 4b, and 4c, we will make guesses by enumerating polynomially $n^{\eps^{-O(1)}}$ many potential choices.  In the following, we describe the algorithm and its analysis by presenting each of these steps. When presenting the next step, we analyze the branch of the algorithm that took the successful guesses up to that point. The algorithm outputs the solution of the smallest cost among the $n^{\eps^{-O(1)}}$ constructed solutions obtained by enumerating all possible guesses (see Section~\ref{sec:stable:everythingtogether} for a more formal treatment where we put everything together).
 In particular, the final solution will have cost at most the solution that is output in the analyzed branch of ``correct'' guesses. 

We end this section by introducing basic definitions and concepts, followed by the definition of matroids and submodular functions. We also state the known result, an  $(1-1/e)$-approximation algorithm for maximizing monotone submodular functions subject to a matroid constraint, that we use in the last step (Section~\ref{sec:submodularopt}).

\paragraph{Definitions of leaders.} 
We use the following definitions throughout this section. Given a set of centers $A\subseteq \facilities$, we let $\clcost(A):=\sum_{p\in \clients}d(p,A)$ denote the cost of the corresponding clustering (where each client $p$ is assigned to the closest center in $A$ and pays the corresponding distance). Sometimes we will consider a (possibly suboptimal) assignment $\mu:\clients \rightarrow A$ of clients to centers in $A$, and let $\clcost(A,\mu):=\sum_{p\in \clients}d(p,\mu(p))$ be the corresponding cost. Let $\opt$ be a fixed optimum solution for the $k$-median instance, and let $\sopt = \clcost(\opt)$ be its cost. 
For a client $p\in \clients$, let $\opt_p:=\dist(p,\opt)$ 
be the cost paid by $p$ in $\opt$.

Let $C^*$ be an $\opt$ cluster with center $c^*\in \opt$. We let $\avg_{C^*, \opt}$ be the distance from $c^*$ to the $\eps |C^*|$th closest client of $C^*$ to $c^*$.
We further let $C^*_\avg$ be the set of clients of $C^*$ at a distance at most
$\avg_{C^*, \opt}$ from $c^*$, we refer to these clients
as \emph{leaders} for $C^*$. So $\avg_{C^*, \opt}$ is the \emph{maximum distance} from a leader to its center $c^*$. Also note that we have, $|C^*_{avg}|\geq {\eps}|C^*|$. 
We thus have that $\avg_{C^*, \opt}$ is upper bounded by $\frac{1}{1-\eps}$ times the average cost $\frac{1}{|C^*|} \sum_{p\in C^*} d(p, c^*)$, and at the same time $C^*_{avg}$ contains a constant fraction of the clients. 

\paragraph{Partition matroids.}
In Section~\ref{sec:submodularopt}, we formulate a submodular function that we then maximize over a partition matroid using known results. We define those concepts and also state the result that we will use. First, recall the definition of a matroid and partition matroid. 
A \emph{matroid} is a tuple $(E, \mathcal{I})$ defined on a ground set $E$ with a family of independent sets $\mathcal{I}$ that satisfy: (1) if $A \in \mathcal{I}$ and $A' \subseteq A$, then $A' \in \mathcal{I}$; and (2) if $A, B \in \mathcal{I}$ and $|A| < |B|$, then there exists an element $e \in B - A$ such that $A \cup \{e\} \in \mathcal{I}$.
  In the special case of  a \emph{partition matroid}, the ground set $E$ is partitioned into disjoint subsets $E_1, E_2, \dots, E_k$, and there are non-negative integers $r_1, r_2, \dots, r_k$ such that:
\begin{itemize}
    \item A subset $I \subseteq E$ is independent if and only if $|I \cap E_i| \leq r_i$ for each $i = 1, 2, \dots, k$.
\end{itemize}
In other words, each subset $E_i$ has a capacity limit $r_i$, and an independent set $I$ cannot contain more than $r_i$ elements from $E_i$.
To work with matroids (to get a running time that is polynomial in the size of the ground set $E$), we often use \emph{independence queries}, which allow us to determine if a given subset $A \subseteq E$ is independent (i.e., whether $A \in \mathcal{I}$). Specifically, an independence query on a subset $A$ returns {true} if $A$ is independent and {false} otherwise.

\paragraph{Submodular functions.} Having defined matroids, we proceed to define non-negative monotone submodular functions. For a finite ground set $E$, a set function $f: 2^E \to \mathbb{R}$ is \emph{submodular} if it satisfies the diminishing returns property: for every pair of sets $A \subseteq B \subseteq E$ and any element $e \in E - B$, it holds that
\[
f(A \cup \{e\}) - f(A) \geq f(B \cup \{e\}) - f(B).
\]
Intuitively, this means that adding an element $e$ to a smaller set $A$ provides at least as much additional value as adding $e$ to a larger set $B$.
In addition, $f$ is called \emph{monotone} if for any pair of sets $A \subseteq B \subseteq E$, we have $f(A) \leq f(B)$. In other words, adding elements to a set does not decrease the function value. Finally, $f$ is \emph{non-negative} if $f(A) \geq 0$ for all $A \subseteq E$. 

A celebrated result~\cite{CalinescuCPV11} gives a polynomial-time $(1-1/e)$-approximation algorithm for the problem of maximizing a non-negative monotone submodular function over a matroid constraint. While the original result was a randomized algorithm, a recent striking result~\cite{BuchbinderFeldman24}, building upon the work of~\cite{FilmusW14}, gives a deterministic algorithm with the same guarantee. We remark that polynomial time here refers to a running time that is polynomial in the size of the ground set $E$.   To summarize, they show 
the following theorem.
\begin{theorem}%
    Let $f: 2^E \to \mathbb{R}$ be a non-negative monotone submodular function that we can evaluate in polynomial time, and let $(E, \mathcal{I})$ be a matroid on the same groundset, for which we can answer independence queries in polynomial time. Then, for every $\zeta>0$, there is a polynomial-time algorithm that outputs a set $X\subseteq E$ with $X \in \calI$ so that
    \[
       f(X) \geq (1-1/e - \zeta) \cdot \max_{X^*\in \calI} f(X^*)\,.  
    \]
    \label{thm:submodularmatroidoptimization}
\end{theorem}

\subsection{Properties of a Locally Optimal Solution}
\label{sec:localSearchAnalysis}
Let $S$ be a {locally optimal} solution output by 
the following standard local search algorithm\footnote{Here and in the following procedures, the input instance is considered as passed implicitly.}. We remark that its running time is polynomial since by Lemma \ref{lem:aspectratio} the cost of a solution is a polynomially bounded integer and, at each improving step, the cost decreases by at least one.

\begin{mdframed}[hidealllines=true, backgroundcolor=gray!15]
\vspace{-5mm}
\paragraph{$localSearch(\cdot)$}\ \\
\begin{algorithmic}[1]
\State $S \leftarrow$ arbitrary solution for $k$-Median
\While{there exists a solution $S'$ such that $|S \Delta S'|\le 2$ and $\clcost(S') <\clcost(S)$}
\State $S \leftarrow S'$
\EndWhile
\State \textbf{return} $S$
\end{algorithmic}

\end{mdframed}
Let $S$ be the solution returned by $localSearch()$.
Local search achieves a $5$ approximation for
$k$-median, see e.g. Chapter~9.2 in~\cite{WilliamsonShmoysBook11}. 
\begin{lemma}
    $\clcost(S) \leq {5}\cdot \sopt$.
    \label{lemma:localsearchis5approximation}
\end{lemma}

In the remaining part of the algorithm and its analysis we fix $S$, and 
 recall that we fixed the reference optimal solution $\opt$. We distinguish different types of clusters of $\opt$ and analyze their properties in the next two sections. We let $S_p$ be the cost paid
by $p$ in $S$ and recall that $\opt_p$ denotes the cost $p\in \clients$ pays in $\opt$.

\subsubsection{Pure Clusters}
\label{sec:pure}
 We say that a cluster $C^*$ of $\opt$ is \emph{pure} 
(w.r.t. $S$) if there exists a cluster $C'$ of $S$ such that $|C' \Delta C^*| \le {3\eps} \min(|C^*|,|C'|)$. In which
case we also say that $C'$ is pure w.r.t. $\opt$. We say that $C'$ and $C^*$ are associated.
The following lemma follows from the fact that the input is $\beta$-stable
for $\beta := \eps/\log n$ and that the solution $S$  is a ${5}$-approximation by Lemma \ref{lemma:localsearchis5approximation}. Given a client or center $a$ and a distance $r$, we let $B(a,r)$ be the set of clients and centers at distance at most $r$ from $a$. This lemma is originally proved in~\cite{Cohen-AddadS17}, we include a proof for completeness. In this proof we use in a few places the assumption that $\eps>0$ is upper bounded by a small enough constant.
\begin{lemma}[Restatement from~\cite{Cohen-AddadS17}, Lemma IV.4]
\label{lem:numnonpure}
The number $k_{imp}$ of clusters of $\opt$  that are not pure w.r.t. $S$ (and thus the number of clusters of $S$ which are not pure w.r.t $\opt$)
    is at most {$\frac{\log n}{\eps^3}$}. %
\end{lemma}
\begin{proof} We let $C^*_o$ denote the optimal cluster for each optimal center $o\in \opt$ and we let $C_c$ denote the cluster in $S$ corresponding to center $c\in S$. 
     We refine  $\opt$  in steps so that the final set only consists of centers of pure clusters. 

    We start by removing expensive clusters in $\opt$. Specifically, let 
    \[
        \widetilde{\opt} = \{ o \in \opt \mid \sum_{p\in C^*_o} d(p, o) \geq  \frac{5\eps^3}{\log n}\sopt\}\,.
    \]
    Clearly $|\widetilde{\opt}|\leq \frac{\log n}{5\eps^3}$. Let $\opt_1:=\opt-\widetilde{\opt}$.

     For the next step,  for every $o \in \opt_1$, let $\pi(o)$ be the closest center in $\opt - \{o\}$ to $o$, and let $d_o = d(o, \opt - \{o\}) = d(o, \pi(o))$ be the distance.  As the instance is $\beta=\frac{\eps}{\log n}$ stable we have for every $o\in \opt_1$:
    \begin{gather*}
      |C^*_o| d_o =  |C^*_o| d(o, \pi(o)) \geq  \sum_{p\in C^*_o} \left(d(p, \pi(o)) - d(p,o)\right)  \geq \beta {\sopt}-\frac{5\eps^3}{\log n}\sopt\geq \frac{\beta}{2}\sopt, 
    \end{gather*}
    where the first inequality is by the triangle inequality. 
    Therefore, as the cluster $C^*_o$ has cost less than $5\eps^3 \sopt/\log n$ for $o\in \opt_1$, we have 
    \[
        |C^*_o \cap B(o, 10\eps \cdot d_o| \geq (1-\eps) |C^*_o|\,. 
    \]
    Indeed, otherwise, we would have the contradiction
    \[
         \sum_{p \in C^*_o} d(p, o) \geq |C^*_o - B(o, \eps\cdot d_o)|\cdot 10\eps\cdot d_o > |C^*_o| 10\eps^2 d_o  \geq 10 \eps^2\frac{\beta}{2} \sopt = 5 \eps^3 \frac{\sopt}{\log n}\,.
    \]
    Now let $\widetilde{\opt}_1$ be the subset of of $\opt_1$ so that for every $o\in \widetilde{\opt}_1$, we have $B(o, 20\eps \cdot d_o) \cap S = \emptyset$. We claim that $|\widetilde{\opt}_1| \leq \frac{\log n}{5\eps^3}$. 
    Indeed, for each $o\in \widetilde{\opt}_1$, the cost of the points in $C^*_o \cap B(o, 10\eps \cdot d_o)$ in the solution $S$ would be at least $10\eps d_o (1-\eps) |C^*_o| 
 \geq 10(1-\eps)\eps^2 \frac{\sopt}{\log n}$ . Therefore, using that local search gives a $5$-approximate solution, we must have $|\widetilde{\opt}_1| \leq \frac{5 \log n}{10(1-\eps)\eps^2} \leq  \frac{\log n}{5\eps^3}$.

    We define $\opt_2 = \opt_1 - \widetilde{\opt}_1$ and so now we have $B(o, 20\eps \cdot d_o) \cap S \neq \emptyset$ for every $o\in \opt_2$. Now let $\widetilde{\opt}_2$ be the subset of  $\opt_2$ so that  every $o\in \widetilde{\opt}_2$ is such that $|B(o, d_o/3) \cap S| > 1$. By the pigeon hole principle we must have that $|\widetilde{\opt}_2| \leq |\opt| - |\opt_2| =|\widetilde{\opt}|+|\widetilde{\opt}_1| \leq \frac{2\log n}{5\eps^3} $.     
    
    Let $\opt_3:=\opt_2-\widetilde{\opt}_2$. For each center $o\in \opt_3$ there is thus a center $c(o)\in S$ such that $d(o, c(o)) \leq 20 \eps d_o$ and any other center $c'\in S$ with $c'\neq c(o)$  satisfies $d(o, c') \geq d_o/3$. It follows that all points in $B(o,  d_o/8)$ are assigned to $c(o)$ in $S$, where we impose that $\eps$ is small enough so that $d_o(\frac{1}{8}+20\eps)<d_o(\frac{1}{3}-\frac{1}{8})$. In other words, at least a  $(1-\eps)$ fraction of the points of $C^*_o$ are assigned to the cluster with center $c(o)$ in $S$. Now let $\widetilde{\opt}_3$ be those centers in $\opt_3$ so that for every $o \in \widetilde{\opt}_3$ the associated center $c(o) \in S$ is assigned more than $\eps |C^*|$ points from outside the ball $B(o, d_o/8)$. Imposing that $\eps$ is small enough so that $d_o(\frac{1}{8}-20\eps)>\frac{d_o}{9}$, the cost of the cluster of $c(o)$ in $S$ is then at least
\[
    \frac{d_o}{9} \cdot \eps |C^*| \geq \frac{\eps^2}{9 \log n} \sopt.
\]
It follows that $\widetilde{\opt}_3$ has cardinality at most $45 \frac{\log n}{\eps^2} \leq \frac{\log n}{5\eps^3}$, where the last inequality again holds for $\eps$ small enough. Our final set is $\opt_4:=\opt_3-\widetilde{\opt}_3$.

Now for each center in $c \in \opt_4$, with associated cluster $C^*$, we have that it  together with its associated center $c(o) \in S$, with associated cluster $C'$, satisfies the following:

\begin{itemize}
    \item All clients in $B(o,  d_o/8)$, i.e. at least $(1-\eps) |C^*|$ many clients, are both assigned to $o$ and $c(o)$ in $\opt$ and $S$, respectively. Notice that this implies $\min\{|C^*|,|C'|\}\geq (1-\eps)|C^*|$.
    \item The additional clients assigned to $o$ and not to $c(o)$ is at most $\eps |C^*|$.
    \item The additional clients assigned to $c(o)$ and not $o$ is at most $\eps |C^*|$.
\end{itemize}
Notice that $|C^*\Delta C'|\leq 2\eps |C^*|\leq 3\eps\min\{|C^*|,|C'|\}$, in particular $C^*$ is pure. It follows that all centers in $\opt_4$ correspond to pure clusters. Furthermore, $\opt_4$ was obtained by removing at most $\frac{\log n}{\eps^3}$ centers. The claim follows.   
\end{proof}

Let $C^*$ be a pure cluster of $\opt$ and $C'$ be the associated cluster of $S$. Let
$c^*$ be the center of $C^*$ and ${c'}$ be the center of ${C'}$. We define $t(c^*) := {c'}$.
For any client
$p \in C^*$, let $r(p) := d(p, {c'}) = d(p, t(c^*))$.
We have the following lemma which is the only lemma
that requires that the solution $S$ is obtained via the local search (the proof of the previous lemma only used that it was a constant-factor approximation). In words, it says that $S$ approximates the connection cost of clients belonging to pure clusters of the optimal solution almost perfectly.

\begin{lemma}[Adapted from~\cite{Cohen-AddadS17}]
\label{lem:purecost}
    Let $\clientspure \subseteq \clients$ be the subset of clients that belong to pure clusters  of the optimal solution $\opt$. We have
    \begin{gather*}
        \sum_{p\in \clientspure} r(p) \leq \sum_{p\in \clientspure} \opt_p + O(\eps \cdot \sopt)\,.
    \end{gather*}
\end{lemma}
\begin{proof}

Recall that $S_p = \dist(p,S)$. Consider a pair of pure clusters $C^*$ and $C'$ with the same notation as above. We prove the following:
     \begin{enumerate}
         \item $\sum_{p \in C^* \cap C'} r(p) = \sum_{p \in C^* \cap C'} d(p,S) \le  \sum_{p \in C^* \cap C'} (1+{4\eps}) \opt_p + {4\eps} S_p$; and
         \item  
         $\sum_{p \in C^* {-} C'} r(p) {=} \sum_{p \in C^*{-} C'} d(p, c') \le \sum_{p \in C^*{-} C'} \opt_p + {4\eps} \sum_{p \in C^*} (\opt_p + {S_p})$.
     \end{enumerate}
     The statement of the lemma then follows by summing up the above bounds and using that $\clcost(S) \leq {5} \cdot {\sopt}$ by Lemma~\ref{lemma:localsearchis5approximation}.

For proving the first bullet, consider the swap that swaps in $c^*$ and swaps out the center $c'$. The cost change for the clients in $C^* \cap C'$ is exactly $\sum_{p \in C^*\cap C'} (\opt_p - S_p)$.
For any $p \not\in C^*$, if $p$ is not served by ${c'}$ in $S$ the cost is unchanged. Otherwise the cost of 
$p \in C' - C^*$ in the new 
solution $S - \{{c'}\} \cup \{c^*\}$ 
is by triangle inequality at most $S_{{p}} + \frac{1}{|C^*\cap C'|} \sum_{p' \in C^*{\cap C'}} (S_{p'} + \opt_{p'})$.
Summing over all clients $p \in C'-C^*$, this is at most 
$\sum_{p \in C'-C^*} S_{{p}} + \frac{|C'-C^*|}{|C^*\cap C'|} \sum_{p' \in C^*{\cap C'}} (S_{p'} + \opt_{p'})$. Thus, the cost difference $\clcost(S- \{c'\} \cup \{c^*\}) - \clcost(S)$
for these clients is at most $\frac{|C'-C^*|}{|C^*\cap C'|} \sum_{p' \in C^*{\cap C'}} (S_{p'} + \opt_{p'})\leq {4\eps} \sum_{p' \in C^*{\cap C'}} (S_{p'} + \opt_{p'})$, where in the inequality we used the fact that $C^*$ is pure, {hence $\frac{|C'-C^*|}{|C^*\cap C'|}\leq \frac{{3\eps}|C^*|}{(1-\eps)|C^*|}\leq 4\eps$.} The first bullet follows since, by local optimality, $\clcost(S- \{c'\} \cup \{c^*\}) - \clcost(S) \ge {0}$.

We turn to the second bullet. By triangle inequality, for each $p\in C^*-C'$, $d(p,c')$ is at most
$d(p, c^*) + \frac{1}{|C^* \cap C'|} \sum_{p' \in C^* {\cap} C'} (S_{p'} + \opt_{p'})$. Summing up over all clients
$p \in C^* - C'$, we have that
\[\sum_{p \in C^*{-} C'} d(p, c')\le \sum_{p \in C^*- C'} \opt_p + \frac{|C^*{-} C'|}{|C^*\cap C'|} \sum_{p' \in C^*{\cap}C'} (\opt_{p'} + S_{p'}).\]
The second bullet follows from the assumption that $C^*$ is pure, and hence $\frac{|C^*{-} C'|}{|C^*\cap C'|}\leq {4\eps}$ similarly to the first bullet.
\end{proof}

\subsection{Cheap Clusters and $D$-Sample Process to Hit Expensive Clusters}
\label{sec:dsampleproc}
In the previous section, we showed that the cost of pure clusters in $\opt$ is close to to their optimal cost in our local search solution $S$. Our goal in this section is to ``hit'' those non-pure clusters of $\opt$
 that have a high cost in $S$.
We say that a non-pure cluster $C^*$ of $\opt$ with center $c^*$ is \emph{basic-cheap} if the total cost in 
$S$ of the clients in $C^*_\avg$ is less than $\eps^{5}\sopt/\log n$ or 
if there exists a center of $S$ at distance at most
{$(1+\eps)\avg_{C^*,\opt} + \eps \gamma_{C^*}$} from $c^*$, where $\gamma_{C^*} := \frac{1}{|C^*|} \sum_{p\in C^*} (\opt_p + S_p)$; in the latter case we also say that $C^*$ is \emph{covered} by $S$.

The following procedure,  that aims to ``hit'' all clusters of $\opt$ that are non-pure and non-basic-cheap,  is inspired by the classic $k$-means++ algorithm and its
    so-called $D^2$-sampling procedure to
    sample clients proportional to their cost in solution $S$. For $k$-median this becomes $D$-sampling (instead of $D^2$-sampling). 
\begin{mdframed}[hidealllines=true, backgroundcolor=gray!15]
\vspace{-5mm}
\paragraph{$D$-$Sample(S,s)$}\ \\
\begin{algorithmic}[1]
\State $W \leftarrow \emptyset$ 
\For{$s$ many times}
\State Add to $W$ one client $p$ sampled with probability $\frac{\dist(p, S)}{\clcost(S)}$
\EndFor
\State \textbf{return} $W$
\end{algorithmic}
\end{mdframed}

Next let $W$ be the set returned by $D$-$Sample(S,s^*)$ with
\[
    s^* = {5} \cdot \frac{\log n}{\eps^{5}}\ln\left( \frac{7}{\eps^{2}}\right)\,.
\]
We say that $W$ \emph{hits} a cluster $C^*$ of $\opt$ if $C^*_{\avg}\cap W\neq \emptyset$, i.e., if $W$ contains at least one leader of $C^*$. We have the following lemma that says that we hit all non-pure and non-basic-cheap clusters except for a set $V$ of clusters that have an insignificant cost in $S$. 

For a non-pure optimal cluster $C^*$ with center $c^* \in \opt$, we let $t(c^*)$ be the closest center in $S$ to $c^*$, and for any $p \in C^*$, 
we let the replacement cost of $p$ be $r(p) := d(p, t(c^*))$. 

\begin{lemma}
\label{lem:probcost}
Let $V$ be the set of clusters $C^*$ of $\opt$ that are non-pure and non-basic-cheap  such that $C^*_\avg  \cap W = \emptyset$. 
With probability at least $1-\eps$, 
\begin{align}\sum_{C^* \in V} \sum_{p \in C^*} r(p) \le \eps \cdot \sopt\,.
\label{eq:probcost}
\end{align}
\end{lemma}
\begin{proof}
Recall that $S_p = \dist(p,S)$.
Consider a non-pure and non-basic-cheap cluster $C^*$ of $\opt$.
We say that a sampled point $p\in W$ \emph{hits} $C^*$ if $p\in C^*_{avg}$ (i.e., $p$ is a leader of $C^*$). 
As $C^*$ is not basic-cheap, we have $\sum_{p\in C^*_{avg}}S_p\geq \eps^{5} \sopt/\log n$. Thus the probability that a sampled $p\in W$ hits $C^*$ is at least $\frac{\sum_{p\in C^*_{\avg}}S_p}{\clcost(S)}\geq \frac{\eps^{5}}{{5} \log n}$, where we used that $\clcost(S) \leq {5} \sopt$ (Lemma~\ref{lemma:localsearchis5approximation}). Thus the probability that $C^*$ is not hit by any point in $W$ is at most $(1-\frac{\eps^{5}}{{5} \log n})^{s^*}$ which by the selection of $s^*$ is at most  $\eps^{2}/{7}$. For a point $p\in C^*$ which is assigned to $c'$ in $S$, one has
\begin{align*}    
r(p) & =d(p,t(c^*))\leq d(p,c^*)+d(c^*,t(c^*))\leq d(p,c^*)+d(c^*,c')\\
& \leq d(p,c^*)+d(p,c^*)+d(p,c')=2\opt_p+S_p. 
\end{align*}
Thus, 
by linearity of expectation,
$$
E[\sum_{C^*\in V}\sum_{p\in C^*}r(p)] \leq \sum_{p\in \clients} (2\opt_p+S_p)\cdot(\eps^{2}/{7}) \leq \eps^{2}\cdot \sopt\,,
$$
where we again used that $\clcost(S) \leq {5} \sopt$ by Lemma~\ref{lemma:localsearchis5approximation}.
The total cost of the clusters in $V$ is thus at most $ \eps\cdot {\sopt}$ with probability at least $1-\eps$ by Markov's inequality. 
\end{proof}

\paragraph{Successful $W$.}
We say that the set of clients $W$ obtained from the D-sample process is a 
\emph{successful} sample  if~\eqref{eq:probcost} holds.
We define the following two types of $\opt$ clusters. We say that 
a non-pure cluster $C^*$ of $\opt$ with center $c^*$ is \emph{cheap} if it is basic-cheap 
or is in $V$; a  non-pure cluster is 
\emph{expensive} otherwise.  In other words, a cluster $C^*$ of the optimal solution is expensive if it is a non-pure and non-basic-cheap cluster such that $C^*_{avg} \cap W \neq \emptyset$. 

Assuming that $W$ is successful, we have a good bound on the cost of the clients assigned to cheap centers in the optimal solution. 
\begin{lemma}
\label{lem:cheapcost}
    Let $\clientscheap \subseteq \clients$ be the subset of clients that belong to cheap clusters of the optimal solution $\opt$. If $W$ is a successful sample, 
    \begin{gather*}
        \sum_{p\in \clientscheap} r(p) \leq \sum_{p\in \clientscheap} 2\opt_p + O(\eps \cdot \sopt)\,.
    \end{gather*}
\end{lemma}
\begin{proof}
We claim that for any cluster $C^*$ of $\opt$ that is basic-cheap, and   any $p \in C^*$: 
    \begin{enumerate}
    \item If $C^*$ is not covered, then
      $r(p) = d(p, t(c^*)) \le \opt_p + {\avg_{C^*,\opt}} + \frac{1}{\eps |C^*|} \sum_{p \in C^*_{\avg}}  S_p$. 
     \item If $C^*$ is covered, then
        $r(p) \le \opt_p + (1+\eps) \avg_{C^*,\opt} + \eps \gamma_{C^*}$.
    \end{enumerate}
If $C^*$ is covered, the bound follows immediately by the triangle inequality. 
Otherwise, observe that for any such cluster, we have that $d(p, t(c^*)) \le \opt_p +  d(t(c^*),c^*)$, 
and let us bound $d(t(c^*),c^*)$.  We have that 
    $ d(t(c^*),c^*) \le \frac{1}{|C^*_\avg|} \sum_{p \in C^*_\avg} (S_p + \opt_p)$.
    Which is at most
    ${\avg_{C^*,\opt}} + \frac{1}{|C^*_\avg|} \sum_{p \in C^*_\avg} S_p $ by definition. So the inequality follows from the fact that $|C^*_{\avg}| \ge \eps |C^*|$.

    Now, to prove the above lemma, we sum up the cost of all the clients of cheap clusters $C^*$. For a covered cluster $C^*$, by 2. we have  
    \begin{align*}
        \sum_{p\in C^*} r(p)& \leq \sum_{p\in C^*} \left(\opt_p + (1+\eps)\avg_{C^*, \opt} + \eps \gamma_{C^*}\right)\\
        &\leq \left( 1+ (1+2\eps) (1+\eps)\right) \sum_{p\in C^*} \opt_p + \eps \sum_{p\in C^*} (\opt_p + S_p)\,,
    \end{align*}
    where we used the definition of $\gamma_{C^*} = \frac{1}{|C^*|} \sum_{p\in C^*} (\opt_p + S_p)$ and that $|C^*| \cdot \avg_{C^*, \opt} \leq (1+2\eps) \sum_{p\in C^*} \opt_p$, which holds because $C^*$ has at least $(1-\eps) |C^*|$ clients $p\in C^*$ so that $\opt_p \geq \avg_{C^*, \opt}$.

    For a basic-cheap cluster $C^*$ which is not covered, by 1. one has
    \[
        \sum_{p\in C^*} r(p) \leq \sum_{p\in C^*} \left( \opt_p + {\avg_{C^*,\opt}} + \frac{1}{\eps |C^*|} \sum_{p \in C^*_{\avg}}  S_p\right) \leq (2+2\eps) \sum_{p\in C^*} \opt_p +  \eps^{4}\sopt/\log(n)\,,
    \]
    where we again used that $|C^*| \cdot \avg_{C^*, \opt} \leq (1+2\eps) \sum_{p\in C^*} \opt_p$ and we additionally used that, by the definition of basic-cheap clusters, $\sum_{p\in C^*_{avg}} S_p \leq \eps^{5} \sopt/\log n$.
    Finally, we have that if we sum up all clusters in $V$, i.e., clusters that are cheap but not basic-cheap, then  Lemma~\ref{lem:probcost} says
    \begin{gather*}
        \sum_{C^* \in V} \sum_{p \in C^*} r(p) \le \eps \cdot \sopt\,.
    \end{gather*}
    The statement of the lemma now follows by summing up the above bounds for all cheap clusters and using that $\clcost(S) \leq 5\cdot \sopt$ by Lemma~\ref{lemma:localsearchis5approximation} and the property that there are at most $\log(n)/\eps^3$ basic-cheap clusters. The latter claim is true because Lemma~\ref{lem:numnonpure} says that there are at most $\log(n)/\eps^3$ non-pure clusters and basic-cheap clusters is a subset of non-pure clusters.
\end{proof}

Given $S$ and $W$, the set of clusters of the optimal solution $\opt$ is thus partitioned into pure clusters, cheap clusters, and expensive clusters. We let $\ho$ be the centers of $\opt$ that are expensive. We further define the set $\clientsexpensive \subseteq \clients$  to be the subset of clients that belong to the expensive clusters in $\opt$. We have thus partitioned $D$ into sets $\clientspure, \clientscheap$, and $\clientsexpensive$ depending on the type of optimal cluster they belong to. Lemma~\ref{lem:purecost} bounds the cost of the clients in $\clientspure$ in $S$ and Lemma~\ref{lem:cheapcost} bounds the cost of the clients in $\clientscheap$. For future reference, we summarize these two lemmas  (by weakening the upper bound for clients in $\clientspure$).
\begin{lemma}
Assuming  $W$ is a successful sample, 
    \[
    \sum_{p\in D - \clientsexpensive} r(p) \leq \sum_{p \in D - \clientsexpensive} 2 \opt_p + O(\eps \cdot \sopt)\,.
    \]
    \label{lemma:convenience_upper_bound}
\end{lemma}

It follows that the connection cost in solution $S$ of clients in $\clients - \clientsexpensive$  is within a factor $2$ of their cost in the optimal solution. The remaining part of the algorithm is thus devoted to modifying $S$ to obtain a small connection cost of the clients in $\clientsexpensive$ without significantly increasing the cost of the other clients.

\subsection{Identifying Balls of Expensive Clusters}
\label{sec:ballguesses}

By the definition of expensive clusters, each expensive cluster $C^*$ of $\opt$ satisfies $C^*\cap W \neq \emptyset$ where $W$ is the sample obtained by running the sampling procedure of the last subsection. Here, our goal is to guess a subset $\calB$ of balls so that for each expensive cluster $C^*$ there is a ball $B(\ell, \rho) \in \calB$ so that $\ell$, which we refer to as a leader, is in  $C^*_{avg}$ and $\rho$ is a very close approximation to $\avg_{C^*, \opt}$. This guarantees that the center $c^* \in \ho$ of cluster $C^*$ is in $B(\ell, \rho)$ and that no other center in $B(\ell, \rho)$ is "too" far from $c^*$ since the radius of the ball is $\approx \avg_{C^*, \opt}$. The balls $\calB$ will then be used when we approximate the centers in $\ho$ via submodular function maximization subject to a partition matroid constraint (the balls will correspond to the partitions of the matroid). 
As the algorithm does not know the set of expensive clusters, it enumerates all subsets of $W$, and for each guessed point, it enumerates a constant number of radii.

\begin{mdframed}[hidealllines=true, backgroundcolor=gray!15]
\vspace{-5mm}
\paragraph{$ballGuess(S,W)$}\ \\
\begin{algorithmic}[1]
\State $\calL_\texttt{ball} \leftarrow \emptyset$
\For{all subsets $\{p_1, p_2, \ldots, p_q\} \subseteq W$}
\For{all integers $i_1, i_2, \ldots, i_q$ such that
     $\eps^3  {S_{p_j}} \le (1+\eps^3)^{i_j} \le {S_{p_j}}/\eps^3$ for $j=1, 2,\ldots, q$}
\State add $\calB = \left\{B(p_1, (1+\eps^3)^{i_1}), B(p_2, (1+\eps^3)^{i_2}), \ldots, B(p_q, (1+\eps^3)^{i_q}\right\}$ to $\calL_{\texttt{ball}}$    
\EndFor
\EndFor
\State \textbf{return} $\calL_\texttt{ball}$
\end{algorithmic}
\end{mdframed}

We show that the $ballGuess()$ runs in polynomial time and that  one of the subsets $\calB$ of balls that the procedure constructs is a correct guess: every expensive cluster $C^*$ has a leader, and $\avg_{C^*, \opt}$ is guessed approximately correctly.
We say that a set of balls $\calB$ is a \emph{valid} set of balls if it satisfies the following: For each expensive cluster $C^*$ of $\opt$,  we have a ball $B(\ell, \rho) \in \calB$ such that $\ell\in C^*_{avg}$ is a leader and $$\avg_{C^*,\opt} \leq \rho \leq  \avg_{C^*,\opt} + \eps \gamma_{C^*}, $$
where we recall that $\gamma_{C^*} = \frac{1}{|C^*|} \sum_{p \in C^*} (\opt_p + S_p)$.

\begin{lemma}
\label{lem:ballguesses} 
$ballGuess()$ runs in polynomial time $n^{\eps^{-O(1)}}$ and produces a collection $\calL_{\texttt{ball}}$ of sets of balls such that at least one set $\calB\in \calL_{\texttt{ball}}$ is a valid set of balls.
\end{lemma}
\begin{proof}

We start by showing that the running time is bounded by $n^{\eps^{-O(1)}}$. Indeed, we have $|W| = s^*$, and so the number of subsets of $W$ is $2^{s^*}= n^{\eps^{-O(1)}}$ since $s^* = {5} \cdot \frac{\log n}{\eps^{5}}\ln\left( \frac{7}{\eps^{2}}\right)$. Moreover, for a fixed subset $\{p_1, p_2, \ldots, p_q\}$, we consider a constant number  $\alpha = \log_{1+\eps^3}(1/\eps^6)$ of balls (i.e., values of $i_j$) for each $p_j$. So the total number of balls considered added for each subset $\{p_1, p_2, \ldots, p_q\}$ is $\alpha^{q}$ which is $n^{\eps^{-O(1)}}$ since $q \leq |W| \leq s^*$.

We now turn our attention to showing that there is a set of balls $\calB \in \calL_{\texttt{ball}}$ that satisfies the properties of the lemma.
Let $C^*(1), \ldots, C^*(|\ho|)$ be the expensive clusters of \opt. 
By definition, we have $C^*_{avg}(j) \cap W\neq \emptyset$ for each such cluster. Now arbitrarily fix one client $\ell_j\in C^*_{avg}(j) \cap W$  for $j=1,2, \ldots, |\ho|$, which we refer to the leader of $C^*(j)$. Since the algorithm enumerates over all subsets
    of $W$, there exists one such subset $\{\ell_1, \ldots, \ell_{|\ho|}\}$ that is a valid set of leaders for all the expensive clusters.
    We consider the set  $$\calB = \left\{B(\ell_1, (1+\eps^3)^{i_1}), B(\ell_2, (1+\eps^3)^{i_2}), \ldots, B(\ell_{|\ho|}, (1+\eps^3)^{i_{|\ho|}})\right\}$$ of balls, where $i_j$ is the smallest integer so that $(1+\eps^3)^{i_j} \geq \max\{\avg_{C^*, \opt}, \eps^3 S_{\ell_j}\}$. 

    It remains to prove that, for each leader $\ell_j$, we have $(1+\eps^3)^{i_j} \leq S_{\ell_j}/\eps^3$  and $\avg_{C^*(j),\opt} \leq (1+\eps^3)^{i_j}\leq  \avg_{C^*(j),\opt} + \eps \gamma_{C^*(j)}$. The first property guarantees that  $\calB \in \calL_{\texttt{ball}}$ and the second property is the desired bound on each radius $\rho$ for $\calB$ to be a valid set of balls.

     We now verify these two properties for each leader $\ell_j$. To simplify notation, we let $\ell = \ell_j$, $C^* = C^*(j)$ and $\rho = (1+\eps^3)^{i_j}$.
    We have
    $d(\ell, S) = S_{\ell} > \eps\, \avg_{C^*,{\opt}}$ since otherwise
    $C^*$ would be covered by $S$, which would contradict that $C^*$ is an expensive cluster. By the selection of $i_j$,  $\rho = (1+\eps^3)^{i_j}\leq (1+\eps^3) \max\{\avg_{C^*, \opt}, \eps^3 S_\ell\} \leq (1+\eps^3) S_\ell/\eps$ and thus $(1+\eps^3)^{i_j} \leq S_\ell/\eps^3$ as required. 
    
    We further claim that $\rho \in [\avg_{C^*,\opt},  \avg_{C^*,\opt} + {\eps^2} S_\ell)$. 
    Assume first that $\avg_{C^*,\opt}\geq \eps^3 S_\ell$. Then $\rho = (1+\eps^3)^{i_j}\geq\avg_{C^*,\opt}$, for the smallest possible integer $i_j$ which satisfies this condition. The claim then follows  since $\rho \leq (1+\eps^3)\avg_{C^*,\opt}\leq \avg_{C^*,\opt}+\eps^2 S_\ell$ as required (where we used that $\avg_{C^*, \opt} \leq S_\ell/\eps$). In the complementary case, namely $\avg_{C^*,\opt}< \eps^3 S_\ell$, one has that $i_j$ was selected to be the smallest integer so that $(1+\eps^3)^{i_j} \geq \eps^3 S_\ell$. Moreover,  the interval $[\avg_{C^*,\opt},  \avg_{C^*,\opt} + \eps^2 S_\ell)$ contains the interval $[\eps^3 S_\ell,\eps^2 S_{\ell}]$, and the latter interval contains at least one power of $(1+\eps^3)$. We thus have
    \[
        \avg_{C^*, \opt} \leq \rho \leq \avg_{C^*, \opt} + \eps^2 S_{\ell}.
    \]
    
    We conclude the proof by showing that $S_\ell \le (2+2\eps)\gamma_{C^*}$. Let $\ell$ be assigned to $c'$ in $S$, and consider any $p\in C^*$ which is assigned to $c''$ in $S$. Then one has $S_\ell=d(\ell,c')\leq d(\ell,c'')\leq d(\ell,c^*)+d(c^*,p)+d(p,c'')=d(\ell,c^*)+\opt_p+S_p$. Hence, averaging over $C^*$, one gets $S_\ell \le d(\ell, c^*) + \frac{1}{|C^*|} \sum_{p \in C^*} (S_p + \opt_p)$. Since $\ell$ is a leader, i.e., $\ell \in C^*_{avg}$, $d(\ell,c^*)\leq \avg_{C^*,\opt}\leq  \frac{1+2\eps}{|C^*|}\sum_{p \in C^*}\opt_p$, where we used that there is at least $(1-\eps)|C^*|$ clients $p\in C^*$ so that $d(p, c^*) \geq \avg_{C^*, \opt}$. Therefore, $S_\ell \leq (2+2\eps) \gamma_{C^*}$, which, as aforementioned, concludes the proof of the lemma.

\end{proof}

We thus have that the ball procedure produces a family $\calL_{\texttt{ball}}$ such that a set $\calB \in \calL_{\texttt{ball}}$ satisfies the conditions of the lemma.  While our algorithm will try all possible sets in $\calL_{\texttt{ball}}$ (since it does not know which one is valid), we do our analysis  by considering the run of the algorithm when it selects this set  $\calB$ of valid balls. 

\subsubsection{\emph{Dummy} Centers and Mixed Solutions $M_O$ and $M_D$}
\label{sec:dummycenters}
    Next, the algorithm creates, for each ball  $B(\ell, \rho) \in \calB$ with center $\ell$ and
    radius $\rho$, 
    a \emph{dummy} center $\delta$ at distance $\rho$ from $\ell$ and at distance 
    $\rho + d(\ell, p)$ from any other input point $p$. Let $\Lambda$ be the set of dummy centers, so $|\Lambda| = |\calB|$. Recall that  $\clientsexpensive$ is the set of clients that are in an expensive cluster of $\opt$, i.e., served by a center
in $\ho$ in the optimal solution.

\begin{lemma}
\label{lem:structSminusS0}
Suppose the sample $W$ is successful, and that the set of balls $\calB$ is valid, then there exists a set of centers $S_0$ of $S$ that satisfies the following properties:
  \begin{enumerate}
  \item $|S_0| = |\ho| \le \frac{\log n}{\eps^3}$, and;
  \item $\forall c \in S_0$, there is no $c^* \in \opt-\ho$ such that $t(c^*) = c$, and;
  \item $\clcost(S - S_0 \cup \dummyset) \le
  3 \sum_{p \in \clientsexpensive}\opt_p + \sum_{p \in \clients {-} \clientsexpensive}
  r(p) +
  O(\eps \cdot{\sopt})) \le (3+O(\eps))\sopt$.
  \end{enumerate}
  \label{lemma:S0properties}
\end{lemma}
\begin{proof}
By definition, the centers in $\hat{\opt}$ are not pure, hence $|\hat{\opt}|\leq \frac{\log n}{\eps^3}$ by Lemma \ref{lem:numnonpure}.
Furthermore, since there are $|\opt - \ho|$ centers of $\opt$ not
in $\ho$, there are (at least) $|\ho|$ centers $c$ of $S$ such that there is no $c^* \in \opt-\ho$ such that $t(c^*) = c$, and we let $S_0$
refer to these centers: 1. and 2. follow immediately.

Let us turn to 3. The bound on the cost 
of the clients that are not in $\clientsexpensive$ 
follows from 2.: for 
each client $p$ in a cluster of $\opt$
whose center is $c^* \in \opt-\ho$, we have
that $t(c^*) \in S-S_0$ and so its cost
is at most $r(p)$.
We finally consider the points in $\clientsexpensive$.
Since the balls are valid, Lemma~\ref{lem:ballguesses} implies that 
for each expensive cluster $C^*$ of $\opt$, there is a 
ball $B(\ell,\rho)$  where $\rho \in [\avg_{C^*,\opt},  \avg_{C^*,\opt} + \eps \gamma_{C^*})$ and 
$\ell \in C^*_\avg$. Therefore, if we let $c^*$ be the center of $C^*$,
\begin{align*}
\sum_{p\in C^*} d(p, \dummyset) & \leq \sum_{p\in C^*} \left(d(p, c^*) + d(c^*, \ell) +  d(\ell, \dummyset) \right)\\
&\leq\sum_{p\in C^*} \left(\opt_p  + \avg_{C^*,\opt}+ \rho \right)\\
     & \leq \sum_{p\in C^*} \left(\opt_p + 2(\avg_{C^*,\opt} + {\eps} \gamma_{C^*}) \right) \\
    & = \sum_{p\in C^*} \opt_p + |C^*|2(\avg_{C^*,\opt} + {\eps} \gamma_{C^*})\\
    & \le  (3+O(\eps)) \sum_{p\in C^*} \opt_p
    + O(\eps \sum_{p \in C^*} {S_p}).
\end{align*}
The lemma follows from 
 $\clcost(S) \leq {5} \cdot {\sopt}$ (Lemma~\ref{lemma:localsearchis5approximation}) and by
invoking Lemma~\ref{lemma:convenience_upper_bound} to bound the sum
of the $r(p)$ values (where we use the assumption that the sample $W$ is successful).
\end{proof}

\paragraph{Definition of the mixed solutions $M_O$ and $M_D$.}
Next, let us analyze
the cost of the swap $(\ho, S_0)$. Namely, of the solution $M_O = S - S_0 \cup 
\hat{\opt}$, where $M$ in $M_O$ stands for  ``mixed solution''  and the subscript $O$ stands for that the mixed solution  is obtained by swapping in some elements from $\opt$  (and removing $S_0$). We will also analyze the ``dummy'' version of $M_O$ where instead of swapping in $\ho$, we swap in the dummy centers ${\dummyset}$. We denote that solution by $M_D = S - S_0 \cup \dummyset$, where the subscript $D$ stands for ``dummy''.  
\begin{lemma}
    We have  $d(p, \ho) = \opt_p$ if $p\in \clientsexpensive$ and $d(p, S- S_0)\leq r(p)$ if $p\in\clients - \clientsexpensive$.
    \label{lemma:costsMO}
\end{lemma}
\begin{proof}

  We have $d(p, \ho) = \opt_p$ when $p\in \clientsexpensive$ because the center serving $p$ in the optimal solution is in $\ho$. Consider the remaining case and let $p\in \clients - \clientsexpensive$ be a client whose serving center in $\opt$ is $c^*$. 
  By property 2. of Lemma~\ref{lemma:S0properties}, the center $t(c^*)$ is in $S - S_0$  and
so  $d(p, S- S_0) \leq d(p, t(c^*)) = r(p)$.  %
\end{proof}

Notice that $M_{O}$ contains $S- S_0$ and $\ho$. Therefore, if we sum up the above bounds for all clients we get
\begin{align*}
    \clcost(M_O) & \leq \sum_{p\in \clientsexpensive} \opt_p  + \sum_{p\in \clients {-} \clientsexpensive} r(p)\,, 
\end{align*}
which by Lemma~\ref{lemma:convenience_upper_bound} is at most $\sum_{p\in \clientsexpensive} \opt_p  + \sum_{p\in \clients {-}\clientsexpensive} 2 \opt_p + O(\eps \cdot \sopt)$ (assuming that that the sample $W$ was successful). By swapping $S_0$ with $\ho$ we thus get a $1$-approximation on the clients  $\clientsexpensive$ and a $2$-approximation of the remaining clients (plus a small error term). Our  goal in the next sections will thus be to approximate this swap. The next few steps will be aimed at simplifying (or guessing parts of) $S_0$. The last step will then approximate $\ho$ and the remaining part of $S_0$ by submodular function maximization subject to a partition matroid constraint. For that part, it will be helpful to have a bound on the dummy solution as well; specifically, a modified version of it (see Lemma~\ref{lemma:costboundofMOandMDwithassignments} and Claim~\ref{claim:Mprimebounds}). For intuition of those statements, let us here say that one can show 
\begin{align*}
        \clcost(M_D) & \leq \sum_{p\in \clientsexpensive} 3\opt_p  + \sum_{p\in \clients - \clientsexpensive} r(p)  + O(\eps \cdot \sopt)\,\,
\end{align*}
by observing that only the clients $\clientsexpensive$ have a different connection cost in $M_D$ than in $M_O$. 
Moreover, each such client $p\in \clientsexpensive$ that was previously assigned to a center $c^* \in \ho$, and belongs to cluster $C^*$ of $\opt$, has an associated dummy center $\delta$ at distance at most $d(p, c^*) +  2 \avg_{C^*, \opt} + \eps \gamma_{C^*}$ (assuming the set $\calB$ of balls is valid). Simplifications then give the stated inequality as $|C^*| \avg_{C^*, \opt} \leq (1+2\eps) \sum_{p\in C^*} \opt_p$ and $|C^*| \eps \gamma_{C^*} = \eps \sum_{p\in C^*} ( \opt_p + S_p)$.
\subsection{Removal of Expensive Centers of $S_0$}
\label{sec:removalofExpensive}

Thus, it remains to show that our algorithm can yield a good approximation to
the gain that the swap $(\ho, S_0)$ would provide. Of course, it will not 
necessarily be optimum, but we can show it
is enough to conclude 
the proof of our theorem.
We focus on the iteration of the procedure
that produces valid ball guesses. In particular, at this point given that $W$ is successful and the ball guesses are valid, the number of balls in our ball guess is equal to $|S_0|=|\ho|=|\Lambda|\leq \frac{\log n}{\eps^3}$. Our algorithm then
makes use of the following procedure that takes as input the local search solution $S$ and the set $\dummyset$ of dummy centers. 
\begin{mdframed}[hidealllines=true, backgroundcolor=gray!15]
\vspace{-5mm}
\paragraph{$expRem(S,\dummyset)$}\ \\
\begin{algorithmic}[1]
\State $\calLexp\leftarrow \emptyset$
\For{$(\frac{8}{\eps})^{|\Lambda|+1}\ln n$ many iterations}
\State $\calQ\leftarrow \emptyset$
\For{$|\Lambda|+1$ many iterations}
\State With probability 1/2 do the following. Consider the clustering induced by $\dummyset\cup S-\calQ$ and, for each $c\in S-\calQ$, let $\clcost(c)$ be the cost of the cluster associated with $c$ in the considered solution
\State Sample one $c\in S-\calQ$ with probability 
$\frac{\clcost(c)}{\sum_{c'\in S-\calQ}\clcost(c')}$ 
\State Set $\calQ\leftarrow \calQ\cup \{c\}$

\EndFor
\State $\calLexp\leftarrow \calLexp\cup \{\calQ\}$ 
\EndFor
\State \textbf{return} $\calLexp$
\end{algorithmic}
\end{mdframed}
The goal of the above procedure is to guess a subset $\calQ$ of all ``costly'' centers in $S_0$ (the set $\calLexp$ contains all guesses). Specifically, if we consider the solution $(S - \calQ) \cup \dummyset$, then we wish that the cost of the clusters corresponding to the remaining centers in $S_0 - \calQ$ is at most $\eps \cdot \sopt$. Formally, the clusters corresponding to $S_0 - \calQ$ 
 in $S - \calQ \cup \dummyset$ consists of all the clients whose closest center in $S - \calQ \cup \dummyset$ is from $S_0 - \calQ$. If we let $\clients' \subseteq \clients$ be those clients then the total cost in solution  $S- \calQ \cup \Lambda$ of the clusters with centers in $S_0 - \calQ$ is defined as 
 \[
 \sum_{p\in \clients'} d(p, S- \calQ \cup \Lambda) = \sum_{p\in \clients'} d(p, S_0 - \calQ)\,,
 \]
 where the equality holds because of the definition of $\clients'$.

    \begin{lemma}
    \label{lem:successguessprocess}
$expRem()$ runs in $n^{1/\eps^{O(1)}}$ time and produces a collection $\calLexp$ of at most $n^{1/\eps^{O(1)}}$ subsets of $S$ such that with probability at least $1-1/n$ there exists $\calQ\in \calLexp$ satisfying
    \begin{enumerate}
    \item $\calQ \subseteq S_0$, and
    \item The total cost in the solution $S - Q \cup \dummyset$ of the clusters with centers in $S_0 - \calQ$ is at most $\eps \cdot \sopt$.
    \end{enumerate}
\end{lemma}
\begin{proof}
The claim on the running time and on the size of $\calLexp$ is trivially satisfied since $|\Lambda|=|S_0|\leq \frac{\log n}{\eps^3}$. 

It remains to show that with probability at least $1-\frac{1}{n}$ at least one set $\calQ\in \calLexp$ satisfies the desired properties. Let us consider some iteration of the external for loop, and let $\calQ$ be the corresponding set. We say that the $j$-th iteration of the corresponding inner for loop is successful if the following happens:
\begin{enumerate}
    \item If condition (2) of the lemma is satisfied considering the current value of $\calQ$, then the event from line 5 does not happen (hence in particular $\calQ$ is not updated in this $j$-th iteration).
    \item Otherwise, the event from line 5 happens and furthermore the sampled $c$ belongs to $S_0$.
\end{enumerate}
Let $A_j$ denote the event that the considered $j$-th iteration is successful.
Let us condition on the event $A_{<j}$ that the previous iterations $A_1,\ldots,A_{j-1}$ are successful. In particular, one has $\calQ\subseteq S_0$ at the beginning of the $j$-th iteration. 
Let $B_j$ denote the event that condition (2) is satisfied at the beginning of the iteration. Then trivially $Pr[A_j|B_j,A_{<j}]\geq 1/2$. Suppose next that $B_j$ does not hold. In that case, with probability $1/2$, the procedure samples a center $c$. Since by assumption $\sum_{c\in S_0-\calQ}\clcost(c)>\eps \sopt$, the probability that the procedure samples some $c\in S_0-\calQ$ is at least
$$
\frac{\eps\sopt}{\sum_{c'\in S-\calQ}\clcost(c')}\geq \frac{\eps \sopt}{\clcost(\dummyset\cup S-\calQ)}\geq \frac{\eps \sopt}{\clcost(\dummyset\cup S-S_0)}\geq \frac{\eps}{4},
$$
where in the last inequality we used the fact that $\clcost(\dummyset\cup S-S_0)\leq 4\cdot\sopt$ by Lemma \ref{lem:structSminusS0}.

Altogether $Pr[A_j|\overline{B}_j,A_{<j}]\geq \eps/8$. Thus $Pr[A_j|A_{<j}]\geq \eps/8$. Chaining the obtained inequalities we obtain that all the events $A_1,\ldots,A_{|S_0|+1}$ are simultaneously true with probability at least $(\eps/8)^{|S_0|+1}$. Notice that when the latter event happens, the corresponding $\calQ$ satisfies all the properties. Indeed, $\calQ\subseteq S_0$ and furthermore the event $B_j$ must happen for some $j\leq |S_0|+1$ (hence for the next iterations if any) since otherwise $\calQ$ would contain more than $|S_0|$ elements.

Thus, the probability that the overall procedure fails is at most 
$(1-(\frac{\eps}{8})^{|S_0|+1})^{(8/\eps)^{|S_0|+1}\ln n}\leq \frac{1}{n}$.
\end{proof}

\subsubsection{Consistently Assigning Clients in Mixed Solutions and the sets  $\calU$ and $\calR$}
\label{sec:consistentlyassigning}
While our algorithm tries all possible $\calQ \in \calLexp$, we focus  on the execution path of a set $\calQ$ of centers satisfying the properties of Lemma~\ref{lem:successguessprocess}, i.e., 
\begin{enumerate}
    \item 
      $\calQ \subseteq S_0$, and
      \item  the cost in the solution $S - \calQ \cup \dummyset$ of the clusters with centers in $S_0 - \calQ$ is at most $\eps \cdot \sopt$.
\end{enumerate}
    We define the solution $S_{\calQ} := S- \calQ \cup \dummyset$ and for a center $c\in S_{\calQ}$ we let $S_{\calQ}(c)$ be the subset of clients closest to $c$ (i.e., assigned to $c$) in the solution $S_{\calQ}$. The second property above  allows us to consider each cluster $S_{\calQ}(c)$, $c\in S_0 - \calQ$, as contracted by paying a small extra cost of $\eps \cdot \sopt$. In particular, this allows us to give a ``consistent''  assignment $\mu_O$ of clients in the mixed solutions. Here, consistent means that all clients in $S_Q(c)$, for $c\in S_0 - \calQ$, are assigned to the same center in the mixed solutions where all of $S_0$ is removed.  
    
    \paragraph{Definition of $\mu_O$ and $\mu_D$.}
    We modify how the clients are assigned in the solution $M_O$ to obtain the assignment $\mu_O$.
    For every $c \in S_0 - \calQ$, all the clients in $S_{\calQ}(c)$ are reassigned to the same 
    center of $M_O$ as follows. Let $c_1$ be the 
    center of $\ho$ that is the closest to 
    $c$ and let $c_2$ be the center of $M_O-\ho = S - S_0$ that
    is the closest to $c$. The clients of $S_{\calQ}(c)$ are all assigned to either $c_1$ or $c_2$:
    \begin{itemize}
        \item If $d(c,c_1) \le d(c,c_2)/2$, assign
        $S_{\calQ}(c)$ to $c_1$; $S_{\calQ}(c)$ is 
        called a \emph{type-1} cluster.
        \item Otherwise ($d(c,c_1) > d(c,c_2)/2$), assign
        $S_{\calQ}(c)$ to $c_2$; $S_{\calQ}(c)$ is 
        called a \emph{type-2} cluster.
    \end{itemize}
    Let  $T_1$ be the set of clients in a type-1 cluster and $T_2$ the set of clients in a type-2
    cluster. The remaining clients in $\clientsexpensive - (T_1 \cup T_2)$ are assigned to their closest centers in $\ho$, and the remaining clients in $(\clients - \clientsexpensive)- (T_1 \cup T_2)$ are assigned to their closest centers in $S- S_0$.
    Notice that, even neglecting the extra cost due to the consistent assignment, the above assignment of clients in $T_2$ is suboptimal when $d(c,c_2)>d(c,c_1)>d(c,c_2)/2$. The motivation for this reassignment is technical. More specifically, it will be used in Lemma \ref{lemma:costboundofMOandMDwithassignments} to have a better upper bound on the cost of clients in $C^*\cap T_1$ for a non-expensive cluster $C^*$ (case b2). This leads to a larger upper bound on the cost of clients in $C^*\cap T_2$ for an expensive cluster $C^*$ (case a3), which is however tolerable.

    We further obtain an assignment $\mu_D$ of clients in $M_D$ by modifying $\mu_O$ as follows:  each client $p$ with $\mu_O(p) \in \ho$ is assigned to the associated dummy center by $\mu_D$. So for each client $p$ we have $\mu_D(p) = \mu_O(p) $ unless $\mu_O(p)\in \ho$ in which case $\mu_D(p)$ equals the dummy center associated with $\mu_O(p)$.

    \paragraph{Sets $\calR$ and $\calU$.} By the definition of type-1 and type-2 cluster, we can split the
    set of centers of $S_0 - \calQ$ into two groups $\calR$ and $\calU$. Let $\calR$
    be the set of centers of $S_0- \calQ$ whose set of clients is completely assigned
    to a center of $\ho$ in the assignment $\mu_O$, and let $\calU = S_0 - \calQ- \calR$ be the remaining ones that are assigned to centers in $S- S_0$.

    We end this section by analyzing the costs of the assignments $\mu_O$ and $\mu_D$. {Recall that} we use the notation $\clcost(M_O,\mu_O)$ to denote the cost of $M_O$ equipped with assignment $\mu_O$ and, similarly $\clcost(M_D, \mu_D)$ for the cost of $M_D$ equipped with assignment $\mu_D$. We say that the selection of $(W, \calB, \calQ)$ is successful, if the sample $W$ selected in~\ref{sec:dsampleproc} is successful, the set of balls $\calB$ from Section~\ref{sec:ballguesses} is valid, and $\calQ$ selected in this section satisfies the properties of Lemma~\ref{lem:successguessprocess}.
\begin{lemma}
    If the selection of $(W, \calB, \calQ)$ is successful,
    \[
        \clcost(M_O, \mu_O) + \clcost(M_D, \mu_D) \leq 4\cdot \sopt + O(\eps \cdot \sopt)\,. 
    \]
    \label{lemma:costboundofMOandMDwithassignments}
\end{lemma}
\begin{proof}
Throughout the proof of the lemma, we will repeatedly upper bound $d(p,\mu_D(p))$ in terms of $d(p,\mu_O(p))$ using the following bounds:
\begin{claim}\label{clm:costboundofMOandMDwithassignments}
Consider a client $p\in S_{\calQ}(c)$ with $\mu_O(p) \neq \mu_D(p)$ and so $\mu_O(p)\in \ho$ is the center of an expensive cluster $C^*$ of $\opt$. Then
\begin{enumerate}
\item $d(p, \mu_D(p))  \leq d(p, \mu_O(p)) + 2  ( \avg_{C^*, \opt} + \eps \gamma_{C^*})$, and
\item $d(p,\mu_D(p)) \leq d(p, S_\calQ) + 3 d(c, \mu_O(p))$.
\end{enumerate}
\end{claim}
\begin{proof}[Proof of Claim \ref{clm:costboundofMOandMDwithassignments}]
Let $\mu_O(p)  = c^*$. By the assumption that $\calB$ is a valid set of balls, we have by Lemma~\ref{lem:ballguesses} that there is a ball $B(\ell, \rho) \in \calB$ such that $c^* \in B(\ell, \rho)$ and 
\[
    \avg_{C^*, \opt} \leq \rho \leq \avg_{C^*, \opt} + \eps \gamma_{C^*}\,.
\]
Thus, if we let $\delta \in \Lambda$ be the dummy center associated with this ball (at a distance $\rho$ from $\ell$), then $\mu_D(p) = \delta$ and, by the triangle inequality, 
\[
    d(p, \mu_D(p)) \leq d(p, c^*) + d(c^*, \ell) + d(\ell, \delta) \leq d(p, \mu_O(p)) + 2  ( \avg_{C^*, \opt} + \eps \gamma_{C^*})\,.
\]
For the second bound, we again use the triangle inequality:
\begin{align*}
 \notag   d(p, \mu_D(p)) & \leq d(p, c)  + d(c, c^*) + d(c^*, \ell) + d(\ell, \delta) \\
\notag    &\leq d(p, S_\calQ) + d(c, c^*) + 2  ( \avg_{C^*, \opt} + \eps \gamma_{C^*}) \\
    &\leq d(p, S_\calQ) + 3d(c, c^*)\,,
    \label{eq:lastinequality}
\end{align*}
where the last inequality follows from the fact that we must have $d(c, c^*) > (1+\eps) \avg_{C^*, \opt} + \eps \gamma_{C^*}$ since otherwise we would have the contradiction that $c^*\in \ho$ is covered by $S$; recall that $\ho$ only contains expensive centers of $\opt$. 
\end{proof}
Having proven Claim \ref{clm:costboundofMOandMDwithassignments}, we prove the lemma by distinguishing between expensive and non-expensive clusters $C^*$ of $\opt$. 

\paragraph{Case a: $C^*$ is an expensive cluster of $\opt$.} Let $c^*\in \ho$ be the center of $C^*$. 
Let $p\in C^*$ be a client in $C^*$ and let $c$ be its closest center in $S_\calQ$, i.e.,   $p\in S_\calQ(c)$.  We bound $d(p, \mu_O(p)) + d(p, \mu_D(p))$  by considering three cases:
\begin{description}
    \item[Case a1: $p\in C^* - (T_1 \cup T_2)$:]  By definition, we have  $\mu_O(p)$ is the closest center in $\ho$. Moreover, $\mu_O(p)$ equals the cluster center $c^*$ of $C^*$ since $p$ is closest to $c^*$ among all centers in $\opt \supseteq \ho$. Claim \ref{clm:costboundofMOandMDwithassignments}.1 together with  Lemma~\ref{lemma:costsMO} thus gives us,
    \[
        d(p,\mu_O(p)) + d(p, \mu_D(p)) \leq 2\opt_p +  2  ( \avg_{C^*, \opt} + \eps \gamma_{C^*})\,.
    \]
\item[Case a2: $p\in C^* \cap T_1$:]  By definition of type-1 clusters,  $\mu_{O}(p) = c_1$ is the closest center in $\ho$ to $c$. Now, by Claim \ref{clm:costboundofMOandMDwithassignments}.2 and triangle inequality, 
\begin{align*}
d(p, \mu_O(p))  + d(p, \mu_D(p)) & \leq d(p,\mu_O(p)) + d(p, S_{\calQ}) + 3 d(c, c_1) \leq 2 d(p, S_\calQ) + 4 d(c,c_1)\\
& \leq 2 d(p, S_\calQ) + 4 d(c,c^*) \leq 2 d(p, S_\calQ) + 4 \left( d(c, p) + \opt_p\right) \\
& = 6 d(p, S_\calQ) + 4 \opt_p\,.
\end{align*}
\item[Case a3: $p\in C^*\cap T_2$.] By the definition of type-2 clusters, $\mu_O(p) = c_2$ is the closest center to $c$ in $S- S_0$, and $d(c, c_2) \leq 2 d(c, \ho)$. So, by the triangle inequality, $d(c, c_2) \leq 2( d(c,p) + \opt_p)$ where we additionally used that $d(p, \ho) \leq \opt_p$ by Lemma~\ref{lemma:costsMO}. As $\mu_O(p) = \mu_D(p)$ in this case,
\begin{align*}
d(p, \mu_O(p))  + d(p, \mu_D(p)) & = 2d(p,\mu_O(p))   \leq 2( d(p, c) + d(c, c_2)) \\
& \leq 2(d(p,c) + 2d(c, \ho))  \leq 2 d(p, c) + 4( d(c,p) + \opt_p)  \\
& = 6 d(p, S_\calQ) + 4 \opt_p\,.
\end{align*}
\end{description}
By the above bounds, we have that 
\begin{align*}
& \sum_{p\in C^*} (d(p, \mu_O(p)) + d(p, \mu_D(p)))\\
\leq &
\sum_{p\in C^* - (T_1 \cup T_2)} \left(2\opt_p +  2(\avg_{C^*, \opt} + \eps \gamma_{C^*})\right) + \sum_{p\in C^*\cap (T_1 \cup T_2)}(4\opt_p + 6 d(p, S_\calQ))\,.  
\end{align*}
We complete the analysis of the expensive clusters by upper bounding the terms $2\eps\gamma_{C^*}$ and $2\avg_{C^*, \opt}$. By definition we have $|C^* - (T_1 \cup T_2)| \cdot 2\eps \gamma_{C^*}\leq  2\eps \sum_{p\in C^*} (\opt_p + S_p)$. For the other term, we claim that
\[
    2 |C^* - (T_1 \cup T_2)|\cdot \avg_{C^*, \opt} \leq \sum_{p\in C^* - (T_1 \cup T_2)} 2\opt_p + 3\eps \cdot \sum_{p\in C^*} \opt_p\,. 
\]
This holds because there are at least $(1-\eps)|C^*|$ clients  in $C^*$ such that $\opt_p \geq \avg_{C^*, \opt}$, which implies 
\begin{align*}
    \sum_{p\in C^* - (T_1 \cup T_2)} 2\opt_p + 3\eps \cdot \sum_{p\in C^*} \opt_p &\geq 2(|C^* - (T_1 \cup T_2)|-\eps|C^*|)\cdot \avg_{C^*, \opt} + 3\eps(1-\eps)|C^*|\cdot \avg_{C^*, \opt} \\
    & \geq 2 |C^*- (T_1 \cup T_2)|\cdot \avg_{C^*, \opt}\,. 
\end{align*}
In summary, for an expensive cluster $C^*$ we have that 
\begin{align}
& \sum_{p\in C^*} (d(p, \mu_O(p)) + d(p, \mu_D(p)))\nonumber\\
\leq &
\sum_{p\in C^* - (T_1 \cup T_2)} 4\opt_p  + \sum_{p\in C^*\cap (T_1 \cup T_2)}(4\opt_p + 6 d(p, S_\calQ)) + \sum_{p\in C^*} (5 \eps \opt_p  + 2\eps S_p)\,.
\label{eq:expensive_cluster_bound}
\end{align}

\paragraph{Case b: $C^*$ is a non-expensive cluster of $\opt$.} We proceed with a similar analysis. Indeed, consider a client $p\in C^*$  and let $c$ be its closest center in $S_\calQ$, i.e.,   $p\in S_\calQ(c)$.  We again bound $d(p, \mu_O(p)) + d(p, \mu_D(p))$  by considering three cases:

\begin{description}
    \item[Case b1: $p\in C^* - (T_1 \cup T_2)$.] By the definition of $\mu_O$, we have that $\mu_O(p)$ is the closest center in $S - S_0$ to $p$. So $\mu_D(p) = \mu_O(p)$ and  $d(p,\mu_O(p)) + d(p, \mu_D(p)) \leq 2r(p)$, by  Lemma~\ref{lemma:costsMO}. 
\item[Case b2: $p\in C^* \cap T_1$.]  By the definition of type-1 clusters,  $\mu_{O}(p) = c_1$ is the closest center in $\ho$ to $c$. Now, by Claim \ref{clm:costboundofMOandMDwithassignments}.2
and the triangle inequality, 
\begin{align*}
d(p, \mu_O(p))  + d(p, \mu_D(p)) & \leq d(p,\mu_O(p)) + d(p, S_{\calQ}) + 3 d(c, c_1) \\
& \leq 2 d(p, S_\calQ) + 4 d(c,c_1)\,.
\end{align*}
We proceed by upper bounding $d(c, c_1)$. Here we critically use the fact that, by the definition of $T_1$, one has $d(c, c_1) \leq d(c, c_2)/2$ where $c_2$ is the closest center to $c$ in $S- S_0$. Let $c_3$ be the closest center to $p$ in $S - S_0$. Then
$$
d(c, c_1) \leq d(c, c_2)/2\leq d(c,c_3)/2 \leq (d(c,p) + d(p, c_3))/2 = (d(p, S_\calQ) + d(p, c_3))/2.
$$  
Moreover, we have $d(p, c_3) = d(p, S- S_0)\leq r(p)$ by Lemma~\ref{lemma:costsMO} and so
$
    d(c,c_1)  \le(d(p, S_\calQ) + r(p))/2$, which gives us the bound
    \begin{align*}
    d(p, \mu_O(p))  + d(p, \mu_D(p)) \leq 4 d(p, S_\calQ) + 2 r(p)\,.
    \end{align*}
\item[Case b3: $p\in C^*\cap T_2$.] By the definition of type-2 clusters, $\mu_O(p) = c_2$ is the closest center to $c$ in $S- S_0$, and so $\mu_O(p) = \mu_D(p)$. Now by Lemma~\ref{lemma:costsMO} and the triangle inequality, 
$$
d(p, \mu_O(p)) \leq d(p,c) + d(c, c_2) \leq 2d(p,c ) + d(p, S- S_0)\leq  2d(p,c) + r(p) = 2d(p, S_\calQ) + r(p).
$$
So we get the bound
\begin{align*}
    d(p,\mu_O(p)) + d(p, \mu_D(p)) =2 d(p,\mu_O(p))\leq 4 d(p, S_\calQ) + 2 r(p)\,.
\end{align*}
\end{description}
Summing up the above bounds, yields
\[\sum_{p\in C^*} (d(p, \mu_O(p)) + d(p, \mu_D(p))) \leq 
\sum_{p\in C^*} 2 r(p) + \sum_{p\in C^* \cap (T_1 \cup T_2)} 4 d(p, S_\calQ)
\]
for a non-expensive cluster $C^*$ of $\opt$. 
If we sum  up the above inequality for all non-expensive clusters and the bound~\eqref{eq:expensive_cluster_bound} for expensive clusters of $\opt$ we thus get
\[
\clcost(M_O, \mu_O) + \clcost(M_D, \mu_D) \leq \sum_{p\in \clientsexpensive} 4 \opt_p + \sum_{p\in \clients- \clientsexpensive} 2 r(p) + O( \eps \cdot \sopt)\,,
\]
where we used that $\sum_{c \in S_0 - \calQ} \sum_{p \in S_{\calQ}(c)} d(p, S_\calQ) = \sum_{p\in T_1 \cup T_2} d(p, S_\calQ) \leq \eps \cdot \sopt$ (Lemma~\ref{lem:successguessprocess}) and $\clcost(S) \leq 5 \sopt$ (Lemma~\ref{lemma:localsearchis5approximation}) to bound the error term $O(\eps \cdot \sopt)$.

The proof is now concluded by Lemma~\ref{lemma:convenience_upper_bound}, which says that $\sum_{p\in \clients - \clientsexpensive} r(p) \leq \sum_{p\in \clients - \clientsexpensive} 2\opt_p + O(\eps \cdot \sopt)$ if the sample $W$ is successful.

\end{proof}

\subsection{Removal of Cheap Centers in $S_0$}
\label{sec:removalofCheap}

At this stage we have defined a clustering $S_{\calQ} = S - \calQ \cup \dummyset$, where $\dummyset$ is the set of dummy centers. Recall that the mixed solution $M_O$ equals $S-S_0 \cup \ho$. So $S_\calQ - \Lambda$ represents progress compared to $S$ in that we have already removed a subset $\calQ$ of $S_0$. Furthermore, with respect to the assignment of clients $\mu_O$, we have that the remaining centers $S_0 - Q$, which we want to remove, are divided into two sets $\calR$ and $\calU$. The clusters with centers in $\calR$ are completely assigned (by $\mu_O$) to centers in $\ho$, and every cluster with center in $\calU$ is completely assigned to a center in $S - S_0$. 
The task of this section is to make further progress by guessing the centers in $\calU$ and their reassignment to centers in $S- S_0$. This turns out to be a difficult task, and we will instead ``approximately'' guess a set $\bcalU$ and reassignment $\tmu$ of the clients of those clusters. Specifically, we will output a solution $S_{\calQ \cup \bcalU} = S_{\calQ} - \bcalU$ (obtained by removing $\bcalU$ from $S_{\calQ}$) and an assignment $\tmu$ of clients to centers in $S_\calQ-\bcalU$. This is done through the following lemma.   {Recall that for $c\in S_{\calQ}$, $S_{\calQ}(c)$ is the set of clients assigned to $c$ in $S_\calQ$, i.e., closest to $c$ among all centers in $S_\calQ$.}    
\begin{lemma}\label{lem:successcheapremove}
Given $S$ and $\calQ$ as described above, there is a polynomial-time algorithm that produces a collection $\calL_{cheap}$ of subsets $\bcalU$ of $S-\calQ$, and for each such $\bcalU$ and assignment $\tmu$ of clients to centers in $S_{\calQ \cup \bcalU} = S_\calQ-\bcalU$, such that at least one such pair $(\bcalU,\tmu)$ satisfies the following properties:
\begin{enumerate}
    \item $|\bcalU| = |\calU|$.
    \item $\bcalU \cap \calR = \emptyset$, i.e., $\bcalU$ does not contain any center of $S_0 - \calQ$ whose set of clients is completely assigned
    to a center of $\ho$ in the assignment $\mu_O$.
    \item $\tmu$ satisfies: \begin{enumerate}
        \item For every center $c\in S_{\calQ} - \bcalU$, we have $\tmu(p)= c$ for every $p\in S_{\calQ}(c)$.
        \item For a center $c\in \bcalU$, the clients in $S_\calQ(c)$ are reassigned to a center $c' \in S- \calQ - \bcalU - \calR$, i.e., $\tmu(p) = c'$ for every $p\in S_{\calQ}(c)$.
       \item The cost increase of the reassignment $\tmu$ of clients previously assigned to $\bcalU$ compared to that  of the reassignment $\mu_O$ of clients previously assigned to $\calU$ is bounded by $O(\eps \cdot \sopt)$:
       \begin{gather*}
        \sum_{c \in \calU \cup \bcalU}  \sum_{p\in S_\calQ(c)} \left( d(p, \tmu(p)) - d(p, \mu_O(p)) \right) = O (\eps \cdot \sopt)\,. 
        \end{gather*}
    \end{enumerate}
\end{enumerate}
\end{lemma}
In words, the above properties of $\tmu$ say that we maintain the assignment of clients whose closest center remains the same, and the clients associated to the removed centers $\bcalU$ are reassigned to centers not in $\calR$ at no higher cost (up to $O(\eps \cdot \sopt)$) than the cost of the reassignment of $\calU$ in the modified solution $M_O$.

Let $\calX$ be the centers in $S-S_0$ to which the points $p\in S_{\calQ}(c)$ with $c\in \calU$ are assigned according to $\mu_O$. Notice that $|\calX|\leq |\calU|$ and $\calX \cap (\calQ\cup \calU\cup \calR)=\emptyset$.
We assume that the values of $|\calU|$, $|\calR|$, and $|\calX|$ are known. This can be achieved by trying, {for each such value}, all the integers between $0$ and $|S_0|\leq \frac{\log n}{\eps^3}$. Let $\ell:=|\calU|+|\calR|+|\calX|\leq 2|\calU|+|\calR|\leq 2|S_0 - \calQ|$. Then we run the recursive procedure $cheapRem()$ described in the box with input $(\calU',\calR',\calX',\calN,\bcalU)=(\emptyset,\emptyset,\emptyset,\emptyset,\emptyset)$. We assume that $cheapRem()$ has access to the quantities $S$, $\calQ$, $|\calU|$, $|\calR|$, $|\calX|$ and $\ell$, as well as to a global variable $\calL_{cheap}$, which is initialized to $\emptyset$. The intuition for the parameters is as follows. Intuitively $\calU'$, $\calR'$ and $\calX'$ are subsets of $\calU$, $\calR$ and $\calX$, resp., that we have already identified, while $\bcalU$ is the current value of the set $\bcalU$ under construction. Intuitively, $\calN$ are the centers to which the points $p\in S_{\calQ}(c)$ with $c\in \bcalU$ are reassigned, namely $\tmu(p)\in \calN$. At the end of the root call, the global variable $\calL_{cheap}$ contains the desired collection of sets $\bcalU$ as in Lemma 
\ref{lem:successcheapremove}. Whenever some $\bcalU$ is added to $\calL_{cheap}$, we define a corresponding $\tmu$ as follows: For each $p\in S_{\calQ}(c)$, $\tmu(p)=c$ if $c\in S_{\calQ}-\bcalU$. Otherwise, i.e. if $c\in \bcalU$, $\tmu(p)=next(c)\in \calN$, where $next(c)$ is defined in the recursive call when $c$ is added to $\bcalU$.

\begin{mdframed}[hidealllines=true, backgroundcolor=gray!15]
\vspace{-5mm}
\paragraph{$cheapRem(\calU',\calR',\calX',\calN,\bcalU)$}\ \\
\begin{algorithmic}[1]
\If{$|S-\calQ|\leq 4\ell$}\label{alg:cheapRem:fewCenters}
\For{All $\bcalU\subseteq S-\calQ$ of size $|\calU|$ and $\bcalR\subseteq S-\calQ-\bcalU$ of size $|\calR|$}
\State Add $\bcalU$ to $\calL_{cheap}$ and, for each $c\in \bcalU$, set $next(c)$ to the closest center to $c$ in $S-\calQ-\bcalU-\bcalR$\label{alg:cheapRem:optimalGuess}
\EndFor
\State \textbf{halt}
\EndIf
\If{$|\bcalU|=|\calU|$}\label{alg:cheapRem:termination}
\State Add $\bcalU$ to $\calL_{cheap}$ and \textbf{halt}
\EndIf
\State Let $c\in S-\calQ-\calR'-\calX'-\calN-\bcalU$ minimize $R(c):=\sum_{p\in S_{\calQ}(c)}d(p,next(c))$, where $next(c)$ is the closest center to $c$ in $S-\calQ-\calU'-\calR'-\bcalU-\{c\}$.\label{alg:cheapRem:selectc}
\If{$c\notin \calU'$ and $|\calX'|<|\calX|$} $cheapRem(\calU',\calR',\calX'\cup \{c\},\calN,\bcalU)$\label{alg:cheapRem:discoveredX}
\EndIf
\If{$c\notin \calU'$ and $|\calR'|<|\calR|$} $cheapRem(\calU',\calR'\cup \{c\},\calX',\calN,\bcalU)$\label{alg:cheapRem:discoveredR}
\EndIf
\If{$|\calU'|+|\calR'|=|\calU|+|\calR|$ or $next(c)\in \calN$} 
$cheapRem(\calU',\calR',\calX',\calN\cup \{next(c)\},\bcalU\cup \{c\})$\label{alg:cheapRem:easyAddtU}
\Else
\If{$|\calR'|<|\calR|$} $cheapRem(\calU',\calR'\cup \{next(c)\},\calX',\calN,\bcalU)$ \label{alg:cheapRem:guessNextR}
\EndIf
\If{$|\calU'|<|\calU|$} $cheapRem(\calU'\cup \{next(c)\},\calR',\calX',\calN,\bcalU)$ \label{alg:cheapRem:guessNextU}
\EndIf
\State $cheapRem(\calU',\calR',\calX',\calN\cup \{next(c)\},\bcalU\cup \{c\})$\label{alg:cheapRem:guessNextNotRU}
\EndIf
\end{algorithmic}
\end{mdframed}

    \begin{proof}[Proof of Lemma \ref{lem:successcheapremove}]    
Consider the execution of the above procedure for the correct guess of the values $|\calR|$, $|\calX|$ and $|\calU|$. 
We remark that, when line \ref{alg:cheapRem:selectc} is executed, the set $S-\calQ-\calR'-\calX'-\calN-\bcalU$ is not empty (so that we can choose {an} appropriate $c$). Indeed, at each recursive call we add at most two elements to $\calR'\cup \calX'\cup \calN\cup \bcalU$  
and the value of $|\calU'|+|\calR'|+|\calX'|+|\bcalU|$ grows by at least $1$. The latter value cannot grow more than $2|\calU|+|\calR|+|\calX|\leq 2\ell$ times since at that point we would have necessarily $|\bcalU|=|\calU|$, which makes the condition of line \ref{alg:cheapRem:termination} true. Since the condition of line \ref{alg:cheapRem:fewCenters} cannot hold if line \ref{alg:cheapRem:selectc} is executed at least once, we have that $|S- \calQ|> 4\ell$. Thus $S-\calQ$ contains sufficiently many elements to remove up to $2$ elements for up to $2\ell$ times. 

Concerning the running time of the procedure, each recursive step involves a branching on at most $4$ subproblems, and in each one of them the value of $|\calU'|+|\calR'|+|\calX'|+|\bcalU|$ grows by at least $1$. Thus, by the same argument as before, the depth of the recursion is at most $2\ell$. So the number of recursive calls is at most $4^{2\ell}=n^{{1/\eps^{O(1)}}}$, implying a polynomial running time.

It remains to show that at least one set $\bcalU\in \calL_{cheap}$ (with the associated $\tilde{\mu}$) at the end of the procedure satisfies the claim. If the condition of line \ref{alg:cheapRem:fewCenters} holds, it must happen that in one of the executions of line \ref{alg:cheapRem:optimalGuess} one has $\bcalU=\calU$ and $\bcalR=\calR$. In that case $\bcalU$ and the corresponding $\tmu$ trivially satisfy the claim. In particular we remark that for each $p\in S_{\calQ}(c)$ with $c\in \bcalU=\calU$, $\tmu(p)=next(c)=\mu_{O}(p)$. 

Otherwise, it is sufficient to show that there exists a chain of recursive calls that leads to a solution $\bcalU$ with the desired properties. In more detail, let us focus on the solution $\bcalU\in \calL_{cheap}$ which is generated by the following chain of recursive calls starting from the root call with $(\calU',\calR',\calX',\calN,\bcalU)=(\emptyset,\emptyset,\emptyset,\emptyset,\emptyset)$. Consider the current input parameters $(\calU',\calR',\calX',\calN,\bcalU)$ and let $c$ be selected in line \ref{alg:cheapRem:selectc}. During the process, we will maintain the following invariants:
\begin{equation}\label{inv:UpRpXp}
\calU'\subseteq \calU;\,\calR'\subseteq \calR; \calX'\subseteq \calX.    
\end{equation}
\begin{equation}\label{inv:N}
\calN \cap (\calU \cup \calR)=\emptyset.    
\end{equation}
\begin{equation}\label{inv:tU}
\bcalU \cap (\calX \cup \calR\cup \calN)=\emptyset. 
\end{equation}
The root call satisfies the mentioned invariants trivially since $\bcalU=\calU'=\calR'=\calX'=\calN=\emptyset$.

If $c\in \calX$, we continue with the recursive call of line \ref{alg:cheapRem:discoveredX}. The invariants are maintained since $\calX'\cup\{c\}\subseteq \calX$ and $c\notin \bcalU$. Observe that this choice is excluded if $c\in \calU'$ or $|\calX'|=|\calX|$, however by Invariant \eqref{inv:UpRpXp} the latter conditions cannot happen when $c\in \calX$.

Similarly, if $c\in \calR$, we continue with the recursive call of line \ref{alg:cheapRem:discoveredR}. The invariants are maintained since $\calR'\cup\{c\}\subseteq \calR$, $c\notin \calN$, and $c\notin \bcalU$. Observe that this choice is excluded if $c\in \calU'$ or $|\calR'|=|\calR|$, however by Invariant \eqref{inv:UpRpXp} the latter conditions cannot happen when $c\in \calR$.

Otherwise, namely if $c\notin \calR\cup \calX$, if the condition of line \ref{alg:cheapRem:easyAddtU} is true, we continue with the recursive call of the same line. Let us show that invariants are preserved. Suppose first that $next(c)\in \calN$. In this case the first two invariants are not affected, and the third one is maintained since $c\notin \calN$ (and by assumption $c\notin \calR\cup \calX$). Otherwise one has $|\calU'|+|\calR'|=|\calU|+|\calR|$, which by Invariant \eqref{inv:UpRpXp} implies $\calU'=\calU$ and $\calR'=\calR$. Thus one has $next(c)\notin \calU\cup \calR$ and $c\notin \calR\cup \calX\cup \calN$.

Otherwise, i.e., if the condition of line \ref{alg:cheapRem:easyAddtU} is false, depending on whether $next(c)$ belongs to $\calR$, $\calU$, or none of the previous cases, we continue with the recursive calls of lines \ref{alg:cheapRem:guessNextR}, \ref{alg:cheapRem:guessNextU}, and \ref{alg:cheapRem:guessNextNotRU}, resp.
Suppose first that $next(c)\in \calR$. Observe that Invariant \eqref{inv:UpRpXp} guarantees that in this case $|\calR'|<|\calR|$ as required by the condition of line \ref{alg:cheapRem:guessNextR}. Invariant \eqref{inv:UpRpXp} is maintained since $\calR'\cup \{next(c)\}\subseteq \calR$, while Invariants \eqref{inv:N} and \eqref{inv:tU} are maintained since $\calN$ and $\bcalU$ are not modified.
Suppose next that $next(c)\in \calU$. Observe that Invariant \eqref{inv:UpRpXp} guarantees that in this case $|\calU'|<|\calU|$ as required by the condition of line \ref{alg:cheapRem:guessNextU}. Similarly to the previous case, Invariant \eqref{inv:UpRpXp} is maintained since $\calU'\cup \{next(c)\}\subseteq \calU$, while Invariants \eqref{inv:N} and \eqref{inv:tU} are maintained since $\calN$ and $\bcalU$ are not modified.
The remaining case if that $next(c)\notin \calR\cup \calU$. Invariant \eqref{inv:N} is maintained since $next(c)\notin \calR\cup \calU$. Invariant \eqref{inv:tU} is maintained since $c\notin \calX\cup \calR$ by the assumptions of this case, $c\notin \calN\cup \{next(c)\}$ by construction, and $c\notin \bcalU$ by construction.

Let us show that the final $\tilde{U}$ obtained with the above procedure with the associated $\tilde{\mu}$ satisfies the claim. Condition (1) is trivially satisfied. Invariant \eqref{inv:tU} directly implies property (2). We next focus on property (3). Property (a) is satisfied by definition. When a center $c$ is added to $\bcalU$, the respective $c':=next(c)=\tmu(p)$ for all $p\in S_{\calQ}(c)$ is added to $\calN$ (if not already there). Invariant \eqref{inv:N} guarantees that $c'\notin \calR$, while Invariant \eqref{inv:tU} guarantess that $c'\notin \bcalU$. Property (b) follows.

It remains to prove property (c). Consider first any $c\in \calU\cap \bcalU$. We claim that for the corresponding $p\in S_{\calQ}(c)$ one has $\tilde{\mu}(p)=\mu_O(p)$. Indeed, let $c_1,c_2,\dots$ be the centers in $S- \calQ- \{c\}$ in non-decreasing order of distance from $c$. Let also $c_q$ be the first such center belonging to $S- S_0$. In particular $\mu_O(p)=c_q$. By construction all the centers $c_1,\ldots,c_{q-1}$ must have been already added to $\calU'\cup \calR'$ in previous recursive steps when $c$ is added to $\bcalU$. Furthermore, when the latter event happens either $c_q\in \calN$ already or $c_q$ is added to $\calN$. In both cases one has that $\tilde{\mu}(p)=c_q$ as desired. As a consequence
\begin{align}
\sum_{c\in \calU\cap \bcalU}\sum_{p\in S_{
\calQ}(c)}(d(p,\tilde{\mu}(p))-d(p,\mu_O(p)))=0.\label{lem:successcheapremove:eqn1}
\end{align}
Consider next any $c\in \calU- \bcalU$. For any $p\in S_{\calQ}(c)$ by construction one has $\tilde{\mu}(p)=c$. 
Thus
\begin{align}    
\sum_{c\in \calU- \bcalU}\sum_{p\in S_{
\calQ}(c)}(d(p,\tilde{\mu}(p))-d(p,\mu_O(p))) & \leq \sum_{c\in S_0- \calQ}\sum_{p\in S_{
\calQ}(c)}d(p,c) -\sum_{c\in \calU- \bcalU}\sum_{p\in S_{
\calQ}(c)}d(p,\mu_O(p))
\nonumber \\
& \leq O(\eps) opt-\sum_{c\in \calU- \bcalU}\sum_{p\in S_{
\calQ}(c)}d(p,\mu_O(p)),\label{lem:successcheapremove:eqn2}
\end{align}
where in the last inequality above we used the assumption that $\calQ$ satisfies property (2) of Lemma \ref{lem:successguessprocess}, and the fact that $\calU- \bcalU\subseteq \calU\subseteq S_0-\calQ$.

Finally consider any $c\in \bcalU- \calU$. Define any bijection between each such $c$ and some distinct $c'\in \calU- \bcalU$ (this is possible since $|\bcalU|=|\calU|$). Consider the recursive call when $c$ is added to $\bcalU$. Let $\calU'$, $\calR'$, $\calX'$, $\calN$, $next(c)$, $R(c)$, $next(c')$, $R(c')$ be the associated quantities in that call. Notice that at that time $c'$ was an available candidate to be added to $\bcalU$ since $c'\in S-\calQ-\calR-\calX-\bcalU\subseteq S-\calQ-\calR'-\calX'-\bcalU$ (using Invariant \eqref{inv:UpRpXp})
and $\calN\cap \calU=\emptyset$ by Invariant \eqref{inv:N}, hence $c'\in S-\calQ-\calR'-\calX'-\calN-\bcalU$. Since we added $c$ instead of $c'$, it must be the case that $R(c)\leq R(c')$. 
For every $p\in S_{\calQ}(c')$, $\mu_O(p)\in \calX$, hence $\mu_O(p)\notin \bcalU$ by Invariant \eqref{inv:tU}. It follows that $\mu_O(p)\in S-\calQ-\calU- \calR- \bcalU-\{c'\}\subseteq S-\calQ-\calU'- \calR'- \bcalU-\{c'\}$, where we used again Invariant \eqref{inv:UpRpXp}. Thus $d(c',next(c'))\leq d(c',\mu_O(p))$. Hence for every $p\in S_{\calQ}(c')$, 
$$
d(p,next(c')) \leq d(p,c')+d(c',next(c'))\leq d(p,c')+d(c',\mu_O(p))\leq 2d(p,c')+d(p,\mu_O(p)).
$$
Summarizing,
\begin{align}
& \sum_{c\in \bcalU- \calU}\sum_{p\in S_{\calQ}(c)}(d(p,\tilde{\mu}(p))-d(p,\mu_O(p))) \leq \sum_{c\in \bcalU- \calU}\sum_{p\in S_{\calQ}(c)}d(p,\tilde{\mu}(p))=
\sum_{c\in \bcalU- \calU}R(c)\leq \sum_{c'\in \calU- \bcalU}R(c')\nonumber \\
& =
\sum_{c'\in \calU- \bcalU}\sum_{p\in S_{\calQ}(c')}d(p,next(c'))
 \leq \sum_{c'\in \calU- \bcalU}\sum_{p\in S_{\calQ}(c')}(2d(p,c')+d(p,\mu_O(p)))\nonumber\\
 & \leq O(\eps) opt+ \sum_{c\in \calU- \bcalU}\sum_{p\in S_{\calQ}(c)}d(p,\mu_O(p)),\label{lem:successcheapremove:eqn3}
\end{align}
where in the last inequality above we used again the assumption on 
$\calQ$, i.e., property (2) of Lemma \ref{lem:successguessprocess}.
Property (c) follows by summing \eqref{lem:successcheapremove:eqn1}, \eqref{lem:successcheapremove:eqn2}, and \eqref{lem:successcheapremove:eqn3}.

\end{proof}

\subsubsection{Mixed Solutions after the Removal of $\bcalU$}    

We let $(S_{\calQ \cup \bcalU}, \tmu)$ be the output of a successful run of the $cheapRem()$ procedure, i.e., where $\bcalU \in \calL_{cheap}$ and its associated assignment $\tmu$ satisfies the properties of Lemma~\ref{lem:successcheapremove}. Recall that $S_{\calQ \cup \bcalU}$ consists of the centers $S' = S - \calQ - \bcalU$ and the dummy centers $\dummyset$.
Further recall that $M_O$ consists of the centers $S- S_0 = S - \calQ - \calU - \calR$ and $\ho$.  $M_D$ is the same set of centers except that $\ho$ is replaced by $\dummyset$. Notice that some centers of $M_O$ and $M_D$ are now removed if $\bcalU - \calU \neq \emptyset$. To take care of this, we modify these mixed solutions to obtain $M'_O$ and $M'_D$. 
 
 We first define $M'_O$; the definition of $M'_D$ is then very similar. The centers of $M'_O$ are $S- \calQ - \bcalU
 - \calR$ and $\ho$. Hence, the difference between $M_O$  and $M'_O$  is that in $M_O$ we remove $\calU$ from $S$ and in $M'_O$ we remove $\bcalU$. In other words,  $M'_O = M_O \cup \calU - \bcalU$. Similarly, we let $M'_D = M_D \cup \calU - \bcalU$. We update the assignment $\mu_O$ of $M_O$ to an assignment $\mu'_O$ of $M'_O$ (and the assignment $\mu_D$ to $\mu'_D$). 
 Recall that $\mu_O$ assigns each client $p$ to its closest center in $\ho$ (if $p\in \clientsexpensive$) or its closest center in $S- S_0$ (if $p\in \clients- \clientsexpensive$) except for those clients that belong to a cluster $S_{\calQ}(c)$ with $c\in \calU\cup \calR$. Indeed, the clients in $S_\calQ(c)$ are all assigned to the center of $\ho$ that is closest to $c$  if $c\in \calR$, and if $c\in \calU$ they are all assigned to the center in $S - S_0$ that is closest to $c$. We also recall that $\mu_D$ is the same as $\mu_O$ except when $\mu_O(p) \in \ho$ in which case $\mu_O(p)$ is replaced by its corresponding dummy center.

 \paragraph{Definitions of $\mu_O'$ and $\mu_D'$.}
 For clients $p\in S_{\calQ}(c)$ with $c \not\in \calU \cup \bcalU \cup \dummyset$, we define $\mu_O'(p) = \mu_O(p)$. For $p\in S_{\calQ}(c)$  with $c\in \dummyset$, we let $\mu'_O(p)$ be the closest center in $\ho$.
 Finally, for $p\in S_{\calQ}(c)$ with $c\in \calU \cup \bcalU$, we define $\mu'_O(p) =  \tmu(p)$. 

 The assignment $\mu'_D$ is obtained in the same way from $\mu_D$ with the difference that we use $\dummyset$ instead of $\ho$: 
 For clients $p\in S_{\calQ}(c)$ with $c \not\in \calU \cup \bcalU \cup \dummyset$, we define $\mu_D'(p) = \mu_D(p)$. For $p\in S_{\calQ}(c)$  with $c\in \dummyset$, we let $\mu'_D(p)$ be the closest center in $\dummyset$.
 Finally, for $p\in S_{\calQ}(c)$ with $c\in \calU \cup \bcalU$, we define $\mu'_D(p) =  \tmu(p)$. 
 
 This completes the definition of $\mu'_O$ and $\mu'_D$.
  We remark that we have $\mu'_O(p) = \mu_O(p)$ and $\mu'_D(p) = \mu_D(p)$  for all $p\in S_\calQ(c)$ with $c\not \in \calU \cup \bcalU \cup \dummyset$.
    
    We continue by arguing that $\mu'_O$ is well-defined, i.e., that $\mu'_O(p) \in M_O'$ for every $p\in \clients$ (the proof for $\mu'_D$ is the same).
 This is immediate for a client $p\in S_\calQ(c)$ with $c \in \dummyset$. For a client $p\in S_{\calQ}(c)$ with $c\in \cal U \cup \bcalU$, it holds because, by Lemma~\ref{lem:successcheapremove}, $\tmu(p)$ equals $c\in M'_O$ if  $c \in \calU -  \bcalU$ and otherwise $\tmu(p) \in  S- \calQ - \bcalU - \calR\subseteq M'_O$.
 It remains to verify that $\mu_O(p) \in M_O'$ for a client $p\in S_\calQ(c)$ with $c\not \in \calU \cup \bcalU \cup \dummyset$. If $c\in \calR$, we have $\mu_O(p) \in \ho \subseteq M'_O$.  Similarly, if $c \not \in \calU \cup \bcalU \cup \Lambda \cup \calR $ and $p\in \clientsexpensive$ we have $\mu_O(p) \in \ho \subseteq M'_O$.  In the remaining case when $c \not \in \calU \cup \bcalU \cup \Lambda \cup \calR $ and $p\in \clients- \clientsexpensive$, $\mu_O(p)$ equals $p$'s closest center in $S- S_0$. As   $p\in S_\calQ(c)$, we thus have   $\mu_O(p) = c \in S - (S_0 \cup \bcalU) \subseteq M'_O$.

 The following upper bound on $\clcost(M'_O, \mu'_O) + \clcost(M'_D, \mu'_D)$ is a fairly immediate consequence of the assumption that the cheap-removal process is successful (Property 3c of Lemma \ref{lem:successcheapremove}). 
 Similarly to before, we say that the selection of $(W, \calB, \calQ, \bcalU,\tmu)$ is successful, if the sample $W$ selected in~\ref{sec:dsampleproc} is successful, the set of balls $\calB$ from Section~\ref{sec:ballguesses} is valid,  $\calQ$ selected in Section~\ref{sec:removalofExpensive} satisfies the properties of Lemma~\ref{lem:successguessprocess}, and $\bcalU,\tmu$  selected in this section satisfies the properties of Lemma~\ref{lem:successcheapremove}. 
 \begin{claim}
     If $(W, \calB, \calQ, \bcalU,\tmu)$ is successful,
     $$\clcost(M'_O, \mu'_O) + \clcost(M'_D, \mu'_D) \leq \sum_{p\in \clients} 4 \opt_p + O(\eps \cdot \sopt)\,.$$
     \label{claim:Mprimebounds}
 \end{claim}
 \begin{proof}
     For $x\in \{O, D\}$ the only differences between $\mu'_x$ and $\mu_x$ are clients in $S_{\calQ}(c)$ with $c\in \dummyset \cup \calU \cup \bcalU$.  Consider first a client $p \in S_{\calQ}(c)$  with $c \in \dummyset$. Then, as $c$ is the closest center among $(S - \calQ) \cup \dummyset$, which both contains $M_O - \ho$ and $M'_O - \ho$, we have that the closest center to $p$ in both $M'_O$ and $M_O$ is in $\ho$ (because by the definition of dummy centers, $d(p, \ho) \leq d(p, \dummyset)$). Similarly, the closest center to $p$ in both $M'_D$ and $M_D$ is in $\Lambda$. Hence, the definitions of $\mu'_O$ and $\mu'_D$ to assign $p$ to its closest center in $\ho$ and $\dummyset$, respectively, cannot increase the cost, i.e., $d(p, \mu'_O(p)) \leq d(p, \mu_O(p))$ and $d(p, \mu'_D(p)) \leq d(p, \mu_D(p))$ for such a client $p$. 
     
     Finally, for those clients $p\in S_\calQ(c)$ with $c\in\calU \cup \bcalU$ we have $\mu'_x = \tmu$ by definition and $d(p, \mu_O(p)) \leq d(p, \mu_D(p))$ since $\mu_O(p) = \mu_D(p)$ unless $\mu_O(p)\in \ho$ in which case $\mu_D(p)$ is the dummy center associated with $\mu_O(p)$ that can only be farther away from $p$ than $\mu_O(p)$. Hence:
       \begin{gather*}
       \sum_{c \in \calU \cup \bcalU}  \sum_{p\in S_\calQ(c)} \left( d(p, \tmu(p)) - d(p, \mu_D(p)) \right) \leq \sum_{c \in \calU \cup \bcalU}  \sum_{p\in S_\calQ(c)} \left( d(p, \tmu(p)) - d(p, \mu_O(p)) \right) = O (\eps \cdot {\sopt})\,, 
        \end{gather*}
     where in the equality we used Property (c)  of Lemma~\ref{lem:successcheapremove}. The claim follows since
$$
\clcost(M'_O, \mu'_O) + \clcost(M'_D, \mu'_D)\leq \clcost(M_O, \mu_O) + \clcost(M_D, \mu_D) +O(\eps\cdot \sopt) \overset{Lem. \ref{lemma:costboundofMOandMDwithassignments}}\leq 4 \opt_p + O(\eps \cdot \sopt)\,.
$$     
\end{proof}

To better understand $\clcost(M'_O, \mu'_O)$, let us consider the cost $d(p, \mu'_O(p))$ of a single client $p$:
\begin{itemize}
    \item If $p \in S_{\calQ}(c)$ with $c\not\in \bcalU \cup \calR$ then $\tmu(p) = c$ by Property 3a of Lemma~\ref{lem:successcheapremove}.  Moreover,  $\mu'_O(p)$ assigns $p$ either to a center in $\ho$ or to one in $S - \calQ - \bcalU - \cal R$ and we have $d(p, c) \leq d(p, S - \calQ - \bcalU - \cal R)$.  Hence, in either case, we have
    \[
        d(p, \mu'_O(p)) \geq  d(p, \{\tmu(p)\} \cup \ho) = d(p, \{\tmu(p)\} \cup \ho \cup \dummyset) \,.
    \]
    \item If $p\in S_{\calQ}(c)$ with $c\in \calR$ then $\mu'_O(p) = \mu_O(p) =  c^*\in \ho$ where $c^*$ is the center of $\ho$ that is closest to $c$. So 
    \[
        d(p, \mu_O'(p)) = d(p, c^*)\,. 
    \]
    \item If $p \in S_{\calQ}(c)$ with $c\in \bcalU$ then $\mu'_O(p) = \tmu(p)$ and so
    \[
        d(p, \mu_O'(p)) \geq d(p, \{\tmu(p)\} \cup \ho \cup \dummyset).
    \]
\end{itemize}
Similarly, we can analyze $d(p, \mu_D(p))$ by replacing $\ho$ with the set $\dummyset$ of dummy centers (and $c^*$ by its associated dummy center). 
Summarizing, we have 
\begin{align*}
     \clcost(M'_O, \mu'_O) &\geq \sum_{c\in S_{\calQ} - \calR}\sum_{p\in S_{\calQ}(c)}  d\left(p, \{c\} \cup \ho \cup \dummyset\right) + \sum_{c\in \calR} \min_{c' \in \ho} \sum_{p\in S_{\calQ}(c)}  d\left(p, c' \right)\\
     \intertext{and}
     \clcost(M'_D, \mu'_D) &\geq \sum_{c\in S_{\calQ} - \calR}\sum_{p\in S_{\calQ}(c)}  d\left(p, \{c\}  \cup \dummyset\right) + \sum_{c\in \calR} \min_{c'\in \dummyset}\sum_{p\in S_{\calQ}(c)}  d\left(p,   c'\right)
\end{align*}

 Furthermore, by Properties 3a and 3b of Lemma~\ref{lem:successcheapremove}, we have that, for every $c\in \calR$, the set $S_{\calQ}(c)$ of clients assigned to $c$ in $S_{\calQ}$ equals the set $\tmu^{-1}(c)$ of clients assigned by $\tmu$. So, if we let $\{C_c\}_{c\in S_{\calQ \cup \bcalU}}$ be the partitioning of the set $\clients$ of clients according to $\tmu$, i.e., $C_c = \{p\in \clients \mid \tmu(p) =c\}$,  we can thus rewrite the above  bounds on the cost as 
\begin{align}
     \clcost(M'_O, \mu'_O) &\geq \sum_{c\not \in \calR}\sum_{p\in C_c}  d\left(p, \{c\} \cup \ho \cup \dummyset\right) + \sum_{c\in \calR}\min_{c' \in \ho}\sum_{p\in C_c}  d\left(p,  c'\right) \label{eq:MprimeObound}\\
     \intertext{and}
     \clcost(M'_D, \mu'_D) &\geq \sum_{c \not\in  \calR}\sum_{p\in C_c}  d\left(p, \{c\}  \cup \dummyset\right) + \sum_{c\in \calR}\min_{c' \in \dummyset}\sum_{p\in C_c}  d\left(p,  c'\right)\,.\label{eq:MprimeDbound}
\end{align}
In the next section, we use these bounds to give a polynomial-time algorithm that outputs a solution $S^*$ with $k$ centers so that (see Lemma~\ref{lemma:findingSstar})
\begin{align*}
    \clcost(S^*)  \leq
    \frac{(1+2\eps)}{2} \left(  \clcost(M'_O, \mu'_O) +  \clcost(M'_D, \mu'_D) \right) + O(\eps \sopt)
     \leq \sum_{p\in \clients} 2 \opt_p  + O(\eps \cdot \sopt)
\end{align*}
where the second inequality is by Claim~\ref{claim:Mprimebounds}.
So the proof of Lemma~\ref{lemma:findingSstar}  in the next subsection is the final step in the proof of Theorem~\ref{thm:mainadditivecenters} (see also Section~\ref{sec:stable:everythingtogether} where we put everything together to prove Theorem~\ref{thm:mainadditivecenters}).

\subsection{Finding $S^*$ via Submodular Optimization}
\label{sec:submodularopt}

In this section, we give a polynomial-time algorithm for finding the solution $S^*$ by reducing the problem to maximizing a submodular function subject to a partition matroid constraint. Specifically, we prove the following lemma:
\begin{lemma}
    We can in polynomial-time find a clustering $S^*$ with $k$ centers of cost at most 
    \[
    \frac{(1+2\eps)}{2} \left(  \clcost(M'_O, \mu'_O) +  \clcost(M'_D, \mu'_D) \right) + O(\eps \sopt)\,.
     \]
    \label{lemma:findingSstar}
\end{lemma}
Recall that at this point the algorithm calculated the local search solution $S$, sampled $W$, guessed $\calB$, $\calQ$, and $\bcalU$ with the assignment $\tmu$ of the clustering $S_{\calQ \cup \bcalU}$. We assume that all these choices were successful guesses, i.e., that $(W, \calB, \calQ, \bcalU, \tmu)$ is successful.  Further recall the notation that we partition the clients into the clusters $\{C_c\}_{c\in S_{\calQ \cup \bcalU}}$ where $C_c = \{ p\in \clients : \tmu(p) =c\}$.  

Our goal is to give a polynomial-time algorithm that finds approximations of the sets $\ho$ and $\calR$ that only have slightly worse cost. We start by defining the feasible set of candidates for $\ho$ as a partition matroid. Recall that $\ho$ contains one center from each ball in $\calB$.

\paragraph{Definition of partition matroid $\calM_\calB$.} 
We define the partition matroid that captures the constraint that we wish to open a center in each ball in $\calB$. 
Let $\facilities_\calB$ be the  (multi) subset of facility/center locations containing  $B\cap \facilities$ for each ball  $B\in \calB$. If a center $c$ is in multiple balls in $\calB$, then $\facilities_\calB$ contains one distinct copy of $c$ for each ball. We let $\facilities_{\calB}(B) \subseteq F_\calB$ denote the facilities associated with $B\in \calB$. These sets satisfy the following two properties:
\begin{itemize}
    \item The sets $\facilities_{\calB}(B)$ partition $\facilities_\calB$.
    \item For $B \in \calB$, $\facilities_{\calB}(B)$ contains (a copy of) every facility in $\facilities\cap B$.
\end{itemize}
The first property holds since we took a unique copy of each center for each ball, and they are thus disjoint: $\facilities_\calB(B) \cap \facilities_\calB(B') = \emptyset$ for distinct $B,B' \in \calB$. Indeed, while it is not better for the cost to open multiple copies of a center, we make the copies to ensure the above two properties. This allows us to define the partition matroid $\calM_\calB = (\facilities_\calB, \calI)$ where
\[
    \calI = \{X \subseteq \facilities_\calB : |X\cap F_\calB(B)| \leq 1 \mbox{ for every } B\in \calB\}\,. 
\]
Moreover, since there is exactly one ball in $\calB$ for each center in $\ho$, we have $\ho \in \calI$ (where we slightly abuse notation as we should take the copy of center $c^*\in \ho$ that belongs to its associated ball).

\paragraph{Core, concentrated and hit clusters.} Our goal is now to define a submodular function $f$ so that we obtain a good approximation to $\ho$ by maximizing $f$ over the matroid constraint $\calM_\calB$.  In particular, the domain of $f$ is every subset of $\facilities_\calB$. However, we need some additional steps before defining $f$. In particular, we introduce the concept of the core of a cluster $C_c$ and the notions of concentrated and hit clusters, which allow us to simplify the structure of centers in $\calR$. For a cluster $C_c$, we define the \emph{core} of $C_c$ as
\begin{gather*}
    \core_c = \left\{p \in C_c: d(p, c) \leq \eps \cdot \frac{\clcost(S)}{|\calR|\cdot |C_c|}\right\}\,.
\end{gather*}
We further say that cluster $C_c$ is \emph{concentrated} if 
\begin{gather*}
    |\core_c| \geq (1-\eps) |C_c|
\end{gather*}
and it is \emph{hit} by a set $X$ of centers if there is a point $p \in \core_c$ such that
\begin{gather*}
    d(p, X) < d(p,c)\,.
\end{gather*}
For shortness we will sometimes say that a center is concentrated (resp., hit), if the corresponding cluster is so.

\paragraph{Guessing the centers of $\calR$ that are not concentrated.} We further simplify the task of finding the set of centers $\calR$  by guessing the centers of $\calR$ that are not concentrated. Specifically, partition the set $\calR$ into $\calR_0$ and $\calR_1$, where $\calR_1$ contain those centers of $\calR$ that are concentrated and $\calR_0$ contains those that are not. We can correctly guess $\calR_0$ in polynomial time since
the following simple claim shows that it is a subset of the centers whose cluster costs at least $\eps^2 \cdot \sopt/|\calR|$, of which there are only $O(|\calR|/\eps^2)$ many, 
and so guessing $\calR_0$ can thus be done in time
$2^{O(|\calR|/\eps^2)}$.
\begin{claim}
    Suppose that $C_c$ is not concentrated. Then the cost of $C_c$ is at least $\eps^2 \cdot \sopt/|\calR|$.
\end{claim}
\begin{proof}
    Since $C_c$ is not concentrated we have $|C_c - \core_c| \geq \eps |C_c|$. Moreover, each point $p\in C_c - \core_c$ has $d(p,c) > \eps \frac{\clcost(S)}{|\calR| \cdot |C_c|}$ by definition. Hence, as $\clcost(S) \geq \sopt$,
    \[
        \sum_{p\in C_c} d(p,c) \geq \sum_{p\in C_c - \core_c} d(p,c) \geq |C_c- \core_c|\cdot \eps \frac{\sopt}{|\calR| \cdot |C_c|} \geq \eps^2 \frac{\sopt}{|\calR|}\,.
    \]
\end{proof}

Now let $\calP$ be the potential centers that can be in $\calR$. Specifically, we let $\calP$ be the set that contains a center $c$ in $S -\calQ - \bcalU$ if $C_c$ equals the set of clients in $S_{\calQ}(c)$, where we recall that $S_{\calQ}(c)$ is the set of clients that are closest to $c$ in the clustering $S_{\calQ}$. Notice that the algorithm has all the information $S, \calB, \calQ, \bcalU$ and $\tmu$ to calculate $\calP$. Furthermore, by  Lemma~\ref{lem:successcheapremove}, we have that no client from $\bcalU$ was reassigned  by $\tmu$ to a center in $\calR$. So for $c\in \mathcal{R}$ we have $C_c = S_{\calQ}(c)$.  Hence, $\cal R \subseteq \calP$ and we can guess $\calR_0$ as follows:
\begin{mdframed}[hidealllines=true, backgroundcolor=gray!15]
\vspace{-5mm}
\paragraph{Guessing $\calR_0$}\ \\
\begin{enumerate}
    \item Let $\calC$ be the clusters of $\calP$ whose cost is at least $\eps^2 \cdot \sopt/|\calR|$.
    \item Output each subset of $\calC$.
\end{enumerate}
\end{mdframed}
By the above claim and the definition of $\calP$, one of the outputs is $\calR_0$.  Furthermore, the above guessing procedure outputs a family of polynomial many subsets. Indeed, we have $|\calR| \leq \log(n)/\eps^3$ by Lemma~\ref{lem:numnonpure}. Moreover, the total cost of the clusters in $\calP$ is $O(\sopt)$. To see this notice that the clusters corresponding to $\calP$ is a subset of the clusters of $S_{\calQ}$ and the cost of $S_{\calQ}$ is at most $\clcost(S - S_0  \cup \Lambda)$ (since $S_{\calQ}\supseteq S- S_0 \cup \Lambda$), which has cost at most $O(\sopt)$ by Lemma~\ref{lemma:S0properties}.  This implies that $|\calC| = O(\log(n)/\eps^5)$, so the total number of subsets is $n^{1/\eps^{O(1)}}$. The algorithm proceeds by trying all possible subsets in the output, and we analyze the algorithm when it takes the correct guess of $\calR_0$.

\paragraph{Definition of the submodular function $f$. }
We first define another function $g$ on the same domain as $f$, i.e., on all subsets $X$ of $F_\calB$. We will then define $f$  by $f(X) = g(\emptyset) - g(X)$. Let the \emph{closed cost} of a cluster $C_c$ be defined as
\[
    \closedclcost_c(X) := 
    \begin{cases}
        \sum_{p\in C_c} d(p, \{c\} \cup X \cup \dummyset) & \mbox{if $C_c$ is hit by $X$,} \\
        \min_{c'\in X \cup \dummyset} \sum_{p\in \core_c} d(p, c') + \sum_{p\in C_c - \core_c} d(p, \{c\} \cup X \cup \dummyset) & \mbox{otherwise.} 
    \end{cases}
\]
Further, let $\calP_1$ be the potential centers for $\calR_1$: it contains each center $c\in \calP - \calR_0$ so that $c$ is concentrated, i.e., $|\core_c| \geq (1-\eps) |C_c|$. We remark that the algorithm can calculate this set $\calP_1$ as it only depends on $\calP$, the guessed set $\calR_0$, the value $\clcost(S)$, and the clusters $C_c$ defined by $\tmu$. Moreover, by definition, we have $\calR_1 \subseteq \calP_1$.
For a subset $X \subseteq \facilities_\calB$, we then then define $g(X)$ to be the minimum value of 
\begin{gather*}
     \sum_{c \not \in \calR_0 \cup \calR'_1}\sum_{p\in C_c}  d\left(p, \{c\} \cup X \cup \dummyset\right) + \sum_{c\in \calR_0} \sum_{p\in C_c} d(p, X \cup \dummyset) + \sum_{c\in \calR'_1} \closedclcost_c(X) 
\end{gather*}
over all subsets $\calR'_1 \subseteq \calP_1$ with $|\calR_1'| = |\calR_1| = |\calR| - |\calR_0|$. 
In words, over the best $\calR_1'$, $g(X)$ is the cost of the solution obtained by removing the centers $\calR_0 \cup \calR_1'$  and assigning clients as follows:
\begin{itemize}
    \item If $p\in C_c$ for a remaining center $c \not\in \calR_0 \cup  \calR_1'$ or $p \not \in \core_c$ with $c\in \calR'_1$, $p$ is assigned to its closest center in $\{c\} \cup X \cup \dummyset$.
    \item If $p \in C_c$ for a removed center $c\in \calR_0$, $p$ is assigned to its closest center in $X \cup \dummyset$.
    \item If $p \in \core_c$ for a removed center $c\in \calR'_1$, $p$ is assigned to its closest center in $\{c\} \cup X \cup \dummyset$ if  $X$ hits $C_c$, and otherwise all clients in $\core_c$ are assigned to the same center $c'\in X \cup \dummyset$ that minimizes the cost.
\end{itemize}

We remark that this assignment is infeasible in the sense that it may assign clients to removed centers in $\calR'_1$. Nevertheless, we relate the values of $g(\emptyset)$ and $g(\ho)$ to $\clcost(M'_D, \mu'_D)$ and $\clcost(M'_O, \mu'_O)$, respectively; and, we show that given an $X \subseteq \facilities_\calB$ that is independent in $M_\calB$, i.e., $X\in \calI$, we can in polynomial-time output $k$ centers whose cost is at most $(1+2\eps)g(X) + 15\eps \sopt$.
Finally, the definition of $g$ allows us to prove that $f$ (defined by $f(X) = g(\emptyset) - g(X)$) is a monotone submodular function.
\begin{lemma}
We have that $g$ and $f$  satisfy the following properties:
\begin{enumerate}
\item We can evaluate $g(X)$ in polynomial time for every $X \subseteq \facilities_\calB$, and we can thus evaluate $f(X)$ in polynomial time. 
\item $f$ is a non-negative monotone submodular function.
\item The value $g(\emptyset)$ is at most $\clcost(M'_D, \mu'_D)$.
\item The value $g(\ho)$ is at most $\clcost(M'_O, \mu'_O)$.
\item  Given an $X \subseteq \facilities_\calB$ that is independent in $M_\calB$, i.e., $X\in \calI$, we can in polynomial-time output $k$ centers $S^*$ whose associated cost is at most $(1+2\eps) g(X) + 15\eps \sopt$.
\end{enumerate}
\label{lemma:propertiesfandg}
\end{lemma}
We give the proof of the lemma in the next subsection. We explain here how it implies Lemma~\ref{lemma:findingSstar}.

\begin{proof}[Proof of Lemma~\ref{lemma:findingSstar}]

By the first property of Lemma~\ref{lemma:propertiesfandg}, we can evaluate $f$ in polynomial time. Moreover, it is easy to see that we can answer independence queries in polynomial time for any partition matroid, particularly for $\calM_\calB$. We can thus apply Theorem~\ref{thm:submodularmatroidoptimization} on $f$ and $\calM_\calB$  with $\zeta =  (1-1/e - 1/2)$
to find a solution $X \in \calI$ such that\footnote{We remark that we could select $\zeta$ to be arbitrarily small and get a better guarantee in the statement of Lemma~\ref{lemma:findingSstar}. We have chosen this value to simplify the calculations as improving the guarantee of Lemma~\ref{lemma:findingSstar} does not improve the overall result. }
\begin{gather*}
    g(\emptyset) - g(X) =  f(X) \geq (1-1/e - \zeta)f(\ho) = (1-1/e - \zeta)( g(\emptyset) - g(\ho)) = \frac{1}{2} ( g(\emptyset) - g(\ho))\,,
\end{gather*}
where we used that $\ho$ is one feasible solution, i.e., $\ho \in \mathcal{I}$.
This in turn implies that 
\[
\frac{1}{2} \left(g(\ho)  +  g(\emptyset) \right) \geq  g(X)\,.
\]

Using the upper bounds on $g(\emptyset)$ and $g(\ho)$ of Lemma \ref{lemma:propertiesfandg} we thus have found a set $X$ such that 
\begin{gather*}
 \frac{1}{2} (\clcost(M'_O, \mu'_O) +  \clcost(M'_D, \mu'_D)) \geq g(X)\,.
\end{gather*}
Now using the last property of Lemma \ref{lemma:propertiesfandg} we can output a solution $S^*$ whose cost is at most
\begin{gather*}
    \frac{(1+2\eps)}{2} \left(  \clcost(M'_O, \mu'_O) +  \clcost(M'_D, \mu'_D) \right) + O(\eps \sopt)\,.
\end{gather*}
as required. 

Finally, each of the above steps runs in polynomial time: the algorithm of Theorem~\ref{thm:submodularmatroidoptimization} is polynomial time, and the last property of Lemma~\ref{lemma:propertiesfandg} used to obtain $S^*$ is polynomial-time. 
Moreover, the number of guesses of $\calR_0$ is at most $n^{1/\eps^{O(1)}}$, as argued after the description of that procedure. So we can, in polynomial time, try all possibilities and, among all solutions found (one for each guess of $\calR_0$), return one that minimizes the cost and, in particular, has cost at most that of $S^*$ (which was analyzed assuming the guess of $\calR_0$ was correct). We thus have  a polynomial time algorithm that returns a solution that satisfies the guarantee of the lemma.
\end{proof}
\subsubsection{Proof of Lemma~\ref{lemma:propertiesfandg}}

\paragraph{Proof of Property 1.} Given $X \subseteq \facilities_{\calB}$, we argue that we can evaluate $g(X)$ in polynomial time. For a center $c \in \calP_1$, define 
\[
\increase_c(X) = \closedclcost_c(X) - \sum_{p\in C_c} d(p, \{c\} \cup X \cup \dummyset)\,.
\]
In other words, using that $d(p,c ) \leq d(p, X)$ for $p\in \core_c$ if $C_c$ is not hit by $X$,
\begin{align}
\increase_c(X) = 
    \begin{cases}
        0 & \mbox{if $C_c$ is hit by $X$,} \\
        \min_{c'\in X \cup \dummyset} \sum_{p\in \core_c} (d(p, c') - d(p,\{c\}  \cup \dummyset))  & \mbox{otherwise.} 
    \end{cases}
\label{eq:increase}
\end{align}
We remark that $\increase_c(X)\geq 0$ by definition. 
Let $\clients_0$ denote the clients belonging to clusters $C_c$ with $c\in \calR_0$.
With this notation, $g(X)$ is the minimum value
\begin{align}
    \sum_{p\in \clients - \clients_0} d(p, \{\tmu(p)\} \cup X \cup \dummyset) + \sum_{p\in \clients_0} d(p, X \cup \dummyset) +  \sum_{c \in \calR'_1} \increase_c(X)
    \label{eq:nicedefofg}
\end{align}
over all subsets $\calR'_1 \subseteq \calP_1$ with $|\calR'_1| = |\calR_1|$. We further have that the value of $\increase_c(X)$ for $c\in \calR'_1$ is independent of other centers in $\calR'_1$. We can thus obtain the best choice of $\calR'_1$ by evaluating $\increase_c(X)$ for every $c\in \calP_1$ and choose the $|\calR_1|$ ones of smallest value. Notice that the algorithm knows $|\calR_1|$ since it equals $|S_0| - |\calQ| - |\bcalU| - |\calR_0| =|\calB| - |\calQ| - |\bcalU| - |\calR_0|$.     After we have obtained $\calR'_1$ in polynomial time, we can evaluate $g(X)$ in polynomial time by simply calculating the sum.

\paragraph{Proof of Property 2.} In the proof of this property it will be convenient to use definition~\eqref{eq:nicedefofg} of $g$.
We start by verifying that $f$ is monotone. I.e., that $f(Y) \geq f(X)$ when $X \subseteq Y$, which is equivalent to verifying that $g(Y) \leq g(X)$.  Let $g(X)$ equal
\[
     \sum_{ p\in \clients- \clients_0} d(p, \{\tmu(p)\} \cup X \cup \dummyset) + \sum_{p\in \clients_0} d(p, X \cup \dummyset)+ \sum_{c \in \calR'_1} \increase_c(X)
\]
for some set $\calR'_1$. Then $g(Y)$ is at most
\[
     \sum_{p\in \clients- \clients_0} d(p, \{\tmu(p)\} \cup Y \cup \dummyset)+ \sum_{p\in \clients_0} d(p, Y \cup \dummyset) + \sum_{c \in \calR'_1} \increase_c(Y)
\]
As trivially $d(p, \{\tmu(p)\} \cup Y \cup \dummyset) \leq d(p, \{\tmu(p)\}\cup X \cup \dummyset)$, $d(p,  Y \cup \dummyset) \leq d(p,  X \cup \dummyset)$ and $\increase_c(Y) \leq \increase_c(X)$ (see~\eqref{eq:increase}),  we have  $g(Y) \leq g(X)$ as required. Moreover, as $f(\emptyset) = 0$ by definition, non-negativity follows.

We proceed to verify that $f$ is submodular, i.e., that for every $X \subseteq Y \subseteq \facilities_B$ and $c_0 \in \facilities_B - Y$, $f(X \cup \{c_0\}) - f(X) \geq f(Y \cup \{c_0\}) - f(Y)$, or equivalently
\begin{align}
    g(X \cup \{c_0\}) - g(X) \leq g(Y \cup \{c_0\}) - g(Y)\,.
    \label{eq:submodularineq}
\end{align}
 Let $\calR_1^X, \calR_1^Y$, and  ${\calR}_1^{Y+c_0}$ be subsets of $\calP_1$ of cardinality $|\calR_1|$ such that
\begin{align*}
g(X) &=  \sum_{p\in \clients- \clients_0} d(p, \{\tmu(p)\} \cup X \cup \dummyset) + \sum_{p\in \clients_0} d(p, X \cup \dummyset) + \sum_{c \in \calR^X_1} \increase_c(X)\\
g(Y) &=  \sum_{p\in \clients- \clients_0} d(p, \{\tmu(p)\} \cup Y \cup \dummyset)+ \sum_{p\in \clients_0} d(p, Y \cup \dummyset) + \sum_{c \in \calR^Y_1} \increase_c(Y)\\
g(Y \cup \{c_0\}) & = \sum_{p\in \clients- \clients_0} d(p, \{\tmu(p)\} \cup (Y \cup \{c_0\}) \cup \dummyset)+ \sum_{p\in \clients_0} d(p, (Y \cup \{c_0\}) \cup \dummyset) + \sum_{c \in {\calR}^{Y+c_0}_1} \increase_c(Y\cup \{c_0\})
\end{align*}
We shall define  ${\calR}^{X+c_0}_1$  to be a subset of $\calP_1$ of cardinality $|\calR_1|$ so that  the upper bound
\begin{align*}
g(X\cup \{c_0\}) \leq  \sum_{p\in \clients- \clients_0} d(p, \{\tmu(p)\} \cup (X\cup \{c_0\}) \cup \dummyset)+ \sum_{p\in \clients_0} d(p, (X \cup \{c_0\}) \cup \dummyset)  + \sum_{c \in {\calR}^{X+c_0}_1} \increase_c(X \cup \{c_0\})
\end{align*}
allows us to verify Inequality~\eqref{eq:submodularineq}. First notice that no matter the definition of ${\calR}^{X+c_0}_1$, we have that
\begin{align}
 \sum_{p\in \clients- \clients_0} (d(p, \{\tmu(p)\} \cup (X\cup \{c_0\}) \cup \dummyset) -  d(p, \{\tmu(p)\} \cup X \cup \dummyset) )
 +\sum_{p\in \clients_0} (d(p, (X\cup \{c_0\}) \cup \dummyset) -  d(p,  X \cup \dummyset) )
 \label{eq:boringsubmodular1}
\end{align}
is upper bounded by
\begin{align}
\sum_{p\in \clients- \clients_0} (d(p, \{\tmu(p)\} \cup (Y \cup \{c_0\}) \cup \dummyset) - 
 d(p, \{\tmu(p)\} \cup Y \cup \dummyset)) + \sum_{p\in \clients_0}(d(p, (Y \cup \{c_0\}) \cup \dummyset) - 
 d(p, Y \cup \dummyset))\,.
 \label{eq:boringsubmodular2}
\end{align}
Indeed, if we let $\clients' \subseteq \clients - \clients_0$ be the subset of clients for which $d(p, c_0) < d(p, \{\tmu(p)\} \cup Y \cup \dummyset)$ and $\clients'_0 \subseteq\clients_0$ be the subset of clients for which $d(p, c_0) < d(p,  Y \cup \dummyset)$, then~\eqref{eq:boringsubmodular2} equals
\[
    \sum_{p\in \clients'} (d(p, c_0) - d(p, \{\tmu(p)\} \cup Y \cup \dummyset))+ \sum_{p\in \clients'_0} (d(p, c_0) - d(p,  Y \cup \dummyset))
\]
and we have 
\[
    \eqref{eq:boringsubmodular1} \leq \sum_{p\in \clients'} (d(p, c_0) - d(p, \{\tmu(p)\} \cup X \cup \dummyset)) + \sum_{p\in \clients'_0} (d(p, c_0) - d(p,  X \cup \dummyset))\,.
\]
It follows that \eqref{eq:boringsubmodular1} is at most \eqref{eq:boringsubmodular2} because
\[
    \sum_{p\in \clients'} ( d(p, \{\tmu(p)\} \cup Y \cup \dummyset)  - d(p, \{\tmu(p)\} \cup X \cup \dummyset)) + \sum_{p\in \clients'_0} ( d(p,  Y \cup \dummyset)  - d(p,  X \cup \dummyset)) \leq 0\,.
\]
To prove that $f$ is submodular, it is thus sufficient to define ${\calR}^{X+c_0}_1$ so that 
\begin{align}
    \sum_{c\in {\calR}^{X+c_0}_1} \increase_c(X \cup \{c_0\}) - \sum_{c\in {\calR}^X_1} \increase_c(X)\leq \sum_{c\in {\calR}^{Y+c_0}_1} \increase_c(Y \cup \{c_0\}) - \sum_{c\in {\calR}^Y_1} \increase_c(Y)\,.
    \label{eq:mainsubmodular}
\end{align}
To this end, let $\Delta_{\text{new}}$ contain the centers in ${\calR}^{Y+c_0}_1 - {\calR}^Y_1$ and those centers  $c \in {\calR}^{Y+c_0}_1 \cap \calR^Y_1$ with $\increase_c(Y \cup \{c_0\}) < \increase_c(Y)$. Notice that we may assume that, for each $c\in \Delta_{\text{new}}$,
\[
\increase_c(Y \cup \{c_0\})= \increase_c(\{c_0\}) =  \begin{cases}
        0 & \mbox{if $C_c$ is hit by $c_0$,} \\
         \sum_{p\in \core_c} (d(p, c_0) - d(p,\{c\}  \cup \dummyset))  & \mbox{otherwise.} 
    \end{cases}
\]
because the addition of $c_0$  caused a decrease in $\increase(\cdot)$ for these clusters.   Further, let $\Delta_{\text{old}}^Y$ contain $\Delta_{\text{new}} \cap \calR^{Y}_1$ and the centers $\calR^{Y}_1 - \calR^{Y+c_0}_1$. So $\calR^{Y+c_0}_1 = (\calR^{Y}_1 - \Delta_{\text{old}}^Y) \cup \Delta_{\text{new}}$ and $|\Delta_{\text{new}}| = |\Delta_{\text{old}}^Y|$. We now define $\calR^{X+c_0}_1 = (\calR^{X}_1 - \Delta_{\text{old}}^X) \cup \Delta_{\text{new}}$ where $\Delta^X_{\text{old}}$ is obtained as follows:
\begin{itemize}
    \item Initialize $\Delta_{\text{old}}^X =\calR^X_1\cap \Delta_{\text{new}}$. 
    \item While $|\Delta_{\text{old}}^X|< |\Delta_{\text{new}}|$, add a  center according to the following priorities:
    \begin{itemize}
        \item If there is a center $c\in \calR^X_1 - \Delta_{\text{old}}^X$ so that $c\not \in \calR^{Y}_1$, add $c$ to $\Delta_{\text{old}}^X$.
        \item Else add a center $c\in \calR^X_1- \Delta_{\text{old}}^X$ with $c \in \Delta_{\text{old}}^Y$. 
    \end{itemize}
\end{itemize}
We claim that the above is well-defined, i.e., that if $|\Delta_{\text{old}}^X|< |\Delta_{\text{new}}|$ and there is no center 
$c\in \calR^X_1 - \Delta_{\text{old}}^X$ so that $c\not \in \calR^{Y}_1$, then there must be a center $c\in \calR^X_1- \Delta_{\text{old}}^X$ with $c \in \Delta_{\text{old}}^Y$. This is because $|\Delta_{\text{old}}^X| < |\Delta_{\text{new}}| = |\Delta_{\text{old}}^Y|$ and all remaining centers $\calR_1^X - \Delta_{\text{old}}^X$  are in $\calR^Y_1$ in this case, so at least one of them must be in $\Delta_{\text{old}}^Y$ (recall that $|\calR_1^X| = |\calR_1^Y|$).

With this notation we have 
\begin{align}
   \notag \sum_{c\in {\calR}^{Y+c_0}_1} \increase_c(Y \cup \{c_0\}) - \sum_{c\in {\calR}^Y_1} \increase_c(Y ) & =  \sum_{c\in \Delta_{\text{new}}} \increase_c(Y \cup \{c_0\}) - \sum_{c\in \Delta_{\text{old}}^Y}\increase_{c}(Y) \\
    &=  \sum_{c\in \Delta_{\text{new}}} \increase_c(\{c_0\}) - \sum_{c\in \Delta_{\text{old}}^Y}\increase_{c}(Y)\,. \label{eq:firstsubmodular}
\end{align}
Similarly,
\begin{align}
    \sum_{c\in {\calR}^{X+c_0}_1} \increase_c(X \cup \{c_0\}) - \sum_{c\in {\calR}^X_1} \increase_c(X ) 
    & \leq \sum_{c\in \Delta_{\text{new}}} \increase_c(X \cup \{c_0\}) - \sum_{c\in \Delta_{\text{old}}^X} \increase_{c}(X) \notag\\ 
    & = \sum_{c\in \Delta_{\text{new}}} \increase_c(\{c_0\}) - \sum_{c\in \Delta_{\text{old}}^X} \increase_{c}(X)\,, \label{eq:secondsubmoular}
\end{align}
where we used that $\increase_c(X \cup \{c_0\}) \leq \increase_c(X)$ for any center $c$.

To analyze this, let $\pi$ be a bijection from $\Delta_{\text{old}}^Y$ to $\Delta_{\text{old}}^X$ so that $\pi(c) =c$ for all $c\in \Delta_{\text{old}}^Y \cap \Delta_{\text{old}}^X$. 
\begin{claim}
    For every $c\in \Delta_{\text{old}}^Y$ we have
    \[
    \increase_c(Y) \leq \increase_{\pi(c)}(X)\,.
    \]
\end{claim}
\begin{proof}
    If $c\in \Delta_{\text{old}}^Y \cap \Delta_{\text{old}}^X$, the claim holds because $\increase_c(Y) \leq \increase_c(X)$ for every $c\in \calP_1$. Otherwise, we have $\pi(c) \not \in \calR^Y_1$. This is because $(\Delta_{\text{old}}^X - \Delta_{\text{new}}^Y) \cap \calR_1^Y = \emptyset$. Indeed, in the construction of $\Delta_{\text{old}}^X$, we initialized with $\calR^X_1\cap \Delta_{\text{new}}$ and, as $\Delta_{\text{old}}^Y$ contains $\Delta_{\text{new}} \cap \calR^Y_1$, all elements in  $\calR^X_1\cap \Delta_{\text{new}}$ are either in $\Delta_{\text{old}}^Y$ or not in $\calR^Y_1$. The property is then maintained by definition of the while-loop in the construction of $\Delta_{\text{old}}^X$. 

    It follows that $\calR^Y_1$ where we remove $c$ and add $\pi(c)$ is another subset of $\calP_1$ of cardinality $|\calR_1|$. Hence, by the selection of $\calR^Y_1$ (to be the collection of centers with smallest $\increase(\cdot)$) we have
    \[
        \increase_c(Y) \leq \increase_{\pi(c)}(Y) \leq \increase_{\pi(c)}(X)\,.
    \]
\end{proof}
By the above claim,
\begin{align*}
\sum_{c\in \Delta_{\text{old}}^Y}\increase_{c}(Y) - \sum_{c\in \Delta_{\text{old}}^X} \increase_{c}(X)  = \sum_{c\in \Delta_{\text{old}}^Y}(\increase_{c}(Y) - \increase_{\pi(c)}(X)) \leq 0
\end{align*}
and thus by the above arguments (see~\eqref{eq:firstsubmodular} and~\eqref{eq:secondsubmoular})
\[
    \sum_{c\in {\calR}^{X+c_0}_1} \increase_c(X \cup \{c_0\}) - \sum_{c\in {\calR}^X_1} \increase_c(X)\leq \sum_{c\in {\calR}^{Y+c_0}_1} \increase_c(Y \cup \{c_0\}) - \sum_{c\in {\calR}^Y_1} \increase_c(Y)\,.
\]
We have thus verified~\eqref{eq:mainsubmodular},  which concludes the proof that $f$ is submodular.

\paragraph{Proof of Property 3.} For every $c\in \calR_1$, we have
\begin{align*}
\closedclcost_c(\emptyset) & =\min_{c'\in \Lambda}\sum_{p\in C^{core}_c}d(p,c')+\sum_{p\in C_c-C^{core}_c}d(p,\{c\}\cup \Lambda)\\
& \leq \min_{c'\in \Lambda}(\sum_{p\in C^{core}_c}d(p,c')+\sum_{p\in C_c-C^{core}_c}d(p,c'))=\min_{c'\in \Lambda}\sum_{p\in C_c}d(p,c'). 
\end{align*}
Thus 
\begin{align*}
    g(\emptyset) &\leq 
     \sum_{c \not \in \calR}\sum_{p\in C_c}  d\left(p, \{c\}  \cup \dummyset\right) + \sum_{p\in \clients_0} d(p, \dummyset) + \sum_{c\in \calR_1} \closedclcost_c(\emptyset)  \\
     &\leq \sum_{c \not \in \calR}\sum_{p\in C_c}  d\left(p, \{c\}  \cup \dummyset\right)+ \sum_{p\in \clients_0} d(p, \dummyset) + \sum_{c\in \calR_1} \min_{c'\in \dummyset} \sum_{p\in C_c} d(p, c') \\
     & \leq \sum_{c \not \in \calR}\sum_{p\in C_c}  d\left(p, \{c\}  \cup \dummyset\right) + \sum_{c\in \calR} \min_{c'\in \dummyset} \sum_{p\in C_c} d(p, c') \overset{\eqref{eq:MprimeDbound}}{\leq}\clcost(M'_D, \mu'_D)\,.
\end{align*}

\paragraph{Proof of Property 4.} Consider any $c\in \calR_1$. If $C_c$ is hit by $\ho$ one has
\begin{align*}
\closedclcost_c(\ho) & = \sum_{p\in C_c}d(p,\{c\}\cup \ho\cup \Lambda)\leq \min_{c'\in \ho}\sum_{p\in C_c}d(p,c'). 
\end{align*}
Otherwise
\begin{align*}
\closedclcost_c(\ho) & =\min_{c'\in \ho\cup \Lambda}\sum_{p\in C^{core}_c}d(p,c')+\sum_{p\in C_c-C^{core}_c}d(p,\{c\}\cup \ho\cup \Lambda)\\
& \leq \min_{c'\in \ho}(\sum_{p\in C^{core}_c}d(p,c')+\sum_{p\in C_c-C^{core}_c}d(p,c'))=\min_{c'\in \ho}\sum_{p\in C_c}d(p,c'). 
\end{align*}
Hence in both cases $\closedclcost_c(\ho)\leq \min_{c'\in \ho}\sum_{p\in C_c}d(p,c')$. Thus
\begin{align*}
    g(\ho) &\leq 
     \sum_{c \not \in \calR}\sum_{p\in C_c}  d\left(p, \{c\} \cup \ho \cup \dummyset\right) + \sum_{p\in \clients_0}d(p, \ho \cup \dummyset) + \sum_{c\in \calR_1} \closedclcost_c(\ho)  \\
     &\leq \sum_{c \not \in \calR}\sum_{p\in C_c}  d\left(p, \{c\} \cup \ho \cup \dummyset\right)+ \sum_{p\in \clients_0} d(p, \ho \cup \dummyset) + \sum_{c\in \calR_1} \min_{c'\in \ho} \sum_{p\in C_c} d(p, c') \\
     & \leq \sum_{c \not \in \calR}\sum_{p\in C_c}  d\left(p, \{c\} \cup \ho \cup \dummyset\right) + \sum_{c\in \calR} \min_{c'\in \ho} \sum_{p\in C_c} d(p, c')\overset{\eqref{eq:MprimeObound}}{\leq}\,\clcost(M'_O, \mu'_O).
\end{align*}

\paragraph{Proof of Property 5.}

    Let  $X \subseteq \facilities_\calB$ be independent in $M_\calB$, i.e., $X\in \calI$. We give a polynomial-time algorithm that outputs $k$ centers $S^*$ whose cost is at most $(1+2\eps)g(X) +15\eps \sopt$. Let $\calR'_1$ be a subset of $\calP_1$ of cardinality $|\calR_1|$ so that
\begin{gather*}
   g(X) =   \sum_{c \not \in \calR_0 \cup \calR'_1}\sum_{p\in C_c}  d\left(p, \{c\} \cup X \cup \dummyset\right)+ \sum_{p \in \clients_0} d(p, X \cup \dummyset) + \sum_{c\in \calR'_1} \closedclcost_c(X) 
\end{gather*}
    By the arguments in the proof of the first property, we can calculate $\calR'_1$ in polynomial time. Now we obtain $S^*$ from $S_{\calQ \cup \bcalU}$ as follows:
    \begin{itemize}
    \item Initialize $S^* = S_{\calQ \cup \bcalU} = S \cup \dummyset - \calQ - \bcalU$.
    \item Remove centers $\calR_0$ and $\calR'_1$ from $S^*$ to obtain a set of cardinality $k$. 
    \item Finally, for each ball in $\calB$, if $X$ contains a center $c$ in that ball, replace the associated dummy center by $c$; otherwise, replace the dummy center with an arbitrary center in the ball.  
    \end{itemize}
    By the definition of the dummy centers $\dummyset$, the cost of $S^*$ is at most
    \[
     \sum_{c \not \in \calR_0 \cup \calR'_1}\sum_{p\in C_c}  d\left(p, \{c\} \cup X \cup \dummyset\right)+ \sum_{p \in \clients_0} d(p, X \cup \dummyset) + \sum_{c\in \calR'_1}  \sum_{p\in C_c} d(p, X \cup \dummyset)\,.
    \]
    To relate this to $g(X)$ we need to compare $\closedclcost_c(X)$ to $\sum_{p\in C_c} d(p, X\cup \dummyset)$ for every $c\in \calR'_1$.
    Recall that 
\[
    \closedclcost_c(X) = 
    \begin{cases}
        \sum_{p\in C_c} d(p, \{c\} \cup X \cup \dummyset) & \mbox{if $C_c$ is hit by $X$,} \\
        \min_{c'\in X \cup \dummyset} \sum_{p\in \core_c} d(p, c') + \sum_{p\in C_c - \core_c} d(p, \{c\} \cup X \cup \dummyset) & \mbox{otherwise.} 
    \end{cases}
\]
Consider a cluster $C_c$ with $c\in \calR'_1$ and partition $C_c- \core_c$  into sets $A$ and $B$ where $A$ contain those clients $p$ so that $d(p, \{c\} \cup X \cup \dummyset) =d(p, c)$ and $B$ contain the remaining clients $p$ with $d(p, \{c\} \cup X \cup \dummyset) = d(p, X \cup \dummyset)$. Suppose first that there is a client $p_0 \in \core_c$ so that $d(p_0, X \cup \dummyset) < d(p_0, c)$, i.e., $C_c$ is hit by $X \cup \dummyset$.
By the definition of $\core_c$, we have  $d(p_0, c) \leq \eps \cdot \frac{\clcost(S)}{|\calR| \cdot |C_c|}\leq 5 \eps \cdot \frac{\sopt}{|\calR| \cdot |C_c|}$ where we used that $\clcost(S) \leq 5 \cdot \sopt$ by Lemma~\ref{lemma:localsearchis5approximation}. So,  by the triangle inequality,
\begin{align*}
\sum_{p\in C_c} d(p, X \cup \dummyset)& \leq  \sum_{p\in C_c - B} (d(p, c) + d(c, p_0) + d(p_0, X\cup \dummyset)) + \sum_{p\in B} d(p, X \cup \dummyset)\\
&\leq \sum_{p\in A} d(p, c) + \sum_{p\in B} d(p, X \cup \dummyset)  + 15\eps \cdot \frac{\sopt}{|\calR|}\,,
\end{align*}
where we used $d(p_0, X\cup \dummyset) \leq d(c, p_0) \leq 5\eps \cdot \frac{\sopt}{|\calR| \cdot |C_c|}$ and $d(p, c) \leq 5\eps \cdot \frac{\sopt}{|\calR| \cdot |C_c|}$ for any $p\in \core_c$ for the last inequality. Moreover, by definition of $A$ and $B$, we have $d(p, c) = d(p,  \{c\} \cup X \cup \dummyset)$ for any $p\in A$ and $d(p, X \cup \dummyset) = d(p, \{c\} \cup X \cup \dummyset)$ for any $p\in B$. Hence, in this case,
\begin{align*}
    \sum_{p\in C_c} d(p, X \cup \dummyset) \leq \sum_{p\in C_c - \core_c} d(p, \{c\} \cup X \cup Y) + 15 \eps \cdot \frac{\sopt}{|\calR|} \leq   \closedclcost_c(X) + 15\eps \cdot \frac{\sopt}{|\calR|}\,.
\end{align*}
Let us now consider a cluster $C_c$ with $c\in \calR'_1$ that is not hit by $X \cup \dummyset$, particularly not by $X$. 
Then 
\begin{align*}
    \closedclcost_c(X) &= \min_{c'\in X \cup \dummyset} \sum_{p\in \core_c} d(p, c') + \sum_{p\in C_c - \core_c} d(p, \{c\} \cup X \cup \dummyset) \\
    & \geq \sum_{p\in \core_c \cup B} d(p, X \cup \dummyset) + \sum_{q\in A} d(q, c)
\end{align*}
We further have 
\begin{align*}
    \sum_{p\in C_c} d(p, X \cup \dummyset) & = \sum_{p\in \core_c \cup B} d(p, X \cup \dummyset) + \sum_{q\in A} d(q, X \cup \dummyset)\\
    & \leq \sum_{p\in \core_c \cup B} d(p, X \cup \dummyset) + \sum_{q\in A} \frac{1}{|\core_c|}\sum_{p\in \core_c} (d(q, c) + d(c, p) + d(p, X \cup \dummyset))
\end{align*}
by the triangle inequality. Now using that $d(c,p) \leq 5 \eps \cdot \frac{\sopt}{|\calR| \cdot |C_c|}$ we obtain the upper bound
\begin{align*}
     \sum_{p\in C_c} d(p, X \cup \dummyset) & \leq \left(1 + \frac{|A|}{|\core_c|}\right) \sum_{p\in \core_c\cup B} d(p, X \cup \dummyset) + \sum_{q\in A} d(q,c) + 5\eps \frac{\sopt}{|\calR|} \\
     &\leq \left(1 + \frac{\eps }{(1-\eps)}\right) \sum_{p\in \core_c \cup B} d(p, X \cup \dummyset) + \sum_{q\in A} d(q,c) + 5\eps \frac{\sopt}{|\calR|} \\
     & \leq (1+ 2\eps) \closedclcost_c(X) + 5\eps \frac{\sopt}{|\calR|}
\end{align*}
where for the second inequality we used that every cluster in $\calR_1' \subseteq \calP_1$ is concentrated. In more detail $|\core_c|\geq (1-\eps)|C_c|$, hence $|A|\leq |C_c-\core_c|\leq \eps|C_c|$.

Hence, for any $c\in \calR_1'$ we proved that $\sum_{p\in C_c} d(p,X \cup \dummyset) \leq (1+2\eps) \closedclcost_c(X) + 15 \eps \frac{\sopt}{|\calR|}$. We thus get that the cost of $S^*$ is at most

    \begin{align*}
    & \sum_{c \not \in \calR_0 \cup \calR'_1}\sum_{p\in C_c}  d\left(p, \{c\} \cup X \cup \dummyset\right)+ \sum_{p \in \clients_0} d(p, X \cup \dummyset) + \sum_{c\in \calR'_1}\left( (1+2\eps) \closedclcost_c(X) + 15 \eps \frac{\sopt}{|\calR|}\right) \\
    \leq &  (1+2\eps) \left(\sum_{c \not \in \calR_0 \cup \calR'_1}\sum_{p\in C_c}  d\left(p, \{c\} \cup X \cup \dummyset\right)+ \sum_{p \in \clients_0} d(p, X \cup \dummyset) + \sum_{c\in \calR'_1} \closedclcost_c(X) \right) + 15\eps \sopt
    \end{align*}
    which equals $(1+2\eps) g(X) + 15\eps \sopt$.

\subsection{Putting Everything Together 
-- Proof of Theorem~\ref{thm:mainadditivecenters}}
\label{sec:stable:everythingtogether}
We now turn to proving Theorem~\ref{thm:mainadditivecenters}. We first describe the algorithm we analyse.
\begin{mdframed}[hidealllines=true, backgroundcolor=gray!15]
\vspace{-5mm}
\paragraph{A $(2{+\eps})$-Approximation for $k$-Median on an $\eps/\log n$-
Stable Instance $k, \clients, \facilities, \dist$.}\ \\
\begin{algorithmic}[1]
\State $S\leftarrow localSearch()$ \algorithmiccomment{Section \ref{sec:localSearchAnalysis}}
\State $W\leftarrow D$-$Sample(S,s^*)$ \algorithmiccomment{Section \ref{sec:dsampleproc}}
\State $\calL_{\texttt{ball}} \leftarrow ballGuess(S,W)$ \algorithmiccomment{Section \ref{sec:ballguesses}}
\For{all $\calB\in \calL_{\texttt{ball}}$}
\State Compute the dummy centers $\Lambda$ according to $\calB$
\State $\calLexp \leftarrow epxRem(S,\Lambda)$ \algorithmiccomment{Section \ref{sec:removalofExpensive}}
\For{all $\calQ\in \calLexp$ and all valid values of $|\calU|$, $|\calR|$ and $|\calX|$}
\State $\calL_{cheap} \leftarrow \emptyset$
\State $cheapRem(\emptyset,\emptyset,\emptyset,\emptyset,\emptyset)$ \algorithmiccomment{Section \ref{sec:removalofCheap}}
\For{all $\bcalU\in \calL_{cheap}$ with the associated $\tmu$}
\State Approximately solve the submodular maximization problem subject to a partition matroid induced by $(S,W,\calB,\calQ,\bcalU,\tmu)$ and turn the obtained solution into a $k$-median solution $S_{S,W,\calB,\calQ,\bcalU,\tmu}$ \algorithmiccomment{Section \ref{sec:submodularopt}}
\EndFor
\EndFor
\EndFor
\State \textbf{return} the cheapest solution of type $S_{S,W,\calB,\calQ,\bcalU,\tmu}$
\end{algorithmic}
\end{mdframed}

\begin{proof}[Proof of Theorem \ref{thm:mainadditivecenters}]
Our approach consists of running the above algorithm
$\log n$ times and taking the minimum cost solution output.
In the remaining, we argue that the above algorithm yields a $2+O(\eps)$-approximation with probability at least $(1-\eps)(1-1/n)$. The claim then follows as the probability that all $\log n$ executions of the above algorithm would fail is at most $\left( 1- (1-\eps)(1-1/n)\right)^{\log n} < 1/n$, and so with high probability we output a $2+O(\eps)$-approximation. 

We first discuss the algorithm's running time.
The local search algorithm runs in polynomial time. 
The submodular optimization step also runs in polynomial time by Lemma~\ref{lemma:findingSstar}.
The output of $D$-$Sample()$ (which is clearly a polynomial time procedure) is 
a subset of size $s^*$. By Lemma \ref{lem:successguessprocess}, we construct $\calL_{\texttt{ball}}$ in polynomial time and $|\calL_{\texttt{ball}}| \leq n^{1/\eps^{O(1)}}$. Similarly, by Lemma \ref{lem:successguessprocess}, we construct $\calLexp$ in polynomial time and  $|\calLexp| \leq n^{1/\eps^{O(1)}}$. %
Finally, by Lemma \ref{lem:successcheapremove}, we construct in polynomial-time 
 $n^{1/\eps^{O(1)}}$ candidate pairs $\bcalU,\tmu$. It follows that we solve $n^{1/\eps^{O(1)}}$ submodular function optimization problems, each in polynomial time, and thus the total running time is polynomial.

We turn to proving the approximation guarantee.
We aim to show that there exists a 
$(W,\calB, \calQ, \bcalU, \tmu)$ that is successful with probability at least $(1-\eps)(1-1/n)$.
Assuming this, the $(2+O(\eps))$-approximation follows by combining
Claim~\ref{claim:Mprimebounds} and Lemma~\ref{lemma:findingSstar}.

The local search step  provides a 5 approximation $S$ that satisfies  Lemma~\ref{lem:purecost}. Our algorithm next applies
$D$-$Sample()$ and by Lemma~\ref{lem:probcost} 
finds a successful $W$ with probability at least
$1-\eps$. Condition on having a successful
$W$, we then have that Lemma~\ref{lem:cheapcost} holds, and so does Lemma~\ref{lemma:convenience_upper_bound} by combining Lemma~\ref{lem:purecost} and Lemma~\ref{lem:cheapcost}.

Next, Lemma~\ref{lem:ballguesses} implies that one 
set of balls $\calB$ is a valid set of balls. This
set of balls induces a set of dummy centers 
$\dummyset$. From there, by Lemma~\ref{lem:successguessprocess} we have that 
with probability at least $1-1/n$, the 
$expRem()$ procedure produces a set
of center $\mathcal{Q}$ such that     \begin{enumerate}
    \item $\calQ \subseteq S_0$, and
    \item The total cost in solution $S - Q \cup \dummyset$ of the clusters in $S_0 - \calQ$ is at most $\eps \cdot \sopt$.
    \end{enumerate}
and so we have obtained a successful 
$(W,\calB,\calQ)$ and Lemma~\ref{lemma:costboundofMOandMDwithassignments}
applies. Condition on the event that we obtain
a successful $(W, \calB, \calQ)$, then, applying Lemma~\ref{lem:successcheapremove} we have that 
the $cheapRem()$ procedure outputs a pair
$(\bcalU, \tmu)$ satisfying the properties of 
Lemma~\ref{lem:successcheapremove}.
We thus have a $(W, \calB, \calQ, \bcalU,\tmu)$ such that $W$ is successful, the set of balls $\calB$ from Section~\ref{sec:ballguesses} is valid,  $\calQ$ selected in Section~\ref{sec:removalofExpensive} satisfies the properties of Lemma~\ref{lem:successguessprocess}, and $\bcalU,\tmu$  selected in this section satisfies the properties of Lemma~\ref{lem:successcheapremove} and
so $(W, \calB, \calQ, \bcalU,\tmu)$ is successful. Finally, note that the only probabilistic steps were the sampling of $W$ and of $\calQ$, and both steps are successful with probability at least $(1-\eps) \cdot (1-1/n)$, as required. 
\end{proof}

\bibliographystyle{alpha}

\appendix

\section{Analysis of Robust \logadaptalg}
\label{sec:robust_analysis_full}
In this section, we prove Theorem~\ref{thm:robust}, showing that a solution $\mathcal{H}$ of $\eta$-valid sequences yields an almost $2$-LMP approximation.
The proof is almost identical to that for \logadaptalg presented in Section~\ref{sec:log_adaptivity_analysis}. There are only two minor changes:
\begin{enumerate}
    \item Definition~\ref{def:robust_openable} of $\eta$-openability allows a facility to be paid up to an $n \eta$-additive difference. This difference only appears in Lemma~\ref{lemma:approx_guarantee_robust} below, which is an analog of Lemma~\ref{lemma:approx_guarantee}.
          \begin{restatable}{lemma}{lemmaapproxguaranteerobust}
              We have $\sum_{j\in D} \alpha^*_j \geq (1-\delta) \sum_{j\in D} d(j, S^*) + \sum_{i\in \Sr^*} (\hat f -  n \eta)$.
              \label{lemma:approx_guarantee_robust}
          \end{restatable}

    \item Definiton~\ref{def:valid_sequence} of a $\eta$-valid sequence allows $h_t$ is $\eta$-openable with respect to a super set $S' \supseteq S \cup \{ h_1, \dots, h_{t-1} \}$. This difference only appears in (the last part of) the proof of Lemma~\ref{lemma:dual_feasibility_robust}, an analog of Lemma~\ref{lemma:dual_feasibility}.

          \begin{restatable}{lemma}{lemmadualfeasibilityrobust}
              For a facility $i$,
              $\
                  \sum_{j \in D} [\alpha^*_j - 2d(i,j)]^+ \leq \hat f.
              $
              \label{lemma:dual_feasibility_robust}
          \end{restatable}

\end{enumerate}
These two lemmas immediately imply Theorem~\ref{thm:robust}. The other supporting claims and lemmas (Claim~\ref{claim:dual-feasible-induction_robust} and Lemma~\ref{lemma:dual-feasible-inactive_robust}) remain identical.

In order to use mostly the same analysis, one syntactic difference between the setting of
Section~\ref{sec:log_adaptivity_analysis} and the current setting that we have to reconcile is that the former analyzes the outcome of \logadaptalg while the latter analyzes a solution $\mathcal{H} = \{ H_1, \dots, H_L \}$ of $\eta$-valid sequences.
To use the same terminology from Section~\ref{sec:log_adaptivity_analysis}, let us consider the {\em execution} of the solution $\mathcal{H}$ where for $p = 1, \dots, L$, we run the $p$-th phase according to the sequence $H_p$; in Stage 1, each facility in the sequence becomes open one by one with the corresponding $(\tau_j)_j$ values, and in Stage 2 at the end of the phase, $\theta \leftarrow \theta(1+\eps^2)$.

Then an {\em atomic step} can still be defined in the same way as in \logadaptalg; it corresponds to opening a facility in Stage 1 (and updating $\alpha, S, A$ accordingly) or increasing $\alpha$ values in Stage 2 simultaneously (and updating $S$ accordingly). By ``at any point of the algorithm'', we mean any point in the algorithm's execution that is not in the middle of an atomic step.
The set of open facilities $S$ can be partitioned to be the set of the open free facilities $\Sf$ and the set of the open regular facilities $\Sr$.

Throughout the analysis, we let $\alpha^*_j$ be the final $\alpha$-value of a client $j\in D$, and we let $S^*$ be the set of opened facilities. (Also let $\Sf^*$ and $\Sr^*$ be the set of finally open free and regular facilities respectively.)
We start by analyzing the approximation guarantee of the algorithm with respect to  $\alpha^*$, and we then prove that $\alpha^*/2$ is a feasible solution to the dual.

\subsection{Approximation Guarantee.}

We prove Lemma~\ref{lemma:approx_guarantee_robust} restated below. The only difference from Lemma~\ref{lemma:approx_guarantee}
is that opening a facility only requires $\hat f - n \eta$ instead of $\hat f$.
\lemmaapproxguaranteerobust*

\begin{proof}[Proof of Lemma~\ref{lemma:approx_guarantee_robust}]
    The proof is by induction on the algorithm that at any point of the algorithm
    \begin{gather}
        \sum_{j\in I} \alpha_j \geq   \sum_{j\in I} (1-\delta)d(j, S)  + \sum_{i\in \Sr} (\hat f - n\eta)\,.
        \label{eq:IH_approx_guarantee_robust}
    \end{gather}
    The equality is initially true since $A= D$ (thus $I = \emptyset$) and $S = \emptyset$. We now analyze each of the two stages separately.

    \paragraph{Stage 1.} Consider what happens when we open a facility $h$, i.e., add it to $S$.

    If $h$ is a free facility, the $\alpha$ values and $\Sr$ do not change, and the only change is some client $j$ with $\alpha_j \geq (1 - \delta)d(j, S)$ becoming inactive. The amount of the increase of the left-hand side of~\eqref{eq:IH_approx_guarantee_robust} is at least that of the right-hand side.

    If $h$ is a regular facility $i$, let $(\alpha, S, A, \theta)$ be the state right before opening $i$, $A'$ be the clients that are removed from $A$ by the opening, which means $A' = \{ j \in A : \tau_j \geq (1 - \delta)d(i, j) \}$. (Recall that the opening lets $\alpha_j \leftarrow \tau_j$ for $j \in A$.)
    Further, let $X := \{j \in I: d(j, i) < d(j, S) \}$ be the subset of clients of $I$ that make positive bids to $i$.
    The change of cost of the right-hand side of~\eqref{eq:IH_approx_guarantee_robust} is at most
    \[
        (\hat f - n\eta) +   \sum_{j\in A'} (1-\delta) d(i,j) + \sum_{j \in X}  (1 - \delta)(d(i,j) - d(j, S))\,.
    \]
    Since $i$ is paid up to $\eta$ (the third bullet of Definition~\ref{def:robust_openable}), we also have
    \begin{align*}
        \hat f - n \eta & \leq
        \sum_{j\in A'} ( \tau_j  - (1-\delta) d(i,j)    )
        +
        \sum_{j \in X}  (1-\delta)(d(j, S) -  d(i,j)) \,.
    \end{align*}
    We thus get that the change of cost of the right-hand-side is at most $\sum_{j\in A'} \tau_j$, which is the change of the left-hand-side.

    \paragraph{Stage 2.}
    No facility is open at this stage, and every client $j$ that becomes inactive at this stage has $\alpha_j = (1 - \delta)d(j ,S)$, which implies that the left-hand-side and right-hand-side both increase with the same amount when a client $j$ becomes inactive.
\end{proof}

\subsection{Dual Feasibility}
We prove Lemma~\ref{lemma:dual_feasibility_robust} for dual feasibility. The following claim is identical to
Claim~\ref{claim:dual-feasible-induction} without any modification.
\begin{claim}[No Over-Bidding]
    At any point of the algorithm, for every facility $i$,
    \begin{equation}
        \sum_{j \in A} [\alpha_j - d(i,j)]^+ + \sum_{j \in I} [(1 - \delta)d(j, S) - d(j, i)]^+ \leq \hat{f}.
        \label{eq:dual-feasible-induction_robust}
    \end{equation}
    \label{claim:dual-feasible-induction_robust}
\end{claim}
\begin{proof}
    In the beginning, it is true because every client is active and $\alpha_j = 1 \leq d(i, j)$ for every $j \in D$.
    Let us prove that each step of the algorithm maintains this invariant. We first consider the steps that happen during the first stage and then those during the second stage.

    \paragraph{Stage 1.}
    Suppose that facility $h'$ is open with the bids $(\tau_j)_{j \in A}$.
    It does (1) potentially increase the $\alpha$-values of clients in $B(h', \eps \theta) \cap A$, and (2) make clients in $A \cap B(h', \theta/(1 - \delta))$ inactive, which includes $B(h', \eps \theta) \cap A$ considered in (1).

    Consider any facility $i$ and see how~\eqref{eq:dual-feasible-induction_robust} is impacted by this change.
    First, observe that a client $j$ that was inactive before the opening of $h'$ does not increase its contribution to~\eqref{eq:dual-feasible-induction_robust}. This is because $d(j,S)$ monotonically decreases in $S$.
    We now turn our attention to the more interesting case.
    For a client $j$ becoming inactive from active, the contribution to~\eqref{eq:dual-feasible-induction_robust} changes from $[\alpha_j - d(i,j)]^+$ (where $\alpha$-values are measured right before the opening) to
    $[(1 - \delta)d(j, S) - d(j, i)]^+$ (where $S$ includes $h'$).

    If $\alpha_j$ does not change by the opening of $h'$, the fact that $j$ becomes inactive means $(1 - \delta)d(j, S) \leq \alpha_j$, so $j$'s contribution to~\eqref{eq:dual-feasible-induction_robust} cannot increase.
    If $\alpha_j$ does strictly increase by the opening of $h'$, it means that $d(j, S) \leq d(j, h') \leq \eps \theta$ after the opening, so the contribution to~\eqref{eq:dual-feasible-induction_robust} decreases by at least $\theta / 2$ or becomes $0$.

    We have thus proved that no client increases their contribution to~\eqref{eq:dual-feasible-induction_robust} during the steps of the first stage, and so these steps maintain the inequality.

    \paragraph{Stage 2.}
    Assume towards contradiction that at the end of the phase, if we increase $\alpha_j$ of every $j \in A$ to $\min((1 + \eps^{2})\theta, (1 - \delta)d(j, S))$, the claim is violated for a non-empty subset of facilities. Since any free facility is only farther away in the metric (depending on $u(\cdot)$) we thus must have that the claim is violated by a set of regular facilities. Let $F'$ be that set.
    We select a ``minimal'' such counter example in the following way: let $\tau' \leq (1+\eps^2) \theta$ be the smallest value so that $\alpha'_j := \min(\tau', (1 - \delta)d(j, S))$ satisfies
    \[
        \sum_{j \in A} [\alpha_j' - d(i,j)]^+ +  \sum_{j \in I} [(1 - \delta)d(j, S) - d(j, i)]^+ = \hat{f},
    \]
    for some $i\in F'$.
    We remark that $\tau' \geq \theta$ since~\eqref{eq:dual-feasible-induction_robust} is satisfied for all facilities after the first stage by the previous arguments.
    Furthermore,  $i\not \in S$ is not yet open since if it was open then there would be no client $j\in A$ for which $\alpha'_j-d(i,j)$ is strictly positive. Hence the increase in $\alpha$-values cannot cause $i$ to violate~\eqref{eq:dual-feasible-induction_robust} if it were already opened, which would contradict $i\in F'$.

    We now show that this $i$  with $\tau_j = \alpha'_j$ for $j \in A \cap B(i, \eps \theta)$ and $\tau_j = \alpha_j$ for $j \in A \setminus B(i, \eps \theta)$ satisfies the conditions of Definition~\ref{def:robust_openable} for $i$ with respect to $S$. In other words, $i$ is an $\eta$-openable facility with respect to $S$ (that was not yet opened), which contradicts the {\em maximality} in Definition~\ref{def:valid_sequence} for a $\eta$-valid sequence.

    We verify the conditions of Definition~\ref{def:robust_openable} one-by-one. The first two bullets are satisfied by the definition of the bids $(\tau_j)_{j\in A}$.
    For the third bullet,
    \begin{align*}
                   & \sum_{j \in A}  [\tau_j - (1 - \delta)d(i,j)]^+
        + (1 - \delta)\sum_{j \in I} [d(j, S) - d(i, j)]^+                                       \\
        = \quad    & \sum_{j \in A \cap B(i, \eps \theta)}  [\alpha'_j - (1 - \delta)d(i,j)]^+ +
        \sum_{j \in A \setminus B(i, \eps \theta)}  [\alpha_j - (1 - \delta)d(i,j)]^+
        + (1 - \delta)\sum_{j \in I} [d(j, S) - d(i, j)]^+                                       \\
        \geq \quad & \sum_{j \in A \cap B(i, \eps \theta)}  [\alpha'_j - d(i,j)]^+ +
        \sum_{j \in A \setminus B(i, \eps \theta)}  [\alpha'_j - d(i,j)]^+
        + \sum_{j \in I} [(1 - \delta)d(j, S) - d(i, j)]^+\,.
    \end{align*}
    For the second summation we used that $\alpha_j - (1 - \delta)d(i,j) \geq (1 + \eps^{2})\alpha_j - d(i, j) \geq \alpha'_j - d(i,j) $ when $d(i, j) \geq \eps \theta$ (recall that $\alpha_j = \theta$ for all $j\in A$).

    It remains to verify the fourth bullet of Definition~\ref{def:robust_openable}. Suppose toward contradiction that there is a regular facility $i_0$ and  $k \in A$,
    \begin{gather*}
        \sum_{j\in A } [\tau_{j} - d(i_0, j)]^+ + \sum_{j\in I} [\tau_k  - d(i_0,k) - 2d(i_0, j)]^+ > \hat f.
    \end{gather*}
    We would like to show
    \begin{align*}
         & \sum_{j \in A} [\tau_j - d(i_0,j)]^+ +  \sum_{j \in I} [(1 - \delta)d(j, S) - d(i_0, j)]^+ \\ \geq \quad & \sum_{j\in A} [\tau_{j} - d(i_0, j)]^+ + \sum_{j\in I} [\tau_k  - d(i_0,k) - 2d(i_0, j)]^+
    \end{align*}
    which concludes the proof because then $i_0$ is violated with a smaller value of $\tau'$ which contradicts the minimal selection of $\tau'$.
    To prove the above inequality, note that $k \in A$ implies $\tau_k \leq (1 - \delta)d(k, S)$ and for every $j \in I$,
    \[
        (1 - \delta)d(k, S)
        \leq (1 - \delta)(d(k, i_0) + d(i_0, j) + d(j, S)),
    \]
    and so
    \[
        \tau_k - d(i_0, k) - 2d(i_0, j) \leq (1-\delta) d(j,S) - d(i_0,j)  \,.
    \]
\end{proof}

As usual, fix an arbitrary regular facility $i \in F$. The setup in this paragraph is also identical to Section~\ref{sec:log_adaptivity_analysis}.
Let
$D^* = \{ j \in D : \alpha^*_j > 2d(i, j) \}$ be those clients that contribute to the left-hand-side. Further, order the clients in $D^*$ according to $\alpha^*$ values by the time they are removed from $A$ (in the algorithm) and break ties according to $\alpha^*$ values (in increasing order). We will use $j, k$ to denote clients in $D^*$, and let us slightly abuse notation and use $j \leq k$ with respect to this ordering. In other words, we have $j\leq k$ if $j$ was removed from $A$ in an atomic step before $k$ was removed, $j$ and $k$ was removed during the same atomic step and $\alpha^*_j \leq \alpha^*_k$.
We also use
the terminology that a facility $i$ {\em freezes} $i_0$ if $i$ becomes open when $i_0$ is not, and $d(i, i_0) \leq (1+3\eps)\theta/2$. (Of course, $i$ freezes itself when it is open.)

The following lemma and the proof are also identical to Lemma~\ref{lemma:dual-feasible-inactive}.
\begin{lemma}
    For any $k \in D^*$ that becomes inactive strictly before $i$ becomes frozen,
    \begin{equation}
        \sum_{j \in D^* : j \geq k} [\alpha^*_k - d(i, j)]^+ +
        \sum_{j \in D^* : j < k} [\alpha^*_k - 2d(i, j) - d(i, k)]^+ \leq \hat f.
        \label{eq:dual-feasible-inactive_robust}
    \end{equation}
    \label{lemma:dual-feasible-inactive_robust}
\end{lemma}
\begin{proof}
    Let us do the following case analysis depending on the stage $k$ becomes inactive.

    \paragraph{Stage 1.}
    $k$ becomes inactive because of the opening of some $h_0$. By the statement of the lemma, it suffices to handle the case that $h_0$ does not freeze $i$. So $d(i, h_0) > (1 + 3\eps)\theta / 2$. The opening of $h_0$ might have strictly increased the $\alpha$-value of the clients in $B(h_0, \eps \theta)$, but no such client $j$ will be in $D^*$, since they immediately become inactive while
    \[
        d(i, j) \geq d(i, h_0) - d(h_0, j) > ((1 + 3\eps) / 2 - \eps) \theta \geq (1+\eps^2) \theta/2 \geq \alpha^*_j / 2.
    \]
    Applying this argument to every facility open in this phase (which did not freeze $i$), we can conclude that, right before $h_0$ is open, every $j \in D^*$ has $\alpha_j \leq \theta$ and all active ones have $\alpha_j = \theta$. (So, $\alpha^*_k = \theta$ as well since the opening of $h_0$ does not strictly increase $\alpha_k$.)
    Claim~\ref{claim:dual-feasible-induction_robust}, applied right before $h_0$ is open, ensures that
    \[
        \sum_{j \in A} [\alpha_j - d(i, j)]^+
        + \sum_{j \in I} [(1 - \delta)d(j, S) - d(i, j)]^+ \leq \hat f.
    \]
    This satisfies~\eqref{eq:dual-feasible-inactive_robust} as (1) every $j \in D^*$ with $j \geq k$ is in $j \in A$ at that point, and (2) for every $j < k$ with $j \in I$,
    \[
        \alpha^*_k < (1 - \delta)d(k, S)
        \leq (1 - \delta)d(j, S) + d(j, i) + d(i, k)
    \]
    implies
    \[
        \alpha^*_k - 2d(i, j) - d(i, k)
        \leq (1 - \delta)d(j, S) - d(i, j).
    \]
    (Note that this $S$ does not contain $h_0$ and $k$ was not connected yet.)

    \paragraph{Stage 2.}
    Suppose that $k$ becomes inactive by the increase of the $\alpha$-values at the end of a phase.
    Then Claim~\ref{claim:dual-feasible-induction_robust} ensures that, at the end of this phase (after $k$ becomes inactive),
    \begin{equation}
        \sum_{j \in A} [\alpha_j - d(i,j)]^+ + \sum_{j \in I} [(1 - \delta)d(j, S) - d(j, i)]^+ \leq \hat{f}.
        \label{eq:dual-feasible-induction-again_robust}
    \end{equation}

    Let us see that the left-hand-side of~\eqref{eq:dual-feasible-induction-again_robust} is at least that of~\eqref{eq:dual-feasible-inactive_robust} by comparing the contribution of $j \in D^*$ to both. If $j \in A$, its contribution to the former is definitely at least that to the latter, since any active $j$ has $\alpha_j = (1 + \eps^{2})\theta \geq  \alpha^*_k$.

    For $j \in I$,  we must have
    \[
        \alpha^*_k \leq (1-\delta) d(k,S)  \leq  (1-\delta) ( d(k,j) + d(j,S))\,.
    \]
    (since $\alpha_k$ is never increased above $(1-\delta) d(k,S)$).
    and so by rearranging and using the triangle inequality $(d(k,j) \leq d(k,i) + d(i,j)$) we get
    \[
        \alpha^*_k  - 2d(i,j) - d(i,k)  \leq (1-\delta) d(j,S) - d(i,j)
    \]
    which satisfies our goal when $j<k$. Finally, when $j\in I$ but $j \geq k$, it means that $j$ was removed at the same time as $k$ at the end of the phase. In that case $(1-\delta) d(j,S) = \alpha^*_j$ and, as we break ties with $\alpha^*$-values, $\alpha^*_j \geq \alpha^*_k$,  which again satisfies our goal.
\end{proof}

Equipped with the above lemma, we are ready to complete the proof of dual feasibility.

\begin{proof}[Proof of Lemma~\ref{lemma:dual_feasibility_robust}]

    Consider the case that $i$ is frozen by $h_0$ with $(\tau_j)_{j \in A}$.
    The following claim also remains unchanged from Section~\ref{sec:log_adaptivity_analysis}.

    \begin{claim*}
        We have $\alpha^*_j = \tau_j$ for every $j \in D^* \cap A$.
    \end{claim*}
    \begin{proof}[Proof of Claim]
        Consider $j \in A$.
        First, note that if $\tau_j \geq (1-\delta) d(j,h_0)$ then $j$ is removed from $A$ when $h_0$ is opened and so $\alpha^*_j = \tau_j$.

        In the other case, when $\alpha_j \leq \tau_j < (1-\delta) d(j,h_0)$ we show that $j \not \in D^*$.
        Consider the (future) time right after $j$ is removed from the active clients. The $\alpha$-value of $j$ then (which is equal to the final $\alpha^*_j$) cannot be strictly greater than $(1 - \delta)d(j, h_0)$; whether it is increased in Stage 1 (as $\alpha \leftarrow \tau$) or Stage 2, since $h_0 \in S$ (where $S$ is the set of open facilities right before $j$ becomes inactive), the algorithm ensures that it cannot be strictly more than $(1 - \delta)d(j, S) \leq (1 - \delta)d(j, h_0)$.
        By the triangle inequality
        \begin{gather*}
            d(j, i) \geq
            d(j, h_0) - d(h_0, i) =
            \left(1 - \frac{d(h_0, i)}{\theta}\frac{\theta}{d(j,h_0)}\right)d(j,h_0)\,,
        \end{gather*}
        where $\theta$ is the value when $h_0$ was open. As $h_0$ freezes $i$, we have $d(h_0, i) \leq (1+3\eps)\theta/2$ and, by the assumption of the case we have $\theta = \alpha_j \leq (1-\delta) d(j,h_0)$. Plugging in those bounds to the above inequality yields,
        \begin{gather*}
            d(j,i) \geq
            \left(1 - \frac{(1+3\eps)}{2}(1-\delta)\right)d(j,h_0)\,.
        \end{gather*}
        We rewrite the above inequality (multiplying both sides by two and using that $\delta = 3\eps$) to obtain
        \begin{gather*}
            2d(j,i) \geq
            \left(2 - {(1+3\eps)}(1-\delta)\right)d(j,h_0) = \left(1+ 9\eps^2 \right) d(j,h_0) \,.
        \end{gather*}
        In other words we have $\alpha^*_j \leq d(j,h_0) < 2d(j,i)$, so  $j$ cannot be in $D^*$.
    \end{proof}

    The following remaining part of the proof needs some changes from Section~\ref{sec:log_adaptivity_analysis}, since when $h_0$ is $\eta$-openable in the $\eta$-openable sequence, it might be with respect to a superset $S' \supseteq S$. (Here $S$ denotes the set of all opened facilities right before $h_0$.)
    Let $A'$ the subset of $A$ that is still active with respected to the superset $S'$ for opening $h_0$ and let $I' = A \setminus A'$.
    As the $\alpha$-values of clients in $A\cap D^*$ remain unchanged after the opening of $h_0$, the fourth bullet of Definition~\ref{def:robust_openable} ensures that for any $k \in A' \cap D^*$,
    \begin{equation}
        \sum_{j \in A' \cap D^*} [\alpha^*_j - d(i, j)]^+ +
        \sum_{j \in (I \cup I') \cap D^*} [\alpha^*_k - 2d(i, j) - d(k, i)]^+ \leq \hat f.
        \label{eq:openable-fourth-again-again-again_robust}
    \end{equation}

    Moreover, for every client $k\in D^* \cap I$ that became inactive strictly before $h_0$ was opened, Lemma~\ref{lemma:dual-feasible-inactive_robust} says

    \begin{equation}
        \sum_{j \in D^* : j \geq k} [\alpha^*_k - d(i, j)]^+ +
        \sum_{j \in D^* : j < k} [\alpha^*_k - 2d(i, j) - d(i, k)]^+ \leq \hat f.
        \label{eq:dual-feasible-inactive-again-again_robust}
    \end{equation}

    Finally for a client in $k\in D^* \cap I'$, note that $\alpha^*_k = \alpha_k = \theta$, because $k$ was not active with respect to $S'$. Therefore, by Claim~\ref{claim:dual-feasible-induction_robust} applied right before $h_0$ was opened
    \begin{equation*}
        \sum_{j \in A} [\alpha^*_k - d(i,j)]^+ + \sum_{j \in I} [(1 - \delta)d(j, S) - d(j, i)]^+ \leq \hat{f}.
    \end{equation*}
    which by the same argument as in the proof of Lemma~\ref{lemma:dual_feasibility},
    \begin{equation}
        \sum_{j \in A} [\alpha^*_k - d(i,j)]^+ + \sum_{j \in I} [\alpha^*_k- 2d(j, i) - d(i,k)]^+ \leq \hat{f}.
        \label{eq:dual-feasible-induction-again-again_robust}
    \end{equation}

    Let us consider the summation of~\eqref{eq:dual-feasible-inactive-again-again_robust} for every $k \in I \cap D^*$,~\eqref{eq:dual-feasible-induction-again-again_robust} for every $k\in I' \cap D^*$,
    and~\eqref{eq:openable-fourth-again-again-again_robust} for every $k \in A' \cap D^*$,
    and consider how many times each term appears.
    Let $a' = |A' \cap D^*|, i' = |I' \cap D^*|, d = |D^*|$.

    \begin{itemize}
        \item $-d(i, j')$ when $j' \in I\cap D^*$: Say $j'$ is the $\ell$th client in $D^*$ for some $\ell \leq d- i' - a'$. Then $-d(i, j')$ appears $d + (d - \ell)$ times in~\eqref{eq:openable-fourth-again-again-again_robust},
              \eqref{eq:dual-feasible-inactive-again-again_robust}, and~\eqref{eq:dual-feasible-induction-again-again_robust} as $-d(i,j)$, and $(\ell - 1)$ times as $-d(i, k)$ in
              \eqref{eq:dual-feasible-inactive-again-again_robust} when $k = j'$, so the total is $2d - 1$.

        \item $-d(i, j')$ when $j' \in I'\cap D^*$:
              It appears $d + a'$ times in~\eqref{eq:openable-fourth-again-again-again_robust},
              \eqref{eq:dual-feasible-inactive-again-again_robust}, and~\eqref{eq:dual-feasible-induction-again-again_robust} as $-d(i,j)$, and $(d - i' - a')$ times as $-d(i, k)$ in
              \eqref{eq:dual-feasible-induction-again-again_robust} when $k = j'$, so the total is $2d - i'$.

        \item $-d(i, j')$ when $j' \in A'\cap D^*$:
              It appears $d$ times in~\eqref{eq:openable-fourth-again-again-again_robust},
              \eqref{eq:dual-feasible-inactive-again-again_robust}, and~\eqref{eq:dual-feasible-induction-again-again_robust} as $-d(i,j)$, and $(d - a')$ times as $-d(i, k)$ in
              \eqref{eq:openable-fourth-again-again-again_robust} when $k = j'$, so the total is $2d - a'$.

        \item $\alpha^*_{j'}$ for $j' \in I\cap D^*$: $d$ times in~\eqref{eq:dual-feasible-inactive-again-again_robust} as $\alpha^*_k$ when $k = j'$.

        \item $\alpha^*_{j'}$ for $j' \in I' \cap D^*$: $d$ times in~\eqref{eq:dual-feasible-induction-again-again_robust} as $\alpha^*_k$ when $k = j'$.

        \item $\alpha^*_{j'}$ for $j' \in A'\cap D^*$: $(d - a')$ times in~\eqref{eq:openable-fourth-again-again-again_robust} as $\alpha^*_k$ when $k = j'$ and once as $\alpha^*_j$ in~\eqref{eq:openable-fourth-again-again-again_robust} for each $k \in A'$, so it is still $d$ times.

    \end{itemize}
    Therefore, $\sum_{j\in D^*} \left(\alpha^*_j - 2 d(i,j)\right) \leq \hat f$, which proves the lemma in the case the opening of a facility $h_0$ freezes $i$.
    To complete the proof, we note that the remaining case, when no facility freezes $i$, corresponds to the easier situation in the above analysis  when $D^*\cap A = \emptyset$, i.e.,  it corresponds to summing up~\eqref{eq:dual-feasible-inactive-again-again_robust} over all clients in $D^*$, and so $\sum_{j\in D^*} \left(\alpha^*_j - 2 d(i,j)\right) \leq \hat f$ also in that case.
\end{proof}

\section*{Acknowledgment}

The authors are grateful to Rares-Darius Buhai for valuable feedback on the writing. In particular, for identifying the need to explicitly write Lemma~\ref{lemma:payableimpliesopenable}.

\bibliography{literature}
\end{document}